%% file: electric.tex
\newif\ifarxiv\arxivtrue
\newif\ifsubm\submfalse
\newif\ifllncs\llncsfalse

\documentclass[a4paper,english,11pt]{scrartcl}
\input{header}

\usepackage[left=2.5cm, right=2.5cm, top=2cm, bottom=4cm]{geometry}

\usepackage[ngerman, textsize=small, textwidth=4.5cm, color=lightgray,disable]{todonotes}
\definecolor{darkgreen}{rgb}{0.2,0.8,0.55}
\definecolor{aliceblue}{rgb}{0.94, 0.97, 1.0}

\newcommand*\diff{\,\mathop{}\!\mathrm{d}} 
\DeclarePairedDelimiterX{\scalar}[2]{\langle}{\rangle}{#1, #2} 
\newcommand{\VI}{\textrm{VI}} 
\newcommand{\A}{\mathcal{A}} 
\usepackage{dsfont} 

\usetikzlibrary{calc,shapes,backgrounds,patterns}
\usepackage{pgfplots}
\usepackage{xspace}  
\usepackage{pdfpages} 

\tikzstyle{normalNode}=[draw, circle, fill=black, minimum size=1em]
\tikzstyle{terminalNode}=[draw, circle]
\tikzstyle{labeledNode}=[circle, draw, inner sep = 0.1em, minimum size=1em]
\tikzstyle{normalEdge}=[very thick, >=stealth]
\tikzstyle{holdoverEdge}=[normalEdge, dashed, black!60!white]
\tikzstyle{demandNode}=[inner sep=0em]
\tikzstyle{demandEdge}=[>=stealth, double, thick]

\tikzstyle{gadgetNode}=[draw, circle, fill=blue, scale=0.7]
\tikzstyle{gadgetEdge}=[>=stealth, thick]

\def\clap#1{\hbox to 0pt{\hss#1\hss}}

\usepackage{thm-restate}

\usepackage{placeins}

\title{Dynamic Traffic Assignment\\ for Electric Vehicles}
\author{Lukas Graf, Tobias Harks and Prashant Palkar}

\begin{document}
\maketitle

\begin{abstract}
We initiate the study of dynamic traffic assignment for electrical vehicles
addressing the specific challenges such as range limitations
and the possibility of battery recharge at predefined charging locations. 
We pose the dynamic equilibrium problem within the deterministic queueing model of Vickrey
and as our main result, we establish the existence of an energy-feasible dynamic equilibrium.
There are three key modeling-ingredients for obtaining this existence result: 
\begin{enumerate}
\item We introduce a \emph{walk-based} definition of dynamic traffic flows which allows for cyclic
routing behavior as a result of recharging events en route. 
\item We use  abstract convex feasibility sets in an appropriate function space
to model the energy-feasibility of used walks.
\item We introduce the concept of
\emph{capacitated dynamic equilibrium walk-flows} which generalize
the former unrestricted dynamic equilibrium path-flows. 
\end{enumerate}
Viewed in this framework, we show the existence of an energy-feasible dynamic equilibrium
by applying an infinite dimensional variational inequality, which in turn
requires a careful analysis of continuity properties of the network
loading as a result of injecting flow into walks.

We complement our
theoretical results by a computational study in which we design a fixed-point algorithm 
computing energy-feasible dynamic equilibria. We apply the algorithm to standard real-world
instances from the traffic assignment community illustrating the complex interplay of  resulting travel times, energy consumption
and prices paid at equilibrium.

\end{abstract}

\input{chapters/Introduction}

\input{chapters/Model}

\input{chapters/Existence}

\input{chapters/ComputationalStudy}

\input{chapters/Conclusions}

\clearpage
  \bibliographystyle{plain}
\bibliography{master-bib,comp}

\clearpage
\input{chapters/appendix}

\end{document}

%% file: header.tex
\usepackage{authblk}
\pdfoutput=1
\usepackage{tikz}
\usepackage[utf8]{inputenc}
\usepackage[T1]{fontenc}
\usepackage{amsthm}

\usepackage{mathtools,adjustbox}				
\usepackage{amssymb}				
\usepackage{mathrsfs, stmaryrd}		

\ifllncs

\fi

\usepackage{amsthm}			

\usetikzlibrary{chains,shapes.multipart}
\usetikzlibrary{shapes,calc}
\usetikzlibrary{automata,positioning}
\usetikzlibrary{arrows, decorations.markings}
\definecolor{myred}{RGB}{220,43,25}
\definecolor{mygreen}{RGB}{0,146,64}
\definecolor{myblue}{RGB}{0,143,224}

\usetikzlibrary{shapes}

\usepackage{pgfplots}

\tikzstyle{vertex} = [shape=circle,draw=black]
\tikzstyle{namedVertex} = [shape=circle,draw=black]
\tikzstyle{edge} = [draw,->,thick]
\tikzstyle{labeledNodeS}=[circle, color=black!75!white, draw, inner sep = 0.1em, minimum size = 1.5em, scale=1.25]
\tikzstyle{normalEdge}=[very thick, >=stealth]

\usepackage{listings}
\usepackage{color}
\definecolor{purplekeywords}{rgb}{0.5,0,0.33}
\definecolor{greencomments}{rgb}{0,0.5,0}
\definecolor{bluestrings}{rgb}{0,0,1}
\definecolor{types}{rgb}{0.17,0.57,0.68}
\usepackage[bookmarks=false,colorlinks=true, linkcolor=blue, urlcolor=blue, citecolor=blue, breaklinks=true]{hyperref}
\usepackage{cleveref}			

\crefname{cons}{constraint}{constraints}
\Crefname{cons}{Constraint}{Constraints}
\creflabelformat{cons}{(#2#1#3)}
\crefname{claim}{claim}{claims}
\Crefname{claim}{Claim}{Claims}

\theoremstyle{definition}
\ifllncs
	\spnewtheorem{defn}[definition]{Definition}{\bfseries}{}
\else
	\newtheorem{defn}{Definition}[section]
	\newtheorem{example}[defn]{Example}
\fi

\theoremstyle{plain}

\ifllncs\else
	\newtheorem{lemma}[defn]{Lemma}
	\newtheorem{theorem}[defn]{Theorem}
	
\fi
\ifllncs

\fi	
\newtheorem{claim}{Claim}
\newenvironment{proofClaim}[1][]{\ifthenelse{\equal{#1}{}}{\begin{proof}}{\begin{proof}[#1]}}{\end{proof}}

\theoremstyle{remark}
\ifllncs\else
	
\fi

\usepackage{array}					

\usepackage{braket}					
\usepackage{bm}						
\usepackage{nicefrac} 				

\allowdisplaybreaks					

\newcommand{\IN}{\mathbb{N}}	
\newcommand{\IR}{\mathbb{R}}
\newcommand{\R}{\mathbb{R}}	

\newcommand{\F}{\mathcal{F}}	

\newcommand{\Pc}{\mathcal{W}}

\newcommand{\abs}[1]{\left|#1\right|}

\newcommand{\noqed}{\let\qed\relax}

\usepackage{comment}			
\usepackage{textcomp}			

\usepackage{rotating}

\usepackage{setspace}

\usepackage[font={small,it}]{caption} 

\usepackage{transparent}
\usepackage{color}
\usepackage{graphicx}

\usepackage[all]{xy}
\usepackage{lmodern}
\usepackage{tabto}

\usepackage[babel]{csquotes}

\usepackage{framed,enumitem}

\newcommand{\edgeLoad}[1][]{%
	\ifthenelse{\equal{#1}{}}
	{F^{\Delta}}
	{F^{\Delta}_{#1}}
}

\newcommand{\termTime}[1][]{%
	\ifthenelse{\equal{#1}{}}
	{\Theta}
	{\Theta_{\mathrm{#1}}}
}

\newcommand{\network}{\mathcal{N}}

\usepackage{dsfont} 
\newcommand{\CharF}[1][]{%
	\ifthenelse{\equal{#1}{}}
	{\mathds{1}}
	{\mathds{1}_{#1}}
}

\newcommand{\edgesLeaving}[1]{\delta^+_{#1}}
\newcommand{\edgesEntering}[1]{\delta^-_{#1}}

\usepackage[algo2e,ruled,linesnumbered,vlined,lined,boxed]{algorithm2e}

\usepackage{siunitx} 
\newcommand{\norm}[1]{\left\Vert#1\right\Vert}

\newcommand{\algoref}[1]{Algorithm~\ref{algo:#1}}
\newcommand{\eqnref}[1]{Equation~(\ref{eq:#1})}
\newcommand{\figref}[1]{\Cref{fig:#1}}
\newcommand{\tabref}[1]{Table~\ref{tab:#1}}
\usepackage{array}
\newcolumntype{L}[1]{>{\raggedright\arraybackslash}m{#1}}
\newcolumntype{C}[1]{>{\centering\arraybackslash}m{#1}}
\newcolumntype{R}[1]{>{\raggedleft\arraybackslash}m{#1}}
\usepackage{adjustbox}
\usepackage{booktabs}

%% file: chapters/Introduction.tex

\section{Introduction}
Electric vehicles (EVs) are a great promise for the coming 
decades in order to allow for mobility but at the same
time take measures against the climate change by reducing the emissions of classical combustion
engines.
The wide-spread operation of EVs, however, is by far not fully
resolved as the battery technology  comes with several complications, some of which
 listed  below:
\begin{itemize}
	\item\label{enum:bi} The limited battery power implies a limited driving range of EVs resulting
	in  complex resource-constrained routing behavior taking the 
	feasibility of routes w.r.t. the power consumption into account (cf.~\cite{DesaulniersEIS16,StorandtF12}). 
	\item\label{enum:charge} Feasible routes may contain cycles if
	the possibility of recharging at predefined charging stations is included (see~\cite{BaumSWZ17,Strehler17,StorandtF12}). The necessity of multiple recharging operations 
	is especially relevant for longer trips such as long-haul trucking or for the use of EVs in urban logistics~\cite{DesaulniersEIS16}.	
	\item The recharging  strategy itself can be quite complex involving  \emph{mode choices}
 ranging from relatively low-power supply modes (22 kW) up to high-power supply modes (350 kW) (cf.~\cite{EU21}). Different modes may come with substantially different recharging times and prices (cf.~\cite{Stat21}).
	 \item  For a selected recharge mode, the \emph{duration} of the recharge determines both, the resulting battery state (and hence the subsequent reach of the vehicle), and the corresponding total recharge price and, thus,  adds a further strategic dimension.
\end{itemize}
While some of the above challenges
have been partly addressed within the ``battery-constrained routing'' community (cf.~\cite{BaumSWZ17,DesaulniersEIS16,froger22,Strehler17,SchneiderSG14,StorandtF12} and references therein), the majority of these works rely on a \emph{static} and mostly \emph{decoupled} view on traffic assignment: Each vehicle is routed independently (subject to battery related side constraints)
and the interaction of vehicles in terms of congestion effects with increased travel times
is not considered. Only a few works (such as~\cite{Xiong18,Wang16,Zheng17}) take congestion effects of
routing EVs into account, yet, still relying on a static routing model.

In a realistic traffic system, vehicles travel \emph{dynamically} through the network and the route choices of vehicles
are mutually dependent as the propagation of traffic flow  leads to congestion at bottlenecks and in turn determines the route
choices to avoid congestion. 
This complex and self-referential dependency has been under
scrutiny in the traffic assignment community for a long time and it is usually resolved by  \emph{dynamic traffic assignments (DTA)} under which -- roughly speaking -- at any point in time no driver can   opt  to a better route.  As a result,  the actual equilibrium travel times do depend on the collective route choices of all vehicles and, even more strikingly, the equilibrium routes determine the actual power consumption profile of an EV  leading to a complex coupled dynamic system.
Note that emergent congestion effects are even relevant for the pure recharging process of an EV, since  with the rapid growth rates of EVs compared to the relatively scarce recharging infrastructure, significant waiting times at recharging stations are anticipated (cf.~\cite{EU21}).

DTA models have been studied in the transportation science community for more than 50 years  with remarkable success in deriving a concise  mathematical theory of dynamic equilibrium distributions, yet there is no such theory  for DTA models addressing the specific characteristics of
EVs. 
Let us quote a recent survey article by Wang, Szeto, Han and Friesz~\cite{WANG2018} that mentions the lack
of DTA models for the operation of EVs:
``To our best knowledge, a DTA model with path distance constraints for electric vehicles remains undeveloped; so do the corresponding solution algorithms.''

This research gap might have good mathematical reasons: virtually all known existence results in the DTA literature
rely on the assumption that feasible paths must be \emph{acyclic}  in order to obtain a well defined  \emph{path-delay operator} mapping the path-inflows to the experienced travel time (cf.  
~\cite{CominettiCL15,CominettiCO17, Friesz93,Koch11,ZhuM00,MeunierW10,Ran96}).
As explained above, the range-limitations of EVs  requires recharging stops and thus leads to
cyclic routing behavior with path length restrictions requiring a new approach to establish equilibrium existence.

\subsection{Our Contribution}
In this paper, we study a dynamic traffic assignment problem that 
addresses the operation of electrical vehicles including their 
range-limitations caused by limited battery energy and necessary recharging stops.
Our contributions can be summarized as follows:

\begin{enumerate}
\item We propose a  DTA model tailored to the operation of EVs that combines
the Vickrey deterministic queueing model with graph-based gadgets modeling complex
recharging procedures such as mode choices (low to high power supply) and  recharge durations.
A combined routing and recharging strategy of an EV can be reduced to choosing a  walk
within this extended network. 
\item A feasible walk may contain cycles (due to several recharging stops)
and the set of feasible walks that respect the battery-constraints may be quite complex. After establishing
some fundamental properties of the resulting network loading when flow is injected into walks, we introduce abstract
convex, closed and bounded feasibility sets in an appropriate function space to describe the resulting feasible
dynamic walk-flows. The set of such feasible dynamic walk-flows are then used
to set up the formal definition of a capacitated dynamic equilibrium in which also 
the monetary effect of  prices charged at recharging stations is integrated in the utility function of agents.
\item With the formalism of the network loading and the notion of a dynamic capacitated equilibrium,
we then proceed to the key question of equilibrium existence.
We show that the walk-delay operator that maps the walk-inflows to resulting travel times
is sequentially weak-strong continuous on the convex feasibility space (which corresponds to \emph{weakly-continuous} as 
previously used by Zhu and Marcotte~\cite{ZhuM00} for paths under the strict FIFO-condition).
This  allows us to apply a variational inequality
formulation by Lions~\cite{Lions} to establish the existence of dynamic equilibria.
While the general variational inequality approach dates back to Friesz et al.~\cite{Friesz93},
our result generalizes previous works on side-constraint
dynamic equilibria (e.g. Zhong et al.~\cite{zhong11}), because we do not assume a priori compactness
of the underlying convex restriction set, nor strict FIFO as in~\cite{ZhuM00}. 
\item We finally develop
a fixed-point algorithm for the concrete computation  of  energy-feasible dynamic equilibrium
and apply the algorithm to several real-world instances from the literature.
To the best of our knowledge, this work is among the first to  compute dynamic traffic equilibria
for electric vehicles and it can serves as the basis for evaluating the interplay between
congestion, travel times and used energy in a dynamic traffic equilibrium.
\end{enumerate}
\subsection{Related Work}
Our work touches upon several streams of the literature including
the routing aspect of individual EVs and the traffic assignment problem in the static and dynamic
setting.
\subparagraph{Routing Models and Algorithms for EVs.}
The algorithmic problem of computing optimal routes for EVs taking the limited range of the battery into account has been addressed
by~Storandt~\cite{Storandt12}.
Baum et al.~\cite{BaumDPSWZ20,BaumDWZ20}
also considered several variants of constrained shortest path problems 
and take additionally speed variations of EVs into account.
Funke and Storandt~\cite{StorandtF12} studied routing problems taking
possible stops at charging locations into account.
Desaulniers et al.~\cite{DesaulniersEIS16} and Schneider et al.~\cite{SchneiderSG14}
considered EV routing problems with side-constraints such as time windows and
derive integer-programming methods for their solution. Froger et al.~\cite{froger22}
studied EV routing problems and explicitly model hard capacities at recharge stations.
They derive a centralized optimization model and apply an integer-programming method to solve it.
It is worth mentioning that Froger et al.~\cite{froger22} give an excellent survey on the state
of the art for  ``EV-routing problems'' and we refer to this paper for further references.
All these mentioned works, however, do not integrate strategic route choices with coupled congestion effects in their models.

Funke et al.~\cite{FunkeNS15,FunkeNS16}
further studied the location of charging stations.
They reduced the problem -- assuming a static decoupled routing system -- 
to the hitting set problem and derived solution algorithms.
Xiong et al.~\cite{Xiong18} and Zheng et al.~\cite{Zheng17} modeled the location of charging stations
as a bilevel optimization problem modeling the lower level as a static
discrete and continuous congestion game, respectively.
\subparagraph{Static Traffic Assignment.}
The classical mathematical approaches used in the transportation science literature to predict traffic distributions rely on static traffic assignment models based on aggregated static multi-commodity flow formulations some of which allow convex side-constraints  (cf.~\cite{Larsson99,Patriksson1994tap,Sheffi,Wang16,wardrop52}). While static models have seen several decades of development and practical use, they abstract away too many important details and, thus, become less attractive. For the modeling
of traffic assignment of EVs, static models seem not  appropriate as  the time aspect
is crucial  for modeling the battery behavior and the corresponding routing of an EV.

\subparagraph{Dynamic Traffic Assignment.}
Flows over time are an important and more realistic generalization of classical static
flows since they are capable to incorporate
the critical time component (cf.~\cite{Ford62} and~\cite{Vickrey69}
for one of the earliest papers considering a game-theoretic perspective for flows over time.) 
Since then flows over time have been a central topic
in the transportation science literature, cf.~\cite{Friesz93,Ran96,TRC11} for further references. 
Classical works in the area of DTA models such as Friesz et al.~\cite{Friesz93}
introduced a variational inequality approach to characterize dynamic 
equilibria. After the work of Friesz et al., several works further analyzed the existence
of dynamic equilibria. Zhu and Marcotte~\cite{ZhuM00} established the existence of dynamic equilibria under  a strict FIFO condition. This strict FIFO condition is not satisfied for the Vickrey queueing model
as shown by Cominetti et al.~\cite{CominettiCL15} and first existence results
for the Vickrey model were shown by Koch and Skutella~\cite{Koch11} assuming single-commodity
instances and piece-wise constant inflow rates.  Han, Friesz and Yao~\cite{Han2013} then showed
existence of dynamic equilibria for multiple commodities, general inflow-rate functions and allowing for a simultaneous route and departure choice.
They formulated the problem (following~\cite{Friesz93})
as an infinite dimensional variational inequality problem
and then used an existence theorem of Browder~\cite{BROWDER1968}.
Browder's theorem requires continuity of the path-delay operator 
(which was already shown by the same authors in Han et al.~\cite{Han2013a,Han2013b})
and compactness of the underlying path-flow space.
The latter compactness property is in general not fulfilled and Han et al.~\cite{Han2013}
used a discretization scheme (restricting to price-wise constant path inflows which are compact)
and showed that the resulting sequence of discretized dynamic equilibrium flows
converges to a dynamic equilibrium for the original model.
Cominetti et al.~\cite{CominettiCL15} also established an existence result
for dynamic equilibria within the general multi-commodity Vickrey model using also an infinite dimensional variational inequality
formulation. They, however, used a theorem by Br\'{e}zis~\cite{brezis1968},
which does not rely on compactness of the underlying  path inflow space
but requires a stronger form of continuity (sequential weak-strong continuity) of the path-delay operator.
Zhong et al.~\cite{zhong11} studied the existence of dynamic equilibria
for traffic models with side-constraints. They also used a variational inequality formulation
but assumed a strong FIFO condition to hold and further assumed that the restriction
set is compact.

For the computation of dynamic equilibria it is a priori not clear how to obtain convergence
of a discretization scheme for an arbitrary flow over time (disregarding
equilibrium properties) within the Vickrey model unless some restrictive assumptions
such as fixed paths are assumed, see Sering et al.~\cite{SeringKZ21}.
While a computational study by Ziemke et al.~\cite{Limit2020}   shows some positive results with regards to convergence, Otsubo and Rapoport~\cite{Otsubo08} describe
``significant discrepancies'' between the continuous and a discretized
solution for the Vickrey model.  To overcome the
discontinuity issue, Han et al.~\cite{Han2013a} reformulated
the model using a PDE formulation. They
obtained a discretized model whose limit points
correspond to dynamic equilibria of the continuous model.
The algorithm itself, however, is numerical in the sense that a precision
is specified and within that precision an approximate
equilibrium is computed.
The overall discretization approach mentioned above stands in line with a class of numerical algorithms based on fixed point iterations computing approximate equilibrium flows within a certain numerical precision, see Friesz and Han~\cite{Friesz19} for a recent survey.

Iryo and Smith~\cite{IRYO2017} and Cominetti et al.~\cite{CominettiCL15} independently showed uniqueness of equilibrium travel times in the Vickrey model.
Sering and Vargas-Koch~\cite{Sering2019} incorporated spillbacks in the Vickrey model. 
The long term behavior of dynamic equilibria with infinitely lasting constant inflow rate at a single source was studied by  Cominetti, Correa and Olver~\cite{CominettiCO17}.  They introduced the concept of a steady state and showed that dynamic equilibria always reach a stable state provided that the inflow rate is at most the capacity of a minimal $s$-$t$ cut.  Olver, Sering and Koch~\cite{OlverSK21}  further characterized steady state properties of dynamic Nash equilibria under less restrictive assumptions.
For the instantaneous dynamic equilibrium concept in the Vickrey model, Graf, Harks and Sering~\cite{GrafHS20} 
proved equilibrium existence results and Graf and Harks~\cite{GrafH22} gave the first  finite time algorithm for dynamic equilibrium computation.
For further recent work w.r.t. price of anarchy in the Vickrey model, see~\cite{BhaskarFA15,CorreaCO19,Osterwijk21}.

%% file: chapters/Model.tex

\section{The Model}\label{sec:Model}
We now introduce our model for electric vehicles
in which we combine the Vickrey deterministic queuing model
with graph-based extensions in order to model
the key characteristics of the battery recharging technology for electric vehicles.
The complex strategic decision of an EV
involves
\begin{enumerate}
\item the route choice possibly involving necessary recharging stops and cycles,
\item the mode choice of the battery-recharge (e.g., Level 1, 2, 3),
\item the actual duration of each battery-recharge en route, which determines the resulting
battery state and the recharge cost while also influencing the EV's travel time.
\end{enumerate}
We model this complex decision space by using several graph-based gadgets
inside the Vickrey network model leading to the \emph{battery-extended network}. This construction is, in a sense, a local version of time-extended networks. That is, instead of making copies of the whole network we only have to duplicate the recharging nodes themselves such that essentially each copy corresponds to some choice of recharging mode and duration at that node. 
This way, we can reduce the complex strategy choice of an EV to selecting
a \emph{feasible  walk} inside the battery-extended network.
In the following, we first start with the physical Vickrey flow model
and then discuss the battery-extended network.

\subsection{The Physical Vickrey Network Model}

The physical Vickrey network model is based on a finite directed graph $G'=(V',E')$ with positive rate capacities $\nu_e\in \IR_+$
and positive transit times $\tau_e\in \IR_+$ for every edge $e \in E'$. There is a finite set of commodities~$I=[n]\coloneqq\{1,\dots,n\}$, each with a commodity-specific source node $s_i \in V'$ and a commodity-specific sink node $t_i \in V'$.  The (infinitesimally small) agents of every commodity $i\in I$ 
each represent a vehicle (electric or combustion engine) and they enter the network according to a bounded and integrable network inflow rate function $u_i:\IR_{\geq 0}\rightarrow\IR_{\geq 0}$ with bounded support. We denote by $T \coloneqq \sup\set{\theta \in \IR_{\geq 0} | \exists i \in I: u_i(\theta) > 0}$ the last time a vehicle enters the network.
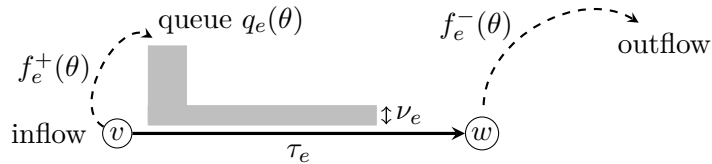
\begin{figure}[h!]
	\begin{center}
		\begin{adjustbox}{max width=0.7\textwidth}
			\input{tikz/QueueModel}			
		\end{adjustbox}
	\end{center}
	\vspace{-0.5cm}
	\caption[format=hang]{An edge $e=vw$ with a nonempty queue of size $q_e(\theta)$ at time $\theta$.
	The terms $f_e^+(\theta)$ and $f_e^+(\theta)$ denote the in- and outflow 
	rates at time $\theta$, respectively.}
	\label{fig:queue}
	
\end{figure}
If the total inflow into an edge  $e=vw\in E'$ exceeds the rate capacity
$\nu_e$, a queue builds up and agents  need to 
wait in the queue before they are forwarded along the edge. The total travel time
along $e$  is thus composed of the waiting time spent in the queue plus the physical transit time $\tau_e$.
A schematic illustration of the inflow and outflow mechanics of an edge $e$ is given in \Cref{fig:queue}.
The Vickrey model is one of the corner stone models in DTA and has been analyzed in the transportation science literature for decades, see 
Li, Huang and Yang~\cite{LI2020} for an up to date research overview of the past 50 years.

\subsection{The Battery-Extended Network}\label{subsec:charging}
For vehicles corresponding to a commodity $i\in I$, we assume that
they all have an equal initial battery state of level $b_i>0, i\in I$. 
This assumption is without
loss of generality as we can introduce copies of commodities
with the same source-sink pair but different initial battery states.
If an agent of commodity $i$ travels along an edge $e\in E$, this comes with a battery cost of $b_{i,e}\in \R$
which may be positive (energy consumption) or negative (recuperation). This battery cost is a fixed value for every commodity-edge pair $(i,e)$ and, in particular, independent of the actual flow on the edge. 
The maximum battery capacity is denoted by $b_i^{\max}$.
Note that the assumption that battery cost is independent of
congestion is well justified, since the engine of an EV  completely
turns off when a vehicle stands still leading to negligible power consumption while queueing up. Yet, the chosen route
does depend on the perceived travel time and, thus, also the realized power consumption does (indirectly) depend on congestion.

Recharging may either take place at home before the trip
actually starts (resulting in a high initial battery state $b_i$) or at public or commercial
charging stations, e.g. at public parking spots or at conventional gas stations. 
An important difference between recharging stations is the offered \emph{mode} and \emph{price} charged for recharging. 
Modes of recharge may range from relatively low power supply (up to  3.7 kilowatts (kW), Level 1) to medium supply (up to  22 kW, Level 2) up to high supply (operates at powers from 25 kW to more than 350 kW, Level 3) or even complete
battery swaps. Each mode may result in different recharging times
for a fixed targeted state of charging (SOC), and also  the resulting prices may  significantly vary not only among modes but also
among recharge locations. The statistics for 2021 for the recharging prices in Germany show for instance a significant price span for the ``cents per kWh tariff'' ranging from 35 Euro cents at public stations to 79 cents at private stations (cf.~\cite{Stat21}). 
Besides the recharge location and mode choice, the  planned duration for the recharge is an important
decision as it directly affects the journey time, the resulting SOC and the price paid.
For an agent with a high preference for fast travel times, it might payoff
to take a detour to some recharging location offering an expensive
Level 3 access resulting in a high battery SOC within a short time span.\footnote{For instance,  Tesla Model S, Renault Zoe, BMW i3 can be recharged at high voltage supplies to roughly  80\% battery capacity after a few minutes, whereas at houshold supplies, recharge to a comparable capacity
	takes several hours. For more on the mathematical modelling of precise charging functions as functions mapping recharge time (and current battery state) to resulting battery state, see~\cite{BaumDPSWZ20}.}
Summarizing, the selection of a recharging station, a recharge mode
and the duration of the actual recharge is an important strategic decision of EV drivers.

Given a tariff for recharging,\footnote{Pricing happens frequently on the basis of a per-minute tariff, other tariffs
charge on a per kWh basis or on a per-session basis, see~\cite{Stat21} for an overview on pricing schemes
in Germany.} 
we can model the set of possible combinations of recharge times, battery states and recharge prices
via tuples of the form $(\tau, b_i, p_i), i\in I$, where $\tau\in \IN$ is the time (in minutes)
spent for recharging, $b_i\equiv b_i(\tau)$ is the resulting increase of the battery level and $p_i\equiv p_i(\tau)\in \R_+$
is the charged price for a vehicle of commodity $i\in I$. 
Note that the functions $b_i(\tau), p_i(\tau)$ can be directly derived from the SOC  function and the resulting tariffs, respectively (cf. Xiao et al~\cite{Xiao2021}).

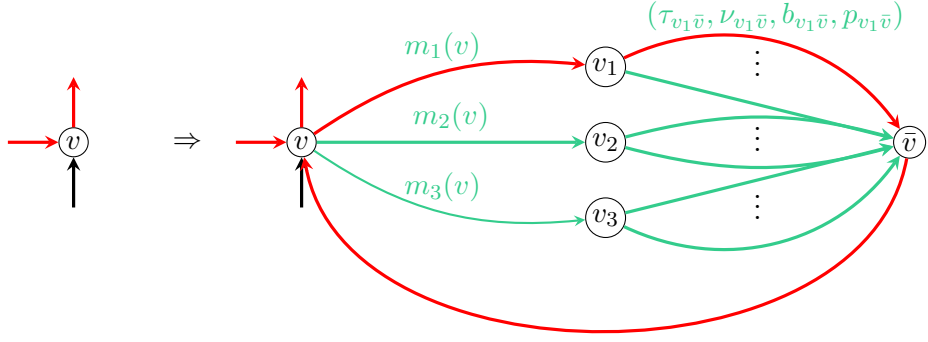
\begin{figure}
	\begin{center}
		\begin{tikzpicture}
			\begin{scope}[xshift=-3cm,yshift=-8mm]
				\node[labeledNode] (s) {$v$};
				\path (0, 1) node[]  (v1)  {}
				edge[<-, red,normalEdge]  (s);
				\path (0, -1) node[]  (v1)  {}
				edge[->, normalEdge]  (s);
				\path (-1, 0) node[]  (v1)  {}
				edge[->,red, normalEdge]  (s);
				\path (4, 0) node[] (v2) {};
				\path (1.5, 0.0) node {$\Rightarrow$};
			\end{scope}

			\begin{scope}[xshift=0,yshift=-8mm]
				\node[labeledNode] (s) {$v$};
				
				\path (0, 1) node[]  (v1)  {}
				edge[<-, red,normalEdge]  (s);
				\path (0, -1) node[]  (v1)  {}
				edge[->, normalEdge]  (s);
				\path (-1, 0) node[]  (v1)  {}
				edge[->, red, normalEdge]  (s);
				
				\path (4, 1) node[labeledNode]  (v1)  {$v_1$}
				edge[<-, normalEdge, red,bend right=20,label={above:$$}] node[above,darkgreen] {$m_1(v)$} (s);
				\path (4, 0) node[labeledNode] (v2) {$v_2$}
				edge[<-, normalEdge,label={above:$$},darkgreen] node[above] {$m_2(v)$} (s);
				\path  (4, -1) node[labeledNode] (v3) {$v_3$}
				edge[<-, gadgetEdge, bend left=20,label={above:$$},darkgreen] node[above] {$m_3(v)$} (s);
				\path  (8, 0) node[labeledNode] (t) {$\bar v$} 
				edge[<-, red,normalEdge, bend right=40] node[above,darkgreen] {$(\tau_{v_1\bar v},\nu_{v_1\bar v},b_{v_1\bar v},p_{v_1\bar v})$} (v1) 
				edge[<-, normalEdge, bend left=0,darkgreen] (v1)
				edge[<-, normalEdge, bend left=15,darkgreen] (v2)
				edge[<-, normalEdge, bend right=15,darkgreen] (v2)
				edge[<-, normalEdge, bend left=40,darkgreen] (v3)
				edge[<-, normalEdge, bend left=0,darkgreen] (v3)
				edge[->, red,normalEdge, bend left=80] (s) ;
				;
				\path (6, 1.15) node {$\vdots$};
				\path (6, 0.15) node {$\vdots$};
				\path (6, -0.75) node {$\vdots$};
				
			\end{scope}
		\end{tikzpicture}
	\end{center}
	\caption{Left: Initial vertex $v$ with an EV using a walk (red edges) without recharging. Right: Expansion of node $v$ using a graph-based gadget modeling the recharging options.  There are three recharging modes, say a low, medium or high power supply (Level 1, Level 2, Level 3) leading to the first three  edges $m_1(v),m_2(v), m_3(v)$. The subsequent parallel edges
		model the different charging times and resulting increase of the battery levels.
		The red edges describe one cycle inside the gadget and represent a recharge using mode $1$ 
		for time $\tau_{v_1\bar v}$ with resulting battery level increase of $|b_{v_1\bar v}|$
		at  price $p_{v_1\bar v}$. }\label{fig:gadget}
\end{figure}
Formally, 
recharging stations are identified with subsets of nodes of $V'$ denoted by $C_i\subseteq V', i\in I$, where we explicitly allow that $C_i$ depends on $i\in I$ to allow for different recharging technologies, that is, 
some vehicles may only
recharge at stations that have the required technology. By introducing copies
of commodities it is again without loss of generality to assume that every agent of commodity $i$ uses the same technology. For a recharging location $v\in C_i, i\in I$, we introduce a subgraph as depicted in \Cref{fig:gadget}.
The node $v\in C_i$ represents the original charging station viable for $i\in I$, 
the parallel edges leaving $v$ correspond to the different recharging modes available and
the subsequent edges model the different recharging times with corresponding recharge states
and prices.\footnote{For the sake of a simple illustration we allow
		parallel arcs but by introducing further dummy nodes subdividing an edge, one obtains a simple graph so that an edge can uniquely be represented by 
		a tuple  $vw$ for $v,w\in V$.} 
		At the end of this series-parallel graph-gadget, a backwards arc towards $v$ is introduced.	We associate with every edge
a tuple of the form  $(\tau_e,\nu_e, b_{i,e},p_{i,e})$, where
$\tau_e$ is the travel time (or recharge duration for a gadget edge), $\nu_e$ the
inflow capacity, $b_{i,e}$ the battery recharge and $p_{i,e}$ the price paid for the 
used recharge on edge $e$. Note that we have $p_{i,e}\equiv p_{i,e}(\tau_e)$ and $b_{i,e}\equiv b_{i,e}(\tau_e)$ for corresponding pricing and SOC functions, respectively.
 Any cycle in such a  gadget is in  one-to-one correspondence to
a mode $(e)$, recharge duration  $(\tau_e)$,  battery recharge $(b_{i,e})$ and price decision $(p_{i,e})$. If a mode is not compatible with the recharging technology used by EVs of type $i\in I$, we can set $b_{i,e}=+\infty$ to close the corresponding recharge edge for $i\in I$.
For every $i\in N$, we denote  the newly constructed vertices and edges by $V(C_i), E(C_i), i\in I$. 
\begin{defn}
The \emph{battery-extended network} is a tuple $\network=(G,\nu,\tau,b,p)$, where 
\begin{itemize}
\item $G=(V,E)$ is the  battery-extended graph with $V:=V'\cup_{i\in I}V(C_i)$
and $E:=E'\cup_{i\in I}E(C_i))$,
\item $\nu_e\in \R_+, e\in E$ denotes the inflow-capacities,
\item $\tau_e\in \R_+, e\in E$ denotes the travel times or recharge durations,
\item $b_{i,e}\in \R, i\in I,e\in E$ denotes the battery-consumption values,
\item $p_{i,e}\in \R_+, i\in I,e\in E$ denotes the recharge prices.
\end{itemize}
An $s_i$,$t_i$-walk in the battery-extended graph $G$  corresponds to a route choice in the original graph $G'$ together with recharging decisions
corresponding to cycles inside the gadgets, see Fig.~\ref{fig:gadget} for an example.
\end{defn}

\subsection{Feasible Walks in the Battery-Extended Network}
Assume that we are given the battery-extended network $\network$
 as described in Subsection~\ref{subsec:charging}.
 Let $W=(e_1,\dots, e_k)$ be a sequence of edges in the graph $G$.
We call $W$ a \emph{walk} if $e_j=v_{j-1}v_j$ for all $j\in [k]$ for $k\in \IN$. 
We assume that all walks considered in this paper are finite
and just use the term \emph{walk} to denote a finite walk. 
Note, that a walk is allowed to contain self-loops and/or nontrivial cycles as required for
a recharge operation.
We denote by $k_W \coloneqq k$ the length of $W$ and by $e^W_j$ the $j$-th edge of walk $W$.
$W$ is an $s_i$,$t_i$-\emph{walk}, if  $v_0=s_i$ and $v_k=t_i$. We denote by $\Pc_i$ the set of all $s_i$,$t_i$-walks and assume that this set is always non-empty, i.e. that every commodity has at least one walk from its source to its sink. Finally, we denote by $\Pc \coloneqq \set{(i,W) | i \in I, W \in \Pc_i}$ the set of all commodity-walk pairs.
The set $\Pc_{i}$ represent the set of strategies for a particle of commodity $i\in I$ and, thus, a complete strategy profile is a family of walk-flows for all commodities and all walks such that for every commodity the sum of its walk-flows matches its network inflow rate. We denote the set of all such strategy profiles by
\begin{align*}
	K \coloneqq \Set{h \in \big(L^2_{\geq 0}([0,T])\big)^{\Pc} | \forall i \in I: \sum_{W\in \Pc_{i}}h_{i}^W(\theta) = u_i(\theta) \text{ for almost all } \theta \in \IR_{\geq 0}},
\end{align*}
where $L^2_{\geq 0}([0,T])$ denotes the set of $L^2$-integrable non-negative functions
over the interval $[0,T]$
and any $h \in K$ is called a \emph{walk-flow}.
The crucial point when modeling electric vehicles is the energy-feasibility of a walk,
that is, the battery must not fully deplete while traversing a walk.
We capture this property in the following definition.
\begin{defn} \label{def:battery}
A walk $W=(e_1,\dots, e_k)\in \Pc_i$ is \emph{energy-feasible} for commodity $i\in I$, if the following condition is satisfied:
\begin{equation}\label{eq:battery-feasible}
	b_W(v_j)\in [0,b_i^{\max}] \text{ for all }j=1,\dots, k,
\end{equation}
where  $b_W(v_j)$ is defined inductively as
\begin{equation} 
	b_W(v_1)=b_i \text{ and  } b_W(v_{j+1})=\min\{b_W(v_j)-b_{i,e_{j+1}},b_i^{\max}\}.
\end{equation}
\end{defn}

\begin{figure}
	\begin{center}
		\begin{tikzpicture}[scale=1.0]
			
			\node[labeledNode] (v1) at (0, 0) {$s_i$};
			\draw (v1) ++ (2, 0) node[labeledNode] (v2) {$v_1$};
			\draw (v2) ++ (2, 0) node[labeledNode] (v3) {$v_2$};
			\draw (v3) ++ (0, -2) node[labeledNode] (v4) {$v_3$};
			\draw (v3) ++ (0, 2) node[labeledNode] (v5) {$t_i$};

			\draw (v4) ++ (0, -1) node[labeledNode] (v6) {$v_4$};
			\draw (v4) ++ (0, -2) node[labeledNode] (v7) {$v_5$};

			\draw[normalEdge,->] (v1) -- node[below] {\color{blue} $1$ } (v2);
			\draw[normalEdge,->] (v1) -- node[above] {\color{blue} $5$ } (v5);
			\draw[normalEdge,->] (v1) -- node[above] {\color{blue} $5$ } (v4);
			\draw[normalEdge,->] (v2) -- node[below] {\color{blue}$1$} (v3);
			\draw[normalEdge,->] (v3)  -- node[right] {\color{blue}$1$} (v4);
			\draw[normalEdge,->] (v3)  -- node[right] {\color{blue}$2$} (v5);
			\draw[normalEdge,->] (v4)  -- node[right] {\color{blue}$1$} (v2);
			\draw[normalEdge,->,darkgreen] (v4)  -- node[right] {\color{blue}$0$} (v6);
			
					\path (v6) edge [out=-45,in=+45,very thick, >=stealth,->,darkgreen] node[right] {\color{blue}$-4$} (v7);
			\path (v6) edge [out=-145,in=+145,very thick, >=stealth,->,darkgreen] node[left] {\color{blue}$-2$} (v7);
			
								\draw[->,normalEdge,darkgreen] (v7) .. controls (2,-4) and (2,-2) .. node[left] {\color{blue}$0$} (v4);
		\end{tikzpicture}
		\begin{tikzpicture}[scale=1.0]
			
			\node[labeledNode] (v1) at (0, 0) {$s_i$};
			\draw (v1) ++ (2, 0) node[labeledNode] (v2) {$v_1$};
			\draw (v2) ++ (2, 0) node[labeledNode] (v3) {$v_2$};
			\draw (v3) ++ (0, -2) node[labeledNode] (v4) {$v_3$};
			\draw (v3) ++ (0, 2) node[labeledNode] (v5) {$t_i$};

			\draw (v4) ++ (0, -1) node[labeledNode] (v6) {$v_4$};
			\draw (v4) ++ (0, -2) node[labeledNode] (v7) {$v_5$};

			\draw[normalEdge,->] (v1) -- node[above] {\color{blue} $5$ } (v5);
			\draw[normalEdge,->] (v1) -- node[above] {\color{blue} $5$ } (v4);
			\draw[normalEdge,red,->] (v1) -- node[below] {\color{blue} $1$ } (v2);
			\draw[normalEdge,red,->] (v2) -- node[below] {\color{blue}$1$} (v3);
			\draw[normalEdge,red,->] (v3)  -- node[right] {\color{blue}$1$} (v4);
			\draw[normalEdge,red,->] (v3)  -- node[right] {\color{blue}$2$} (v5);
			\draw[normalEdge,red,->] (v4)  -- node[right] {\color{blue}$1$} (v2);
			\draw[normalEdge,->,red] (v4)  -- node[right] {\color{blue}$0$} (v6);
			
			\path (v6) edge [out=-45,in=+45,very thick, >=stealth,->,red] node[right] {\color{blue}$-4$} (v7);
			\path (v6) edge [out=-145,in=+145,very thick, >=stealth,->,darkgreen] node[left] {\color{blue}$-2$} (v7);
			
								\draw[->,normalEdge,red] (v7) .. controls (2,-4) and (2,-2) .. node[left] {\color{blue}$0$} (v4);

		\end{tikzpicture}

	\end{center}
	\caption{Example of an instance with start node $s_i$ and sink  node $t_i$, $b_i=3$, $b_i^{\max}=4$.
	The green edges  represent the recharging gadget.
		The blue numbers at edges indicate the power consumption values $b_{i,e}$. 
		The shortest energy-feasible walk (assuming positive travel times) is 
		illustrated with red edges on the right which contains two  simple cycles
		$C_1:=\{v_3,v_4,v_5,v_3\}$ and $C_2:=\{v_1,v_2,v_3,v_1\}$, where the first cycle is contained in the recharging gadget
		and represents a mode and duration choice.}\label{ex:walk}
\end{figure}
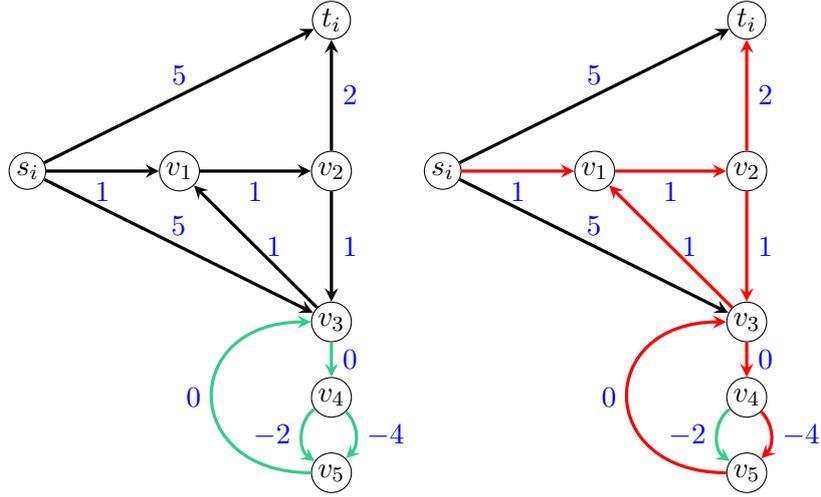

We assume that for every $i\in I$ there is at least one energy-feasible  walk and
that their collection is  denoted by
\[ \Pc_{i,b}:=\{W\in \Pc_{i}\vert W \text{ satisfies }\eqref{eq:battery-feasible}\}.\]

In \Cref{ex:walk}, we give an example illustrating that walking along cycles might indeed be necessary
to reach the sink. 
The set $\Pc_{i,b}$ represent the set of enery-feasible strategies for a particle of commodity $i\in I$. 
We further define $\Pc_{b} =\{(i,W)\vert i\in I, W\in \Pc_{i,b}\}$ to be the set
of commodity and battery-feasible walk pairs. Note that the set $\Pc_{b}$ need not be finite.

Now, a complete energy-feasible strategy profile is a family of walk-flows for all commodities and all walks such that for every commodity the sum of its walk-flows matches its network inflow rate.

\section{Dynamic   Equilibria with Convex  Constraints }
So far, we have reduced the strategy space of every
agent involving the routing and recharging decisions to the set of feasible walks 
inside the battery-extended graph $G$. What is still missing to formally introduce
the traffic assignment problem, or equivalently, the dynamic equilibrium problem,  is the
precise form of the utility function for an agent. We assume that agents want to
travel from $s_i$ to $t_i$ but
have preferences over travel time and recharge prices.
While the recharge prices can be directly derived from the chosen walk $W$,
the resulting travel time can only be described, if the walk-choices of all agents have
been unfolded over time giving the resulting queueing times of a walk.
This dynamic unfolding of the traffic inflow is usually termed as the \emph{network loading}
which is discussed in the following subsections.
\subsection{Edge-Walk-Based Flows over Time}
Given a feasible walk-flow $h\in K$, we develop the theoretical basis for the resulting \emph{network loading}. This network loading provides then the basis for a travel time function $\mu_i^W: \IR_{\geq 0} \to \IR_{\geq 0}$ which for every time $\theta$ provide us with the travel time for a particle entering walk $W$ at time $\theta$. These functions will then be used for our dynamic equilibrium concept which takes energy-feasibility of walks and their resulting travel time into account. 

Let 
$\mathcal{R}:=\set{(i,W,j) | i \in I, W \in \Pc_i, j \in [k_W]}$
denote the set of triplets consisting of the commodity identifier, walk and edge position in the 
walk, respectively. A flow over time is then a tuple $f = (f^+,f^-)$, where
$f^+, f^- \in \big(L^2_{\geq0}(\IR_{\geq 0})\big)^{\mathcal{R}}$  are vectors of $L^2$-integrable, non-negative functions modelling the  inflow rate $f^{W,+}_{i,j}(\theta)$ and  outflow rate $f^{W,-}_{i,j}(\theta)$ of commodity $i$ on the $j$-th edge of some walk $ W\in \Pc_i$. 
For any such flow over time we define the \emph{aggregated} edge in- and outflow rates of an edge $e \in E$ as
\begin{equation}\label{eq:agg-in} 
	f_{e}^+(\theta) \coloneqq \sum_{\substack{(i,W,j) \in \mathcal{R}:\\e^W_j = e}}f^{W,+}_{i,j}(\theta)
	\quad \text{ and } \quad
	f_{e}^-(\theta) \coloneqq \sum_{\substack{(i,W,j) \in \mathcal{R}:\\e^W_j = e}}f^{W,-}_{i,j}(\theta)
\end{equation}
and the cumulative edge in- and outflows by
\begin{align*}
	F_{i,j}^{W,+}(\theta):=\int_{0}^{\theta} f_{i,j}^{W,+}(z)dz \quad \text{ and } \quad F_{i,j}^{W,-}(\theta):=\int_{0}^{\theta}  f_{i,j}^{W,-}(z)dz
\end{align*}
as well as 
\begin{align*}
	F_{e}^+(\theta):=\int_{0}^{\theta} f_{e}^+(z)dz \quad \text{ and } \quad F_{e}^-(\theta):=\int_{0}^{\theta}  f_{e}^-(z)dz.
\end{align*}
Note, that $F^+_e, F^-_e, F_{i,j}^{W,+}$ and $F_{i,j}^{W,-}$ are non-decreasing, absolute continuous functions which satisfy
\begin{align*} 
	F_{e}^+(\theta) = \sum_{\substack{(i,W,j) \in \mathcal{R}:\\e^W_j = e}}F^{W,+}_{i,j}(\theta)
	\quad \text{ and } \quad
	F_{e}^-(\theta) = \sum_{\substack{(i,W,j) \in \mathcal{R}:\\e^W_j = e}}F^{W,-}_{i,j}(\theta).
\end{align*}
Furthermore, we define the \emph{queue length} of an edge $e$ at time $\theta$ by
\begin{align}\label[cons]{eq:Cont-FlowDefProperties-QueueLengthWithF}
	q_e(\theta) \coloneqq F^+_e(\theta) - F^-_e(\theta+\tau_e) & \quad \text{ for all } \theta \in \IR_{\geq 0}.
\end{align}
For any flow particle entering an edge $e=vw$ at time time $\theta$, its travel time on this edge is
\begin{align*}
	c_e(\theta) \coloneqq \tau_e + \frac{q_e(\theta)}{\nu_e}
\end{align*}
and its exit time of edge $e$ is given by
\begin{equation}\label{eq:travel} 
	T_e(\theta) \coloneqq \theta+c_e(\theta).
\end{equation}

Now, given some walk-flow $h \in K$ we call a flow over time $f$ a \emph{feasible flow over time associated with $h$} if it satisfies the following \cref{eq:FlowCons-Nodes,eq:FlowCons-Source,eq:RespCapacity,eq:WeakFlowCons-Edges,eq:FlowCons-Edges}:

The walk-flows of $h$ and $f$ match, i.e. for every $i \in I$, $W \in \Pc_i$ we have
\begin{align}\label[cons]{eq:FlowCons-Source} 
	f^{W,+}_{i,1}(\theta) = h_i^W(\theta) \text{ for almost all } \theta \in \IR_{\geq 0}.
\end{align}
The flow satisfies a balancing constraint at every node intermediate node, i.e. for every $i \in I$, $W \in \Pc_i$ and any $1 \leq j < k_W$ we have
\begin{align}\label[cons]{eq:FlowCons-Nodes} 
	f^{W,-}_{i,j}(\theta) = f^{W,+}_{i,j+1}(\theta) \text{ for almost all } \theta \in \IR_{\geq 0}.
\end{align}
The aggregated outflow respects the edges capacity, i.e. for every edge $e$ we have
\begin{align}\label[cons]{eq:RespCapacity} 
	f^-_e(\theta+\tau_e) \leq \nu_e \text{ for almost all } \theta \in \IR_{\geq 0},
\end{align}
as well as weak flow conservation over edges, i.e. for every edge $e$ we have
\begin{align}\label[cons]{eq:WeakFlowCons-Edges} 
	F^-_e(\theta+\tau_e) \leq F^+_e(\theta) \text{ for all } \theta \in \IR_{\geq 0}.
\end{align}
And, finally, the flow has to satisfy the following link transfer equation for every $i \in I$, $W \in \Pc_i$ and any $1 \leq j \leq k_W$:
\begin{align}\label[cons]{eq:FlowCons-Edges}
	F^{W,-}_{i,j}\left(T_{e^W_j}(\theta)\right) = F^{W,+}_{i,j}(\theta) \text{ for all } \theta \in \IR_{\geq 0}.
\end{align}

It turns out that every walk-flow $h \in K$ has a \emph{unique} associated feasible flow over time which we can obtain by a natural network loading procedure.

\begin{lemma}\label{lemma:UniqueNetworkLoading}
	For any $h \in K$, there is a unique (up to changes on a subset of measure zero) associated flow over time $f$.
\end{lemma}

\begin{proof}
	The proof of this lemma mainly rests on the following two claims:
	\begin{claim}\label{claim:UniqueAggregatedOutflowrate}
		Given an aggregated edge inflow rate functions $f^+_e$ on some interval $[0,\bar\theta]$ then there exists a uniquely defined (up to changes on a set of measure zero) aggregated edge outflow rate $f^-_e$ on the interval $[0,T_e(\bar\theta)]$ satisfying \eqref{eq:RespCapacity}, \eqref{eq:WeakFlowCons-Edges} and 
		\begin{align}\label{eq:FlowCons-Edges-Aggregated}
			F^-_e\left(T_{e}(\theta)\right) = F^+_e(\theta)
		\end{align}
		for all $\theta \in [0,\bar\theta]$.
	\end{claim}
	
	\begin{proofClaim}
		We first observe, that \eqref{eq:FlowCons-Edges-Aggregated} is equivalent to
		\begin{align}\label{eq:WaitingTimeWellDefined}
			\frac{q_e(\theta)}{\nu_e} = \min\left\{w \geq 0 \middle\vert \int_\theta^{\theta+w}f^-_e(\zeta+\tau_e)\diff\zeta = q_e(\theta)\right\}.
		\end{align}
		Indeed, if \eqref{eq:WaitingTimeWellDefined} holds, we get
		\begin{align*}
			F^-_e\left(T_{e}(\theta)\right) 
				&= F^-_e\left(\theta + \tau_e + \frac{q_e(\theta)}{\nu_e}\right) = F^-_e(\theta+\tau_e) + \int_\theta^{\theta+\frac{q_e(\theta)}{\nu_e}}f^-_e(\zeta+\tau_e)\diff\zeta \\
				&\overset{\text{\eqref{eq:WaitingTimeWellDefined}}}{=} F^-_e(\theta+\tau_e) + q_e(\theta) = F_e^+(\theta).
		\end{align*}
		If, on the other hand, \eqref{eq:FlowCons-Edges-Aggregated} holds, then we have
		\begin{align*}
			\int_\theta^{\theta+\frac{q_e(\theta)}{\nu_e}}f^-_e(\zeta+\tau_e)\diff\zeta = F^-_e\left(T_{e}(\theta)\right) - F^-_e(\theta) =  F^-_e\left(T_{e}(\theta)\right) - F^+_e(\theta) + q_e(\theta) \overset{\text{\eqref{eq:FlowCons-Edges-Aggregated}}}{=} q_e(\theta)
		\end{align*}
		and, therefore, $\min\set{w \geq 0 | \int_\theta^{\theta+w}f^-_e(\zeta+\tau_e)\diff\zeta = q_e(\theta)} \leq \frac{q_e(\theta)}{\nu_e}$. At the same time, \eqref{eq:RespCapacity} clearly implies $\min\set{w \geq 0 | \int_\theta^{\theta+w}f^-_e(\zeta+\tau_e)\diff\zeta = q_e(\theta)} \geq \frac{q_e(\theta)}{\nu_e}$ and, thus, \eqref{eq:WaitingTimeWellDefined} holds.
		
		Now we can use \cite[Proposition 1]{CominettiCL15} stating that satisfying \eqref{eq:RespCapacity}, \eqref{eq:WeakFlowCons-Edges} and \eqref{eq:WaitingTimeWellDefined} is equivalent to the queue operating at capacity, i.e. 
		\begin{align*}
			f^-_e(\theta + \tau_e) = \begin{cases}
				\nu_e, &\text{ if } q_e(\theta) > 0 \\
				\min\set{f^+_e(\theta),\nu_e}, &\text{ else }
			\end{cases}
			\text{ for almost all } \theta \leq \bar\theta.
		\end{align*}
		Now, the claim follows from \cite[Lemma 3.4.5]{MarklThesis} or, alternatively, in the way mentioned in \cite[Section 2.2]{CominettiCL15}.
	\end{proofClaim}
	
	\begin{claim}
		Given a family of edge inflow rates $(f^+_{e,i})$ for some edge $e$ on some interval $[0,\theta]$ then there exist uniquely defined (up to changes on a set of measure zero) edge outflow rates $(f^-_{e,i})$ on $[0,T_e(\theta)]$ satisfying \cref{eq:RespCapacity,eq:WeakFlowCons-Edges,eq:FlowCons-Edges}.
	\end{claim}
	
	\begin{proofClaim}
		First, note that any such family of outflow rates must, in particular, have a corresponding aggregated outflow rate satisfying the properties of \Cref{claim:UniqueAggregatedOutflowrate}. Since we already know that there exists exactly one such outflow rate $f^-_e: [0,T_e(\theta)] \to \IR_{\geq 0}$, we can take this as given. Now, the claim reduces to showing that there exist uniquely defined outflow rates $(f^-_{e,i})$ adding up to the fixed aggregated edge outflow rate $f^-_e$ and satisfying \eqref{eq:FlowCons-Edges} for almost all $\theta \leq \bar\theta$. This is now exactly the statement of \cite[Lemma 2]{CominettiCL15}.
	\end{proofClaim}
	
	Using these two claims together with the flow conservation constraints \eqref{eq:FlowCons-Source} and \eqref{eq:FlowCons-Nodes} the (unique) existence of associated flows over time can be shown by induction over time in the same way as in the proof of \cite[Proposition 3]{CominettiCL15}. 
\end{proof}

For any fixed network we then denote by $\mathcal{F}$ the set of all feasible flows over time associated with some $h \in K$. Note, that \Cref{lemma:UniqueNetworkLoading} then provides us with a one-to-one mapping between $K$ and $\mathcal{F}$.

\subsection{Capacitated Dynamic Equilibria}

For a given walk-based flow $h \in K$ with associated feasible flow over time $(f^+, f^-)$, we now want to compute for every commodity type $(i,W)$
with $W=(e^W_1,\dots,e^W_{k_W})$ a label function
giving at time $\theta$ for any node on that walk the arrival time at $t_i$.
Let  $\hat W=(v_{0},\dots, v_{k_W})$ denote the representation of $W$
as a sequence of nodes satisfying  $e^W_{j}=v_{j-1}v_{j}, j\in [k_W]$ with $v_{0}=s_i, v_{k_W}=t_i$.
As a node can appear multiple times in $W$, we use the subindex $j\in [k_W]$ as a unique identifier 
of the position of that node in the walk.  
With this  notation  we can unambiguously 
and recursively define the following label function:
\begin{equation}
	\begin{aligned}
		\ell_{i,k_W}^W(\theta)&:= \theta, \text{ for  all }\theta\geq0,\\
		\ell_{i,j}^W(\theta)&:=\ell_{i,j+1}^W(T_{e^W_{j+1}}(\theta)), \text{ for }  j= [k_W]-1,\dots,0  \text{ and all } \theta\geq0
	\end{aligned}
\end{equation}
where $\ell_{i,j}^W$ is the label function of the $(j+1)$-th node when 
traversing the walk $\hat W$ beginning with the starting node at position $0$.
Recall that
$v_{k_W}=v_{t_i}$  and  $v_{0}=s_i$,
thus, for a particle
entering $W$ at time $\theta$, the value $\ell_{i,0}^W(\theta)$ measures the arrival time at $t_i$ (assuming that the particle  follows $W$). 
Note that $\ell_{i,j}^W$ is only defined for nodes contained in $W$ and
a node $v$  in $\hat W$ may be associated with several label functions whose
number is equal to the number of  occurrences of $v$ in $\hat W$.
We can easily compute the total travel time for a vehicle of commodity $i\in I$
leaving $s_i$ at time $\theta$ as
\begin{equation}\label{eq:travel_time}
	\mu_{i}^W(\theta):=\ell_{i,0}^W(\theta)-\theta.
\end{equation}

Finally, we need to connect the total travel time with the total price 
paid for recharging.
For this, we introduce an \emph{aggregation function}.
\begin{defn}
	A function $c:\R\times\R\rightarrow \R$ is an \emph{aggregation function},
	if $c$ is continuous and non-decreasing in both arguments.
\end{defn}
We will assume that every commodity has a commodity-specific aggregation function $c_i$ and, thus, given some fixed feasible flow $f$ particles of commodity $i$ starting their travel along some walk $W$ at time $\theta$ have a total cost of $c_i\left(\mu_{i}^W(\theta),\sum_{e\in W}p_{i,e}\right)$.
\begin{example}\label{example:AggregationFunction}
	For $i\in I$, we can think of an aggregation function as
	\[ c_i\left(\mu_{i}^W(\theta),\sum_{e\in W}p_{i,e}\right):=\lambda_i \mu_{i}^W(\theta)+ \sum_{e\in W}p_{i,e},\]
	where $\lambda_i>0$ is a parameter that translates the total travel time into 
	disutility measured in Euro.
\end{example}

Now, instead of letting particles choose \emph{any} walk between their respective source and sink node, we impose further restrictions to only use walk-flows from some closed, convex restriction set $S \subseteq L^2([0,T])^{\Pc}$. Using such $S$ we can, for example, not only model battery constraints (making certain walks infeasible) but also temporary road closures or restrictions on the set of feasible flows itself (as every $h$ corresponds to a unique feasible flow) -- though, in the latter case it is in general not obvious to see whether the resulting set $S$ satisfies convexity.

Now, we want to express that some $h \in S$ is an equilibrium, if no particle can improve its total cost (i.e. aggregate of travel time and total price) by deviating from its current path while staying within $S$. However, since individual particles are infinitesimally small, the deviation of a single particle does not influence the feasibility w.r.t. $S$. Instead, we have to consider deviations of arbitrarily small but positive volumes of flow leading to the notion of \emph{saturated} and  \emph{unsaturated} walks as used in the static Wardropian model by Larsson and Patriksson~\cite{Larsson95}. To do that we first define for any given walk-flow $h$, commodity $i$, walks $W, Q \in \Pc_i$, time $\bar\theta \geq 0$ and constants $\varepsilon,\delta > 0$ the walk-flow obtained by shifting flow of commodity $i$ from walk $W$ to walk $Q$ at a rate of $\varepsilon$ during the interval $[\bar\theta,\bar\theta+\delta]$ by 
\begin{align*}
	H^{W\to Q}_i(h,\bar\theta,\varepsilon,\delta) \coloneqq (h'_R)_{R \in \Pc} \text{ with } \begin{aligned} 
		h'^W_i &=[h^W_i-\epsilon \CharF[{[\bar\theta,\bar\theta+\delta]}]]_+\\
		h'^Q_i &=h^Q_i+h^W_i-[h^W_i-\epsilon \CharF[{[\bar\theta,\bar\theta+\delta]}]]_+\\
		h'^{R}_{i'} &= h^{R}_{i'} \text{ f.a. } (i',R) \in \Pc \setminus \{(i,Q),(i,W)\}\end{aligned},
\end{align*}
where 
\[\CharF[{[\bar\theta,\bar\theta+\delta]}]: [0,T] \to \R, \theta \mapsto \begin{cases}
	1, &\text{ if } \theta \in [\bar\theta,\bar\theta+\delta] \\
	0, &\text{ else }
\end{cases}\]
is the indicator function of the interval $[\bar\theta,\bar\theta+\delta]$ and for any function $g: [0,T] \to \IR$ the function $[g]_+$ is the non-negative part of $g$, i.e. the function
	\[[g]_+: [0,T] \to \IR, \theta \mapsto \max\{g(\theta),0\}.\]
Using this notation, we can now define the set of unsaturated alternatives to some fixed walk $W$ of commodity $i$ with respect to some $h \in S$ at time $\bar\theta \geq 0$ as
\begin{equation}\label{eq:feasible-deviation} 
	D^W_i(h, \bar\theta):=\left\{
		Q\in \Pc_i\middle\vert  \forall \delta' > 0: \exists \delta \in (0,\delta'], \varepsilon>0: 
		H^{W\to Q}_i(h,\bar\theta,\varepsilon,\delta) \in S
	\right\}.
\end{equation}

We can now formally define the concept of a dynamic equilibrium in our model.

\begin{defn}\label{def:DCE}
	Given a network $\network=(G,\nu,\tau,p)$, a set of commodities $I$, a restriction set $S$ and for every commodity $i \in I$ an associated source-sink pair $(s_i,t_i) \in V \times V$ as well as an aggregation function $c_i$, a feasible walk-flow $h\in S \cap K$ is a \emph{capacitated dynamic equilibrium}, if for all $(i,W) \in \Pc$ and almost all $\bar\theta\geq 0$ it holds that
	\begin{equation}\label{eq:CDE} 
		h^W_i(\bar\theta) > 0 \implies c_i\left(\mu_{i}^W(\bar\theta),\sum_{e\in W}p_{i,e}\right) \leq c_i\left(\mu_{i}^Q(\bar\theta),\sum_{e\in Q}p_{i,e}\right) \text{ for all }Q\in D^W_i(h, \bar \theta).
	\end{equation}
\end{defn}

A special case of particular interest are sets $S$ defined by restricting the set of walks that can be used by each commodity, i.e. $S \coloneqq \set{h \in K | h^W_i \equiv 0 \text{ for all } W \in \Pc\setminus\widetilde{\Pc}}$ for some subset $\widetilde{\Pc} \subset \Pc$ of feasible walks. Then we have $D^W_i(h,\bar\theta) = \set{W \in \Pc_i | (i,W) \in \widetilde{\Pc}}$ for any $i \in I,W \in \Pc_i,h \in S,\bar\theta \geq 0$ and the above definition just requires that whenever there is positive inflow into a walk, this walk must have minimal total costs among all feasible paths of the respective commodity at that time. 

In particular, in the case where all walks are allowed (i.e. $\widetilde{\Pc} = \Pc$) the above definition is equivalent to the classic definition of dynamic equilibria. For a battery-extended network we can take $\widetilde{\Pc} = \Pc_b$ and will call a capacitated dynamic equilibrium an \emph{energy-feasible dynamic equilibrium}.

%% file: tikz/QueueModel.tex
\begin{tikzpicture}[scale=1.2]

		\begin{scope}
		
		\node (s1) at (0, 0) {};
		\draw (s1) +(0, 0) node[labeledNode] (s1) {$v$};
		\node[rectangle, color=gray!50,draw, minimum width=3cm, minimum height=0.1cm,fill={gray!50}] at (1.59,0.2) {};
		\node[rectangle, color=gray!50,draw, minimum width=0.5cm, minimum height=1cm,fill={gray!50}] at (0.55,0.55) {};

		\draw (s1) +(4, 0) node[labeledNode] (s3) {$w$}
		edge[normalEdge, <-] node[below] {$\tau_e$}
		node[above] {} (s1)

		;
		;
		\node at (-0.75,0) [] (inflow) { inflow};
		\node at (1.25,1.25) [] (queue) { queue $q_e(\theta)$};
		\node at (0.5,1.0) [] (qu) { };
		\draw (s1) edge[thick,>=stealth,dashed,bend left=90,->] node [left] {$f_e^+(\theta)$} (qu);
		\node at (6,1) [] (outflow) { outflow};
		\node at (4,0.2) [] (out) { };
		\draw (out) edge[thick,>=stealth,dashed,bend left=60,->] node [left] {$f_e^-(\theta)\;\;$} (outflow);
		\draw [<->, line width=0.5pt] (2.95,0.1) -- (2.95,0.3);
		\node at (3.2,0.2) [] (nu) { $\nu_e$};
			\end{scope}
\end{tikzpicture}

%% file: chapters/Existence.tex

\section{Existence of Capacitated Dynamic Equilibria}

In this section, we will show the existence of capacitated dynamic equilibria using 
an infinite dimensional variational inequalities as pioneered by Friesz et al.~\cite{Friesz93} and
also used by Cominetti et. al.~\cite{CominettiCL15}. Since we use a more general equilibrium concept and allow for flow to use arbitrary walks (from an a priori \emph{infinite} set of possible walks) instead of just simple paths, we have to adjust several of the technical steps of the proof.

The general structure of the proof will be as follows: First, we introduce the concept of dominating sets of walks which will allow us to only consider some finite subset $\Pc'$ of the set of all walks. We then define a mapping $\mathcal{A}: h \mapsto c_i\left(\mu_{i}^W(\_),\sum_{e\in W}p_{i,e}\right)$ mapping walk-flows to costs of particles of commodity $i$ using walk $W$. Using this mapping we can then formulate a variational inequality for which we can show that any solution to it is a capacitated dynamic equilibrium. Finally, a result by Lions~\cite{Lions} guarantees the existence of such solutions given that the mapping $\mathcal{A}$ satisfies an appropriate continuity property which we will show to hold for our model.

We start by giving the definition of dominating walks and sets to be able to formally state our main theorem:

\begin{defn}\label{def:dominating}
	A walk $(i,Q') \in \Pc$ is a \emph{dominating walk} for another walk $(i,Q)$ with respect to $S$, if for any walk-flow $h \in K \cap S$ we have
		\[c_i\left(\mu_{i}^{Q'}(\bar\theta),\sum_{e\in Q'}p_{i,e}\right) \leq c_i\left(\mu_{i}^Q(\bar\theta),\sum_{e\in Q}p_{i,e}\right) \text{ for all } \theta \in [0,T],\]
	and, additionally, $Q \in D^W_i(h,\bar{\theta})$ always implies $Q' \in D^W_i(h,\bar{\theta})$.

	A subset $\Pc' \subseteq \Pc$ \emph{is a dominating set with respect to $S$}, if it contains for any walk $(i,Q) \in \Pc$, a \emph{dominating  walk} $(i,Q') \in \Pc'$ with respect to $S$.
\end{defn}

\begin{theorem}\label{thm:existence}
	Let $\network = (G,\nu,\tau,p)$ be any network and $I$ a finite set of commodities each associated with an aggregation function $c_i$ and a source-sink pair $(s_i,t_i)$. Let $S \subseteq L^2([0,T])^{\Pc}$ be a restriction set which is closed, convex and has non-empty intersection with $K$ and there exists some finite dominating set $\Pc' \subseteq \Pc$ with respect to $S$. Then there exists a capacitated dynamic equilibrium in $\network$.
\end{theorem}

In order to prove this theorem we first need some definitions from functional analysis:
We will make use of two function spaces, namely the space $L^2([a,b])$ of $L^2$-integrable functions from an interval $[a,b]$ to $\IR$ and the space $C([a,b])$ of continuous functions from $[a,b]$ to $\IR$. The former is a Hilbert space with the natural pairing
\begin{align*}
	\langle.,.\rangle: L^2([a,b]) \times L^2([a,b]) \to \IR, (g,h) \mapsto \langle g,h\rangle \coloneqq \int_a^b g(x)h(x)\diff x.
\end{align*}
The latter is a normed space with the uniform norm $\norm{f}_\infty \coloneqq \sup_{\theta \in [a,b]}\abs{f(\theta)}$. Both, the natural pairing and the norm, can be extended in a natural way to $L^2([a,b])^d$ and $C([a,b])^d$, respectively, for any $d \in \IN$. In particular, all these spaces are topological vector spaces. We say that a sequence $h^k$ of functions in $L^2([a,b])^d$ \emph{converges weakly} to some function $h \in L^2([a,b])^d$ if for any function $g \in L^2([a,b])$ we have $\lim_{k \to \infty}\langle h^k, g\rangle = \langle h,g\rangle$. For any topological space $X$ (in the following this will be either $L^2([a,b])^d$ or $C([a,b])^d$) and any subset $C \subseteq  L^2([a,b])^d$ a mapping $\mathcal{A}: C \to X$ is called \emph{sequentially weak-strong-continuous} if it maps any weakly converging sequence of functions in $C$ to a (strongly) convergent sequence in $X$.

With this, we can now describe the kind of variational inequality we will use to show the existence of capacitated dynamic equilibria. Namely, given an interval $[a, b] \subseteq \IR_{\geq 0}$, a number $d \in \IN$, a subset $C \subseteq L^2([a, b])^d$ and a mapping $\A: C \to L^2([a, b])^d$, the variational inequality~\ref{eqn:VI} is the following:
\begin{equation}\label{eqn:VI}\tag{\ensuremath{\VI(C,\A)}}
	\text{Find }h^* \in C \text{ such that } \scalar{\A(h^*)}{\bar h-h^*} \geq 0 \text{ for all } \bar h \in C.
\end{equation}
Conditions to guarantee the existence of such an element $h^*$ are given by Lions in \cite[Chapitre 2, Théorème 8.1]{Lions} which, following Cominetti et. al.~\cite{CominettiCL15}, can be restated as follows:
\begin{theorem}\label{brezis}
	Let $C$ be a non-empty, closed, convex and bounded subset of $L^2([a, b])^d$. Let $\A : C \rightarrow L^2([a, b])^d$ be a sequentially weak-strong-continuous mapping. Then, the variational inequality~\eqref{eqn:VI} has a solution $h^* \in C$.
\end{theorem}

For our proof we choose $C \coloneqq \pi(S \cap K \cap \iota(\big(L^2([0,T])\big)^{\Pc'}))$, where $\iota: \big(L^2([0,T])\big)^{\Pc'} \to \big(L^2([0,T])\big)^{\Pc}$ is the canonical embedding (i.e. augmenting $\Pc'$-dimensional vectors with zero functions to $\Pc$-dimensional vectors) and $\pi: \big(L^2([0,T])\big)^{\Pc} \to \big(L^2([0,T])\big)^{\Pc'}$ the canonical projection. For ease of notation we will usually omit these embeddings/projections from our notation and assume that they are implicitly applied, whenever we change between elements of $\big(L^2([0,T])\big)^{\Pc'}$ and $\big(L^2([0,T])\big)^{\Pc}$. 
Next, we define a mapping $\A : C \rightarrow L^2([0, T])^{\Pc'}$ by defining for every walk-flow $h \in C$, commodity $i \in I$ and walk $W \in \Pc'_i \coloneqq \set{W \in \Pc_i | (i,W) \in \Pc'}$ the continuous function
$\A^W_i(h)$ by
\[\A^W_i(h): \theta\mapsto c_i\left(\mu_{i}^W(\bar\theta),\sum_{e\in W}p_{i,e}\right) - \min_{Q\in\Pc'_i} c_i\left(\mu_{i}^Q(\bar\theta),\sum_{e\in Q}p_{i,e}\right).  \]

Clearly, the assumptions on $S$ and the fact that $K$ is bounded, closed and convex imply that $C$ is a non-empty, closed, convex and bounded subset of $L^2([0, T])^{\Pc'}$. Thus, in order to be able to apply \Cref{brezis} it only remains to show that $\mathcal{A}$ is sequentially weak-strong continuous. Since taking differences and minima of sequentially weak-strong continuous mappings again results in such a mapping, it suffices to show that the maps
\begin{equation} 
	h\mapsto c_i\left(\mu_{i}^W(\_),\sum_{e\in W}p_{i,e}\right) \text{ for }W\in \Pc'_i, i\in I  
\end{equation}
are sequentially weak-strong continuous from $C$ to $L^2([0,T])$.
The first ingredient for this proof is a result by Cominetti, Correa and Larr\'{e}~\cite[Lemma 3]{CominettiCL15}
showing that the network loading for any $h\in C$ has compact support.
One argument in~\cite{CominettiCL15} uses that walks are simple (paths) 
in order to bound queues. We side-step this argument by using another Lemma
from Graf and Harks~\cite{GrafH20}.
\begin{lemma}\label{lemma:BoundedSupportOfNetworkLoading}
	There is a constant $M\geq 0$ such that for any $h\in C$ all edge flows
	of the network loading corresponding to $h$ are supported on $[0,M]$.
\end{lemma}
\begin{proof}
	Let $h\in K$ and let $f\in \F$ be the associated unique network loading
	which itself is a feasible flow over time.
	For any feasible flow $f\in \F$ and every edge $e \in E$ we define the \emph{edge load} function $\edgeLoad[e]$ that gives us for any time $\theta$ the total amount of flow currently on edge $e$ (either waiting in its queue or traveling along the edge): 
	\[\edgeLoad[e]: \IR_{\geq 0} \to \IR_{\geq 0}, \theta \mapsto F^+_e(\theta)-F^-_e(\theta). \]
	The function $\edgeLoad(\theta) := \sum_{e \in E}\edgeLoad[e](\theta)$ then gives  the \emph{total amount of flow in the network at time $\theta$}. 
	Furthermore, we define a function $Z$ indicating the amount of flow that already reached the sinks $t_i, i\in I$ by time $\theta$:
	\begin{align}\label{eq:DefZ}
		Z: \IR_{\geq 0} \to \IR_{\geq 0}, \theta \mapsto \sum_{i\in I}\sum_{e \in \edgesEntering{t_i}}F^-_{i,e}(\theta) - \sum_{e \in \edgesLeaving{t_i}}F^+_{i,e}(\theta)
	\end{align}
	and for any node $s_i \neq t_i, i\in I$ the \emph{cummulative network inflow at $s_i$} as
	\[U_{i}: \IR_{\geq 0} \to \IR_{\geq 0}, \theta \mapsto \int_{0}^{\theta} u_i(\zeta)d\zeta.\]
	We will make use of the following connection between these functions:			
	\begin{lemma}[{Graf and Harks~\cite[Lemma 2.1]{GrafH20}}]\label{lem:GeIsTotalFlowInGraph}
		Let $f\in \F$ be a feasible flow. Then for any time $\theta$ we have
		\[\edgeLoad(\theta) = \sum_{i\in I }U_{i}(\theta) - Z(\theta). \]
	\end{lemma}
	
	From the above lemma, we immediately get:	
	\[ q_e(\theta)\leq \edgeLoad(\theta)\leq \bar q:= \sum_{i\in I} \int_{\IR}u_{i}(z)dz.\]
	The remaining proof is similar to that in~\cite{CominettiCL15}.
	Let $\delta:=\max_{e\in E}\{\frac{\bar q}{\nu_e}+\tau_e\}$
	and $m$ be the maximum number of edges of any of the (finitely many!) walks in $\Pc'$.
	Hence by setting $M:=T+\delta m$ we get $\ell_{i,0}(\theta)\leq M$ for all $\theta\in [0,T]$.
	With~\eqref{eq:FlowCons-Nodes} and \eqref{eq:FlowCons-Edges} we get that all
	appearing edge flows are supported on $[0,M]$.
\end{proof}

Now we discuss the continuity of the mapping from walk-flows $h$ to label functions $\ell^W_{i,j}$ following along the same lines as \cite[Lemmas 4-7]{CominettiCL15}.

\begin{lemma}\label{lemma:PathInflowToArrivalTimeIsWSC}
	Let $M \geq 0$ be a constant such that all edge flows of network loadings corresponding to any $h \in C$ have their support in $[0,M]$ (cf. \Cref{lemma:BoundedSupportOfNetworkLoading}). Then for any fixed edge $e \in E$ the map
	\[C\to C([0,M]), h \mapsto T_e\]
	is sequentially weak-strong continuous. Here, the map $h \mapsto T_e$ is defined by first finding the unique network loading for the given $h$ (see \Cref{lemma:UniqueNetworkLoading}) and then deriving the resulting exit time function $T_e$.
\end{lemma}

\begin{proof}
	We show the desired sequential continuity by decomposing the map into three maps according to the following commutative diagram:
	\begin{center}
		\begin{tikzpicture}
			\node(domH)at(0,0) {$C$};
			\node(h)at($(domH)+(0.5,-0.5)$) {\small$h$};
			\node(domT)at(8,0) {$C([0,M])$};
			\node(Te)at($(domT)+(-0.5,-0.5)$) {\small$T_e$};
			\node(domF)at(0,-3) {\hspace{6em}$\mathcal{F} \subseteq \big(L^2([0,M])\big)^{\mathcal{R}}$};
			\node(fe)at($(domF)+(0.5,0.5)$) {\small$f$};
			\node(domQ)at(8,-3) {$\big(C([0,M])\big)^2$};
			\node(qe)at($(domQ)+(-0.5,0.5)$) {\small$(F_e^+,F_e^-)$\quad\quad\quad};
			
			\draw[->] (domH) --node[above]{\tiny weak-strong} (domT);
			\draw[->] (domH) --node[below,sloped]{\tiny weak-weak} (domF);
			\draw[->] (domF) --node[below]{\tiny weak-strong} (domQ);
			\draw[->] (domQ) --node[below,sloped]{\tiny strong-strong} (domT);
			
			\draw[|->] (h) -- (Te);
			\draw[|->] (h) -- (fe);
			\draw[|->] (fe) -- (qe);
			\draw[|->] (qe) -- (Te);
		\end{tikzpicture}
	\end{center}

	\begin{claim}\label{claim:fToFisWSC}
		The map $\F \to \big(C([0,M],\IR_{\geq 0})\big)^2, f \mapsto (F_e^+,F_e^-)$ is well defined and sequentially weak-strong continuous.
	\end{claim}
	
	\begin{proofClaim}
		Since $\theta \mapsto \int_0^\theta g(\zeta)\diff\zeta$ is an (absolute) continuous function for any integrable function $g$ the map given in the claim is clearly well defined. For the proof of the continuity we mostly follow the proof of \cite[Lemma 2.7]{SeringThesis}. Let $(f^k)_k$ be a sequence in $\F$ converging weakly to some $f \in \F$ and let ${F^+_e}^{(k)}, F^+_e, {F^-_e}^{(k)}$ and $F^-_e$ be the corresponding cumulative flows on edge $e$. As first step we want to show that ${F^+_e}^{(k)}$ converges point-wise $F^+_e$. So, fix some time $\theta \in [0,M]$ and define $g \in \big(L^2([0,M])\big)^{\mathcal{R}}$ by
			\begin{align*}
				g^{W,+}_{i,j} \coloneqq \begin{cases}
					\CharF[{[0,\theta]}], &\text{ if } e^W_j = e \\
					0					, &\text{ else }
				\end{cases} 
					\text{ and }
				g^{W,-}_{i,j} \coloneqq 0 \text{ for all } (i,W,j) \in \mathcal{R}.
			\end{align*}
		Then we have
			\[F_e^+(\theta) = \langle f,g\rangle = \lim_k \langle f^k,g\rangle = {F^+_e}^{(k)}(\theta).\]
		In exactly the same way, one can also show that ${F^-_e}^{(k)}$ converges point-wise to $F^-_e$. Now, we observe that any aggregated edge inflow rate into edge $e = vw$ for any feasible flow over time is clearly bounded almost everywhere by $L \coloneqq \sum_{e' \in \edgesEntering{v}}\nu_{e'} + \sup_{\theta \in [0,M]}\sum_{i \in I} u_i(\theta)$ and any aggregated edge outflow rate is bounded almost everywhere by $\nu_e$ (using \cref{eq:FlowCons-Source,eq:FlowCons-Nodes,eq:RespCapacity}). Thus, all ${F^+_e}^{(k)}$ are Lipschitz-continuous with Lipschitz constant $L$ while all ${F^-_e}^{(k)}$ are Lipschitz-continuous with Lipschitz constant $\nu_e$. Also, both $F^+_e$ and $F^-_e$ are continuous. Thus, we can apply \Cref{lemma:UniformConvergenceOfPointwiseConvergence} to obtain uniform convergence.
	\end{proofClaim}
	
	\begin{claim}\label{claim:FtoTisSSC}
		The map $\big(C([0,M])\big) \to C([0,M]), (F_e^+,F_e^-) \mapsto T_e$ is sequentially strong-strong continuous.
	\end{claim}
	
	\begin{proofClaim}
		This follows directly from the definition of $T_e(\theta) \coloneqq \theta + \tau_e  + \frac{F^+_e(\theta)-F^-_e(\theta+\tau_e)}{\nu_e}$.
	\end{proofClaim}
	
	\begin{claim}
		The network loading map $C \to \F, h \mapsto f$ is sequentially weak-weak continuous.
	\end{claim}
	
	\begin{proofClaim}
		Let $(h^k)_k \subseteq K$ be a sequence of walk-flows converging weakly to some $h \in C$. Let $f^k$ and $f$ be the associated network loadings. We want to show that then $f^k$ converges weakly to $f$. By way of contradiction we assume that this is not the case. In particular, that means that there exists some $\varepsilon > 0$ and a subsequence $f^{k_j}$ as well as an element $g \in L^2([0,M])^d$ such that $\langle g,f^{k_j}\rangle > \varepsilon$ for all $j \in \IN$. 
		
		By \Cref{lemma:BoundedSupportOfNetworkLoading}, the sequence $(f^{k_j})_j$ is bounded. Since $L^2([0,M])^d$ is a reflexive Banach space, this implies that it contains a weakly convergent subsequence (see \cite[Satz 6.10]{AltFunkAna}). By some abuse of notation we will denote this subsequence by $(f^k)$ and its weak convergence point by $f'$. We will now show that $f' \in \F$ and it is a network loading for $h$, i.e. we want to show that $f'$ satisfies \cref{eq:FlowCons-Source,eq:FlowCons-Nodes,eq:RespCapacity,eq:WeakFlowCons-Edges,eq:FlowCons-Edges}. We will do that by showing that all these constraints are stable under weak limits, i.e. if they hold for all $f^{k}$ and $h^{k}$, they also holds for the weak limit points $f'$ and $h$.
		
		\Cref{eq:FlowCons-Source,eq:FlowCons-Nodes,eq:RespCapacity} are all linear constraints and, thus, it is easy to see that they are stable under weak limits. We show this explicitly for \cref{eq:FlowCons-Source} and note that the proofs for the other two constraints are completely analogous. So, assume for contradiction that \cref{eq:FlowCons-Source} does not hold for $f'$ and $h$, i.e. assume that there is (wlog) some $i \in I, W \in  \Pc'_i, \varepsilon > 0$ and a set $A \subseteq [0,M]$ of positive measure such that for all $\theta \in A$ we have $f'^{W,+}_{i,1}(\theta) - h^W_i(\theta) > \varepsilon$. Then, we clearly have $\CharF[A] \in L^2([0,M])$ and (because \cref{eq:FlowCons-Source} holds for all $f^{k}$ and $h^{k}$ and they converge weakly to $f'$ and $h$, respectively):
			\[0 = \lim_k \langle \CharF[A], f'^{k,W,+}_{i,1}(\theta) - h^{k,W}_i(\theta)\rangle = \langle \CharF[A], f'^{W,+}_{i,1}(\theta) - h^W_i(\theta)\rangle \geq \varepsilon \mu(A) > 0, \]
		which is a contradiction.
		
		Next, \cref{eq:WeakFlowCons-Edges} is stable under weak limits by \Cref{claim:fToFisWSC}. Finally, to show that \cref{eq:FlowCons-Edges} is stable under weak limit we follow the proof of \cite[Lemma 5]{CominettiCL15}. From \Cref{claim:fToFisWSC,claim:FtoTisSSC} we know that the sequences $F^{k,W,+}_{i,j}$, $F^{k,W,-}_{i,j}$ and $T^{k}_e$ converge uniformly to $F'^{W,-}_{i,j}$, $F'^{W,-}_{i,j}$ and $T'_e$, respectively. From this, we directly get that
			\begin{align*}
				F'^{W,-}_{i,j}(T'_{e_j^W}(\theta)) \overset{\text{\Cref{lemma:PointwiseConvergenceOfConcatOfUniformConvergence}}}{=} \lim_k F^{k,W,-}_{i,j}(T^k_{e_j^W}(\theta)) \overset{\text{\eqref{eq:FlowCons-Edges}}}{=} \lim_k F^{k,W,+}_{i,j}(\theta) = F'^{W,+}_{i,j}(\theta)
			\end{align*}
		for any $\theta \in \IR_{\geq 0}$.
		
		From this we can now conclude that $f'$ is a network loading for $h$. However, by \Cref{lemma:UniqueNetworkLoading}, network loadings are unique and thus we have $f' = f$ almost everywhere. This, in turn, is now a contradiction to our initial assumption that $\langle g,f^{k}\rangle > \varepsilon$ for all $j \in \IN$ since we just showed, that $f^{k_j}$ has a subsequence which weakly converges to $f' = f$.
	\end{proofClaim}	
	Combining the three claims above implies the lemma.	
\end{proof}

\begin{lemma}\label{lem:weak-strong}
	For each $W\in \Pc'_i, i\in I$, the map 
	\[C \mapsto L^2([0,T]), h\mapsto \left([0,T] \to \IR, \theta \mapsto c_i\big(\mu_{i}^W(\theta),\sum_{e\in W}p_{i,e}\big)\right)\]
	is sequentially weak-strong continuous.
\end{lemma}
\begin{proof}
	From \Cref{lemma:PathInflowToArrivalTimeIsWSC} we can deduce that $C \mapsto C([0,T]), h \mapsto \ell_{i,0}^W$ is sequentially weak-strong continuous since it maps weakly convergent sequences to compositions of uniformly convergent sequences which, therefore, also converge uniformly. Furthermore, it is easy to see that a constant mapping like $C \mapsto C([0,T]), h \mapsto (\theta \mapsto \theta)$ is also sequentially weak-strong continuous. Thus, $C \mapsto C([0,T]), h \mapsto \mu^W_i$ is sequentially weak-strong continuous as difference of two such mappings. 
	
	Together with the continuity of $c_i$ this directly implies the lemma as follows: Let $h^k \overset{\text{w}}{\to} h$ be any weakly convergent sequence in $C$. We now have to show strong convergence of the image sequence in $L^2([0,T])$ (i.e. $L^2$-convergence). 
	
	We start by showing uniform convergence in $C([0,T])$. So, let $\varepsilon > 0$. As $c_i$ is uniformly continuous on $[0,M]$, there exists some $\delta > 0$ such that $\abs{c_i(x) - c_i(y)} \leq \varepsilon$ whenever $\abs{x-y} \leq \delta$. Furthermore, there exists some $K \in \IN$ such that for any $k \geq K$ we have $\norm{{\mu^W_i}^{(k)}-\mu^W_i}_\infty \leq \delta$ since the ${\mu^W_i}^{(k)}$ converge strongly to $\mu^W_i$ in $C([0,T],\IR_{\geq 0})$ (i.e. uniformly). This then implies that for every $k \geq K$ we have
	\begin{align*}
		&\norm{\left(c_i\big({\mu_{i}^W}^{(k)}(\_),\sum_{e\in W}p_{i,e}\big) - c_i\big(\mu_{i}^W(\_),\sum_{e\in W}p_{i,e}\big)\right)}_\infty \\
		\quad&= \sup_{\theta \in [0,T]}\left(c_i\big({\mu_{i}^W}^{(k)}(\theta),\sum_{e\in W}p_{i,e}\big) - c_i\big(\mu_{i}^W(\theta),\sum_{e\in W}p_{i,e}\big)\right) \leq \varepsilon.
	\end{align*}
	Thus, $C \mapsto L^2([0,T]), h\mapsto \left([0,T] \to \IR, \theta \mapsto c_i\big(\mu_{i}^W(\theta),\sum_{e\in W}p_{i,e}\big)\right)$ maps weakly convergent sequences to uniformly convergent sequences. Since $C([0,T]) \subseteq L^2([0,T])$ and uniform convergence implies $L^2$-convergence, this concludes the proof of this lemma.
\end{proof}

Using this lemma we can now finally show the existence of capacitated dynamic equilibria:

\begin{proof}[Proof of \Cref{thm:existence}]
	With \Cref{lem:weak-strong} we have that for each $W\in \Pc'_i, i\in I$, the map $h \mapsto c_i\left(\mu_{i}^W(\_),\sum_{e\in W}p_{i,e}\right)$	is weak-strong continuous from $C$ to $L^2([0,T])$. Taking the minimum of finitely many weak-strong continuous mappings results in a weak-strong continuous mapping
	and, finally, the difference of two weak-strong continuous  mappings is also weak-strong continuous.
	Thus, $\A$ is sequentially weak-strong-continuous from $C$ to $L^2([0,T])^{\Pc'}$. Applying \Cref{brezis} provides a solution $h^*$ for $\VI(C,\A)$. It remains to show that this is, in fact, a capacitated dynamic equilibrium. We will do this by contradiction, i.e. we assume that $h^*$ is not a capacitated dynamic equilibrium and show that in this case we get a new walk-flow $\bar h$ which contradicts \eqref{eqn:VI}. 
	
	\begin{claim}\label{claim:hStarNoCDE}
		If $h^*$ is not a capacitated dynamic equilibrium, then there exists a time $\bar\theta \geq 0$, a commodity $i$, two walks $W,Q \in \Pc'_{i}$ and three positive numbers $\varepsilon, \delta, \gamma > 0$ such that
		\begin{itemize}
			\item $H^{W \to Q}_i(h^*,\bar\theta,\varepsilon,\delta) \in S$,
			\item $c_i\left(\mu_{i}^W(\theta),\sum_{e\in W}p_{i,e}\right) - c_i\left(\mu_{i}^Q(\theta),\sum_{e\in Q}p_{i,e}\right) \geq \gamma$ for all $\theta \in [\bar\theta,\bar\theta+\delta]$ and
			\item $\int_{\bar\theta}^{\bar\theta+\delta}\min\set{{h^*}^W_i(\theta),\varepsilon}\diff\theta > 0$.
		\end{itemize}
	\end{claim}
	
	\begin{proofClaim}
		If $h^*$ is not a capacitated dynamic equilibrium then, by definition, there exists some commodity $i$, a walk $W \in \Pc'_i$ and a subset $J \subseteq [0,T]$ of positive measure such that for all $\bar\theta \in J$ we have ${h^*}^W_i(\bar\theta) > 0$ and there exists some $Q_{\bar\theta} \in D^W_i(h^*,\bar\theta)$ with 
		\begin{align}\label{eq:CDE-condition-violated}
			c_i\left(\mu_{i}^W(\bar\theta),\sum_{e\in W}p_{i,e}\right) > c_i\left(\mu_{i}^{Q_{\bar\theta}}(\bar\theta),\sum_{e\in Q_{\bar\theta}}p_{i,e}\right)
		\end{align}
		From the definition of $D^W_i(h^*,\bar\theta)$ we get for every such $\bar\theta$ some constants $\delta_{\bar\theta}, \varepsilon_{\bar\theta} > 0$ such that $H^{W\to Q_{\bar\theta}}_i(h^*,\bar\theta,\varepsilon_{\bar\theta},\delta_{\bar\theta}) \in S$. Since $\Pc'$ is a dominating set with respect to $S$ we can, wlog, assume that all $Q_{\bar\theta} \in \Pc'$. Furthermore, since $\Pc'$ is finite, there must be some $Q \in \Pc'$ such that we can restrict $J$ to only those $\bar\theta$ where we can choose $Q_{\bar\theta} = Q$ and still have that $J$ has positive measure. Finally, because $c_i\left(\mu_{i}^W(\theta),\sum_{e\in W}p_{i,e}\right)$ is continuous in $\theta$ we can wlog assume that each $\delta_{\bar\theta}$ is small enough such that \eqref{eq:CDE-condition-violated} holds for all $\theta \in [\bar\theta,\bar\theta+\delta_{\bar\theta}]$. 		
		
		Then, by \Cref{lemma:covering}, there exists some $\bar\theta \in J$ such that $J \cap [\bar\theta,\bar\theta+\delta_{\bar\theta}]$ still has positive measure. Consequently, the fact that we have $\min\set{{h^*}^W_i(\theta),\varepsilon_{\bar\theta}}>0$ for all $\theta \in J \cap [\bar\theta,\bar\theta+\delta_{\bar\theta}]$ implies $\int_{\bar\theta}^{\bar\theta+\delta_{\bar\theta}}\min\set{{h^*}^W_i(\theta),\varepsilon_{\bar\theta}}\diff\theta > 0$.
		Thus, setting $\varepsilon \coloneqq \varepsilon_{\bar\theta}, \delta \coloneqq \delta_{\bar\theta}$ and, finally, $\gamma \coloneqq \min\set{c_i\left(\mu_{i}^W(\theta),\sum_{e\in W}p_{i,e}\right) - c_i\left(\mu_{i}^Q(\theta),\sum_{e\in Q}p_{i,e}\right) | \theta \in [\bar\theta,\bar\theta+\delta]}$ gives us the desired objects.
	\end{proofClaim}

	Now, assume that a solution $h^*$ to \eqref{eqn:VI} is not a capacitated dynamic equilibrium. Then, using $\bar\theta, i, W, Q, \varepsilon, \delta, \gamma$ and $h'$ from \Cref{claim:hStarNoCDE} and setting $\bar h \coloneqq H^{W \to Q}_i(h^*,\bar\theta,\varepsilon,\delta)$ we clearly have $\bar h \in K$. Furthermore, $\bar h$ only uses walks that are already used in $h^*$ and additionally walk $Q$. Therefore all walks used by $\bar h$ are in $\Pc'$. Thus, we can conclude that $\bar h \in C$. But at the same time we also have
	\begin{align*}
		&\scalar{\A(h^*)}{\bar h-h^*} 
			= \int_0^T \scalar{\A(h^*(\theta))}{\bar h(\theta)-h^*(\theta)}\diff\theta\\
		&\quad\quad=\int_{\bar\theta}^{\bar\theta+\delta}\A(h^*)_W(\theta)\cdot\left(\bar h^W_i(\theta)-{h^*}^W_i(\theta)\right) + \A(h^*)_Q(\theta)\cdot\left(\bar h^Q_i(\theta)-{h^*}^Q_i(\theta)\right)\diff\theta \\
		&\quad\quad=\int_{\bar\theta}^{\bar\theta+\delta}\left(\A(h^*)_Q(\theta)-\A(h^*)_W(\theta)\right)\cdot\min\set{{h^*}^W_i(\theta),\varepsilon}\diff\theta \\
		&\quad\quad=\int_{\bar\theta}^{\bar\theta+\delta}\left(c_i\left(\mu_{i}^Q(\bar\theta),\sum_{e\in Q}p_{i,e}\right) - c_i\left(\mu_{i}^W(\bar\theta),\sum_{e\in W}p_{i,e}\right)\right)\cdot\min\set{{h^*}^W_i(\theta),\varepsilon}\diff\theta\\
		&\quad\quad\leq -\gamma \int_{\bar\theta}^{\bar\theta+\delta}\min\set{{h^*}^W_i(\theta),\varepsilon}\diff\theta < 0,
	\end{align*}
	which is a contradiction to $h^*$ being a solution to \eqref{eqn:VI}. Therefore, any solution $h^* \in C$ to \eqref{eqn:VI} is also a capacitated dynamic equilibrium and, in particular there always exists a capacitated dynamic equilibrium.
\end{proof}

\subsection{Special Cases}

We discuss two special cases for which our existence theorem can be applied by suitable choices of the abstract restriction set $S$: dynamic equilibria and energy-feasible dynamic equilibria.

\subsubsection{Dynamic Equilibria}

If we choose $S = L^2([0,T])^{\Pc}$, then capacitated dynamic equilibria are exactly the dynamic equilibria as defined in~\cite{CominettiCL15,Friesz93,Koch11,ZhuM00,MeunierW10}. To see this, note, that in this case we always have $D^W_i(h,\bar\theta) = \Pc_i$. Thus, \eqref{eq:CDE} translates to the constraint that whenever there is positive inflow into some walk $W$, this walk has to be a shortest walk at that time. 
Since dynamic flows in the Vickrey-model satisfy FIFO, the set of simple paths is a dominating set for the set of all walks with respect to $S = L^2([0,T])^{\Pc}$ (i.e. removing a cycle from a walk can never increase its aggregated cost). As the set of simple paths is clearly finite, one can use~\Cref{thm:existence} to show existence of dynamic equilibria. Note that the classical existence proofs for dynamic equilibria (e.g. by Han et. al.~\cite{Han2013} or Cominetti et. al.~\cite{CominettiCL15}) usually have the restriction to simple paths as part of the model itself, i.e. they only allow walk-flows from $L^2([0,T])^{\Pc'}$ where $\Pc'$ is the set of simple source-sink paths.

\subsubsection{Energy-Feasible Dynamic Equilibria}

Now let us turn to the case of energy-feasible dynamic equilibria, i.e. equilibria of flows in battery-extended networks. We show that~\Cref{thm:existence} implies the existence of energy-feasible dynamic equilibria.

\begin{theorem}
	Let $\network$ be any battery-extended network network and $S:=\iota(L^2([0,T])^{\Pc_b}) \subseteq L^2([0,T])^{\Pc}$. Then, there exists an energy-feasible dynamic equilibrium in $\network$, i.e. a capacitated dynamic equilibrium with respect to $S$.
\end{theorem}

\begin{proof}
	First, it is quite obvious that $S$ is closed and convex and has non-empty intersection with $K$ (using our assumption that every commodity has at least one energy-feasible source-sink walk).
	For the existence of a finite dominating walk set, we will show that due to the FIFO condition in the Vickrey model, there exists a constant $\kappa_i$ such that for every agent playing against any walk choices of all other agents there exists an optimal strategy which enters any (recharging) node at most $\kappa_i$ times.
	To see this let us define the following quantities
	 \[ \kappa_{i}:=\max\left\{\frac{b_i^{\max}}{\alpha_{i}}\right\},
	\text{ where } 
	 \alpha_i:=\min_{E'\subseteq E}\left\{\sum_{e\in E'} b_{i,e}\vert \sum_{e\in E'} b_{i,e}>0\right\}.\]
	The quantity $\alpha_i$ is a lower bound on the minimum positive increment for $i\in I$ along any simple cycle.
	 Suppose there is some  node $v$, which is visited $k\in \IN$ times by a walk $W$ of commodity $i$. By renaming indices, we can assume that  $v$ appears in $W$
	 in the order $v_1,\dots,v_k$ with $v_j=v, j\in [k] $. Clearly, whenever we have $b_W(v_{\ell}) \geq b_W(v_j)$ for some $\ell < j$, we can delete the cycles between $v_{\ell}$ and $v_j$ to obtain another energy-feasible walk $W'$ of the same commodity. Due to FIFO and the fact that the aggregation function $c_i$ is non-decreasing, the new walk $W'$ then has at most the same aggregated cost as $W$. Thus, commodity $i$ always has an optimal walk where the sequence $b_W(v_1)<\cdots<b_W(v_k)$ is monotonically increasing with increments of at least $\alpha_i>0$. With $b_W(v_k)\leq b_i^{\max}$, we get $k\leq \kappa_{i}$ as wanted. To explicitly construct a  finite dominating set $\Pc'$,  for a walk  $W\in \Pc_i$, we
	define $\psi_W(v):=|\{v_j\in \hat W\vert v_j=v, j\in [k_W]\}|$.
	Recall that $\hat W$ is the node-multiset representation of $W$.
	Then, a finite dominating set is given as
	\[ \Pc':=\{(i,W)\vert W\in \Pc'_i, i\in I\} \text{ with } \Pc'_i:=\{W\in \Pc_i\vert \psi_W(v)\leq \kappa_{i} \text{ for all }v\in \hat W\}, i\in I.\]
	Thus, all conditions of~\Cref{thm:existence} are satisfied and we obtain the existence of  
	an energy-feasible dynamic equilibrium.
\end{proof}

%% file: chapters/ComputationalStudy.tex
\section{Computational Study}
In this section, we focus on computing energy-feasible dynamic equilibria on a set of
moderate sized networks. We discretize the continuous time scale and then use a fixed point algorithm similar to the one used by Han et al. in \cite{han2019computing} to compute walk-flows in which agents only use walks which are close to the least expensive (with regards to total costs) energy-feasible walks.

\subsection{A Fixed Point Algorithm}

The main steps of our algorithm are as follows: First we compute the set of energy-feasible walks $\Pc_{b}$ as well as one initial walk-flow $h \in K$ using only those walks (e.g. by sending the whole flow volume of any commodity $i$ along a physically shortest walk in $\Pc_{i,b}$). Next, we use the network loading procedure from the proof of \Cref{lemma:UniqueNetworkLoading} to determine the feasible flow $f$ associated with $h$ and derive from this the total travel times $\mu^W_i$ for all walks and commodities. Finally, we update the walk-flow $h$ by shifting flow from walks with high total cost to walks with lower total cost, whereby the amount of flow shifted is proportional to the difference in cost between the two walks, and again compute the network loading for the updated walk-flow. We then repeat this update process until we reach a walk-flow $h$ that changes only very little during the update-step or, equivalently, where the walks used in $h$ all have costs very close to the minimal cost under the current walk-flow $h$.

More formally, the update is computed as follows: Given a walk-flow $h^k$ we denote the costs determined by the associated feasible flow by $c^h_{W,i}(\theta) \coloneqq c_i\left(\mu^W_i(\theta),\sum_{e \in W}p_{i,e}\right)$ for all commodities $i \in I$, energy-feasible walks $W \in \Pc_{i,b}$ and times $\theta \in [0,T]$. We then want to compute functions $v_i: \mathbb{R}_{\geq 0} \to \mathbb{R}$ for each $i \in I$ such that
\begin{align}\label{FP-update}
	\sum_{W \in \Pc_{i,b}} \left[h^{k,W}_{i}(\theta) - \alpha^k \cdot
	c^h_{W,i}(\theta) + v_i(\theta)\right]_+ = u_i(\theta), 
	\quad \text{ for all } \theta \in [0,T] \tag{FP-Update},
\end{align}
where $\alpha^k > 0$ is the current step size, and obtain the new walk-flow $h^{k+1}$ by setting $$h^{k+1,W}_{i}(\theta) := \left[h^{k,W}_{i}(\theta)
- \alpha^k \cdot c^h_{W,i}(\theta) + v_i(\theta)\right]_+.$$
We observe that fixed points of this update procedure are exactly the energy-feasible dynamic equilibria:
\begin{lemma}
	$h^k$ corresponds to a energy-feasible dynamic equilibrium if and only if $h^{k+1} = h^k$ almost everywhere.
\end{lemma}

\begin{proof}
	First note, that as explained after \Cref{def:DCE} a walk-flow $h^k$ corresponds to an energy-feasible dynamic equilibrium if and only if for all $i \in I$ and almost all $\theta \in [0,T]$
	\[h^{k,W}_{i}(\theta) > 0 \implies c^h_{W,i}(\theta) = \min\{c^h_{W',i}(\theta)
	\,|\, W' \in \Pc_{i,b} \}.\]
	
	Now assume that $h^{k+1} = h^k$, i.e. $h^{k,W}_{i}(\theta) := \left[h^{k,W}_{i}(\theta)
	- \alpha^k \cdot c^h_{W,i}(\theta) + v_i(\theta)\right]_+$ for
	all commodities $i$, walks $W \in \Pc_{i,b}$ and times $\theta$. This is equivalent to
	\[\max\left\{v_i(\theta)-\alpha^k c^h_{W,i}(\theta),-h^{k,W}_{i}(\theta)\right\} = 0.\]
	From this, a direct computation shows that $c^h_{W,i}(\theta) \geq
	\frac{v_i(\theta)}{\alpha^k}$ for all walks $W \in \Pc_{i,b}$ with $h^{k,W}_{i}(\theta) > 0$
	implies $c^h_{W,i}(\theta) = \frac{v_i(\theta)}{\alpha^k}$. Thus, $h^k$
	corresponds to a energy-feasible dynamic equilibrium. The other direction of the proof is completely analogous.
\end{proof}

\begin{algorithm2e}
	\caption{A Fixed-point algorithm for computing energy-feasible dynamic equilibrium flows}\label{algo:fixedPoint}
	\KwIn{A battery-expanded network $\network$, constants $N \in \IN$, $\alpha^0>0$}
	Compute the set of energy-feasible walks $\Pc_{i,b}$ for each commodity $i \in I$ \\
	Pick some initial walk-flow $h^0$ and set the iteration count $k \leftarrow 0$.\\
	\Repeat{$\frac{\norm{h^{k+1} - h^k}}{\norm{h^k}} \leq \varepsilon$}{
		Calculate the network loading for $h^k$ and determine the values $c^{h,j}_{W,i}$. \\
		Find values $v_i^j$ satisfying
		\begin{align}
			\sum_{W \in \Pc_{i,b}} \left[\bar h^{k,W,j}_{i} - \alpha^k \cdot
			\bar c^{h,j}_{W,i} + v^j_i\right]_+ = u^j_i \text{ for all } i \in I, j \in {1, \dots, N} \tag{FP-Update}
		\end{align}	
		Set $h^{k,W,j+1}_{i} \coloneqq \left[\bar h^{k,W,j}_{i} - \alpha^k \cdot
		\bar c^{h,j}_{W,i} + v^j_i\right]_+$, $\alpha^{k+1} \coloneqq s(\alpha^k,h^k,h^{k+1})$ and increment $k$.
	}
\end{algorithm2e}

Now, in order to be able to compute the network loading, the travel cost functions and updates we will use discretized time instead of the continuous time of our formal model. Therefore, we split the time interval $[0,T]$ into $N$ parts using equally spaced breakpoints $a_0,a_1,\ldots,a_N$ and then only consider walk-flows which are constant on each of the intervals $[a_{j-1},a_{j}]$, leading in turn to an associated feasible flow with piecewise constant flow rates as well (though, possibly, over smaller intervals). This allows us to exactly compute the network loading. For the update step we then also approximate the network inflow rates $u_i$ for each interval by a single value $\bar u_i^j \coloneqq \int_{a_{j-1}}^{a_j} u_i(\theta) d\theta$ and the cost function $c^h_{W,i}(\theta)$ by the value attained in the middle of each interval, i.e. by $\bar c^{h,j}_{W,i} \coloneqq c^h_{W,i}((a_j+a_{j-1})/2)$. Finally, we denote by $\bar h^{k,W,j}_i$ the value of $h^{k,W}_i$ within the interval $[a_{j-1},a_j]$. This then results in the following discretized version of \eqref{FP-update}: Find values $v_i^j$ satisfying
\begin{align}\label{FP-update-disc}
	\sum_{W \in \Pc_{i,b}} \left[\bar h^{k,W,j}_{i} - \alpha^k \cdot
	\bar c^{h,j}_{W,i} + v^j_i\right]_+ = u^j_i \text{ for all } i \in I, j \in {1, \dots, N}. \tag{FP-Update}
\end{align}
The new walk-flow $h^{k+1}$ is then defined by setting 
$$\bar h^{k+1,W,j}_{i} \coloneqq \left[\bar h^{k,W,j}_{i} - \alpha^k \cdot
\bar c^{h,j}_{W,i} + v^j_i\right]_+ \text{ for all } i \in I, j \in {1, \dots, N}.$$ 
To solve \eqref{FP-update-disc} for $v^j_{i}$, standard root finding algorithms can be used.
We use the Newton's method (available in the Scipy package of Python).

Our complete fixed point algorithm, thus, has the form depicted in \algoref{fixedPoint}.
Here, $s(.)$ represents a function that updates the step size $\alpha$ at the end of each
iteration. In our preliminary experiments, we observed that it helps the algorithm to convergence
faster if $\alpha$ is updated dynamically based on the values of $\alpha^k$,
$h^{k+1}$ and $h^{k}$. We compute a parameter $\gamma^{k+1} := 1 - ({\lvert \lvert {h^{k+1} -
h^k}}\rvert \rvert)/({\lvert \lvert{h^{k+1} + h^k}}\rvert \rvert)$ and set
$\alpha^{k+1} = \gamma^{k+1}(\gamma^{k+1}\alpha^{k}) + (1-\gamma^{k+1})\alpha^{k}$.

\subsection{Data sets}
We illustrate the performance of \algoref{fixedPoint} first on a small 
example and then on two benchmarking instances from 
the literature (cf. e.g. \cite{Han2013}), namely the Nguyen network and the Sioux Falls network.
\tabref{datasets} describes the characteristics of these networks.
\begin{table}[htbp]
\caption{Description of test instances used for the computational study. The variant
A has no energy feasibility constraints while the variant B involves energy
constraints. The variant C includes recharging stations/edges (indicated separately in
the `\# edges' row on the right side of the + sign) in addition to energy constraints.}
\scriptsize
\begin{tabular}{lrrrrrrrrr} \toprule
Network & \multicolumn{3}{c}{Example1} & \multicolumn{3}{c}{Nguyen} &
\multicolumn{3}{c}{Sioux Falls} \\ \cmidrule(lr){2-4}\cmidrule(lr){5-7}\cmidrule(lr){8-10}
& A & B & C & A & B & C & A & B & C \\
\cmidrule(lr){2-2}\cmidrule(lr){3-3}\cmidrule(lr){4-4}\cmidrule(lr){5-5}\cmidrule(lr){6-6}
\cmidrule(lr){7-7}\cmidrule(lr){8-8}\cmidrule(lr){9-9}\cmidrule(lr){10-10}
\# nodes & 4 & 4 & 4 & 13 & 13 & 13 & 24 & 24 & 24 \\
\# edges & 5 & 5 & 5+2 & 19 & 19 & 19+3 & 76 & 76 & 76+1 \\
\# commodities & 1 & 1 & 1 & 4-20 & 4-20 & 4-20 & 4 & 4 & 4 \\
Inflow rate (u) & 3 & 3 & 3 & 3 & 3 & 3 & 3 & 3 & 3 \\
Time horizon & [0,10] & [0,10] & [0,10] & [0, 300] & [0, 300] & [0, 300] & [0, 480] &
[0, 480] & [0, 480] \\
Total inflow & 30 & 30 & 30 & 3600-18000 & 3600-18000 & 3600-18000 & 5760 & 5760 & 5760 \\
Energy cons.? & No & Yes & Yes & No & Yes & Yes & No & Yes & Yes \\
Energy values & -- & 0-4 & 0-4 & -- & 0.5-4 & 0.5-4 & -- & 1-6 & 1-6 \\
\# r/c stations & -- & 0 & 1 & -- & 0 & 3 & -- & 0 & 1 \\
Price for r/c & -- & -- & 0 & -- & -- & 0-60 & -- & -- & 0 \\ \bottomrule
\end{tabular}
\label{tab:datasets}
\end{table}
The edge travel times and edge capacities for the illustrative example are shown in
\figref{example6WithRecharge}. The corresponding values for the Nguyen and the Sioux
Falls network have been taken from \cite{bstabler}. We have chosen the energy consumption
values for edges based on their travel time and involvement in different walks from the
range shown in \tabref{datasets} in the row `Energy values'. The price for
non-recharging edges is set to $0$ for all networks.

\footnote{The source code of our implementation as well as the used input files can be found at \url{https://github.com/ArbeitsgruppeTobiasHarks/electric-vehicles}}

\subsection{Performance Measures}
We assess the performance of \algoref{fixedPoint} based on the following measures:
In order to quantify how well a walk-flow at an iteration $k$ satisfies the dynamic
equilibrium conditions \eqref{def:DCE}, we use the following measure
\begin{align}\label{eq:qopidef}
QoPI_{k} &:= \sum_i \left ( \frac{\mathlarger{\sum}_{W \in \Pc_{i,b}} \mathlarger{\int}_{0}^{T} h^{k,W}_{i}(\theta) \left (\frac{c^h_{W,i}(\theta) - \min_{W' \in
\Pc_{i,b}}\{c^h_{W',i}(\theta)\}}{\min_{W' \in
\Pc_{i,b}}\{c^h_{W',i}(\theta)\}}\right) \; d\theta}{\mathlarger{\int}_{0}^{T} u_i(\theta) \;
d\theta} \right ).
\end{align}
If we calculate this value without dividing by the total inflow volume, we call the resulting value QoPI (absolute).
It is easy to verify that a walk-flow $h$ is an energy-feasible dynamic equilibrium flow if and only if QoPI is
zero. A value close to zero means that agents are using walks whose costs are quite close to the minimum costs at any
given time $\theta$.
The integral value in the numerator of \eqnref{qopidef} is approximated by the area
under a piecewise linear function obtained by connecting the corresponding values at the
discrete time points $(a_{j-1} + a_j)/2$. 
We also measure the absolute change in the $L^1$-norm of walk-flows ($\Delta h^{k+1}:=
\norm{h^{k+1} - h^k}$) and the relative change ($\Delta h^{k+1}/\norm{h^k}$). 

Furthermore, we show the energy consumption profile for a walk-flow as a time-plot of the measure
\begin{align}\label{eq:energyProf}
\eta(\theta) &:= \sum_i \sum_W \frac{h^W_i(\theta) b_W}{u_i(\theta)},
\end{align}
where $b_W = \sum_{e \in W} b_e$ represents the energy consumption corresponding to a
walk $W$. The value $\eta(\theta)$, thus, indicate the average energy consumed by all particles starting their travel at time $\theta$.
We also use the following measures of energy consumption for each commodity for each
time $\theta$.
\begin{itemize}
\item Minimum energy consumption := $\min_{W: h^W_i(\theta) > 0} {b_{W}}$
\item Maximum energy consumption := $\max_{W: h^W_i(\theta) > 0} {b_{W}}$
\item Mean energy consumption := $\sum_W \frac{h^W_i(\theta) b_W}{u_i(\theta)}$
\end{itemize}
For multicommodity networks, we use the following measures of walk travel times
consolidated over the commodities.
\begin{itemize}
\item Minimum walk travel times := $\min_{i \in I} {\mu^W_i(\theta)}$
\item Maximum walk travel times := $\max_{i \in I} {\mu^W_i(\theta)}$
\item Mean walk travel times := $\sum_{i \in I} \frac{\sum_{W:h^W_i(\theta)>0}\mu^W_i(\theta)}{\lvert \{W:h^W_i(\theta)>0\} \rvert}$
\item Mean of minimum walk travel times := $\frac{\sum_{i \in I} \min_{W:h^W_i(\theta)>0}\mu^W_i(\theta)}{\lvert I \rvert}$
\item Mean of maximum walk travel times := $\frac{\sum_{i \in I} \max_{W:h^W_i(\theta)>0}\mu^W_i(\theta)}{\lvert I \rvert}$
\end{itemize}

\subsection{Computational Results}
For simplicity, we assume in our computational study that the batteries are
recharged fully up to $b_i^{max}$ at all recharging stations. 
For experiments with a positive price of recharging, we use an alternate aggregation
function defined as
\begin{align}
c_i\left(\mu_{i}^W(\theta),\sum_{e\in W}p_{i,e}\right) &:= \mu_{i}^W(\theta) +
\widetilde{\lambda}_i\sum_{e\in W}p_{i,e},
\end{align}
to avoid numerical instabilities in our algorithm. Here, the parameter
$\widetilde{\lambda}_i \geq 0$ (used instead of the parameter $\lambda_i$ as defined
in Example~\ref{example:AggregationFunction}) converts the recharging price to travel time.
Also, we terminate the
algorithm based on the criterion $\Delta h < \epsilon$ which is typically more strict
than the one mentioned in \algoref{fixedPoint}, however, we use it for the
computational study along with a soft maximum time limit and a maximum
iteration limit to obtain the best possible walk-flows. 
The computations have been carried out on a 64-bit
Intel(R) Xeon(R) E5-2670 v2, 2.50GHz CPU with 128GB RAM. The source code has been
written in Python 3.5.

We will now describe the results of our tests on our three test networks: The small toy network, the Nguyen network and the Sioux Falls network.

\subsubsection{Results for a Toy Instance}

First, we analyze the performance of \algoref{fixedPoint} on variants of the toy instance Example1
depicted in \figref{example6WithRecharge}.
Example1-A is without energy constraints, Example1-B is with energy constraints but without recharging stations and Example1-C incorporates energy constraints and recharging stations.
There are four
$s,t$-walks for Example1-A. When energy constraints
are introduced (resulting in Example1-B), the walk $W_0:=(e_1, e_3, e_4)$ becomes infeasible.
Example1-C includes a recharging station at the node $v$ with two different modes of
recharge, $m_1$ and $m_2$, represented using green coloured self-loops, each with a single
time-duration and a single price for recharging. $p^{max}$ denotes the price-budget of
an agent. As shown, a value 6 for $p^{max}$ rules out recharging at $v$ via mode $m_2$. 
This results in seven energy-feasible walks
for this network: $W_0$ with recharging at $v$, three walks from Example1-B, and
the same three walks but with recharging at $v$. However, an agent will not travel
via these latter walks with recharging due to a positive travel time of recharging edges.

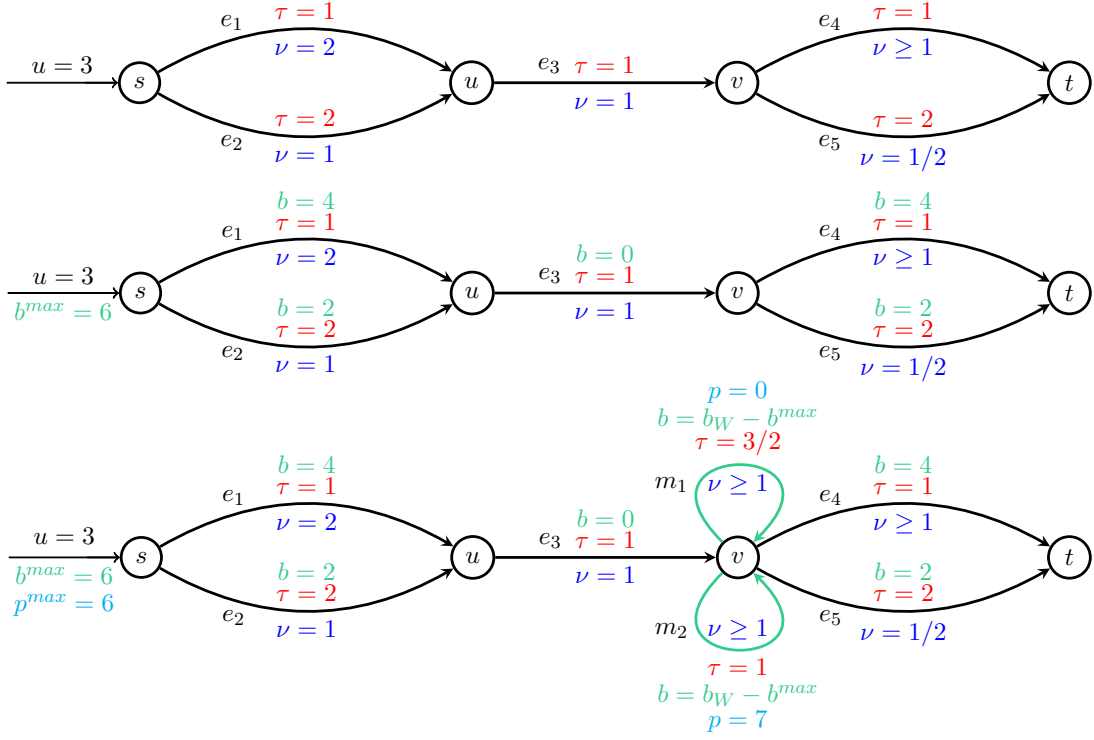
\begin{figure}[h!]
	\centering
	\resizebox{.9\linewidth}{!} {\input{tikz/example6NBC.tex}}
	\vspace{-0.25cm}
	\resizebox{.9\linewidth}{!} {\input{tikz/example6.tex}}
	\resizebox{.9\linewidth}{!} {\input{tikz/example6WRModes.tex}}
	\caption{Pictorial depiction of variants of Example1-A (top), -B (middle) and
		-C (bottom) and the corresponding travel time ($\tau$), capacity ($\nu$), and energy
		consumption values ($b$) for different edges.}
	\label{fig:example6WithRecharge}
\end{figure}

\paragraph*{Convergence measures.} 
\tabref{example6due} shows the corresponding QoPI values upon termination of \algoref{fixedPoint} for all three instances. As these values are quite small, this indicates that the eventual flow largely uses walks with the least possible cost. The last column of \tabref{example6due} indicates the walks with a positive flow at termination. 
\figref{example6Results} provides a more detailed picture of the flow determined in the final iteration of \algoref{fixedPoint} for Example1-A, Example1-B and Example1-C, respectively. The first row shows the inflow rates in the different walks while the second row shows the total travel costs along each walk (i.e. the value of the aggregation function $g_i$). Note that for these instances the
the total costs equal the travel time since the costs at the only used recharging station are zero (via mode $m_1$) here. 
Comparing the plots in the first and second row shows that whenever there is positive inflow into a walk, this walk has minimal total cost among all feasible walks in accordance to the definition of capacitated dynamic equilibrium. The QoPI values for each walk presented in the third row confirm this as they are also quite close to zero at all times.

\begin{table}[htbp]
	\caption{Number of feasible walks ($\lvert \mathcal{W} \rvert$), QoPI attained at
		dynamic equilibrium, change in norm of walk inflows at termination ($\Delta
		h$) and walks with positive inflow rates at dynamic equilibrium on variants of
		Example1. The following parameter settings are used: $\epsilon =
		0.01$, $\alpha^0 = 0.5$.}
	\centering
	\begin{tabular}{lcccl}\toprule
		Network & $\lvert \mathcal{W} \rvert$ & QoPI & $\Delta h$ & Walks with
		$h^+_W > 0$ at termination \\ \hline 
		Example1-A    & 4 & $1.2e^{-4}$ & 0.001 & $W_0:=(e_1, e_3, e_4)$, $W_2:=(e_2,e_3,e_4)$ \\ 
		Example1-B & 3 & $2.4e^{-4}$ & 0.001 & $W_1:=(e_1, e_3, e_5)$, $W_2$ \\ 
		Example1-C & 7 & $3.7e^{-4}$ & 0.001 & $W_1$, $W_2$, $W_4:=(e_1,e_3,m_1,e_4)$ \\ \bottomrule
	\end{tabular}
	\label{tab:example6due}
\end{table}

\begin{figure}[h!]
	\centering
	\includegraphics[width=0.33\textwidth]{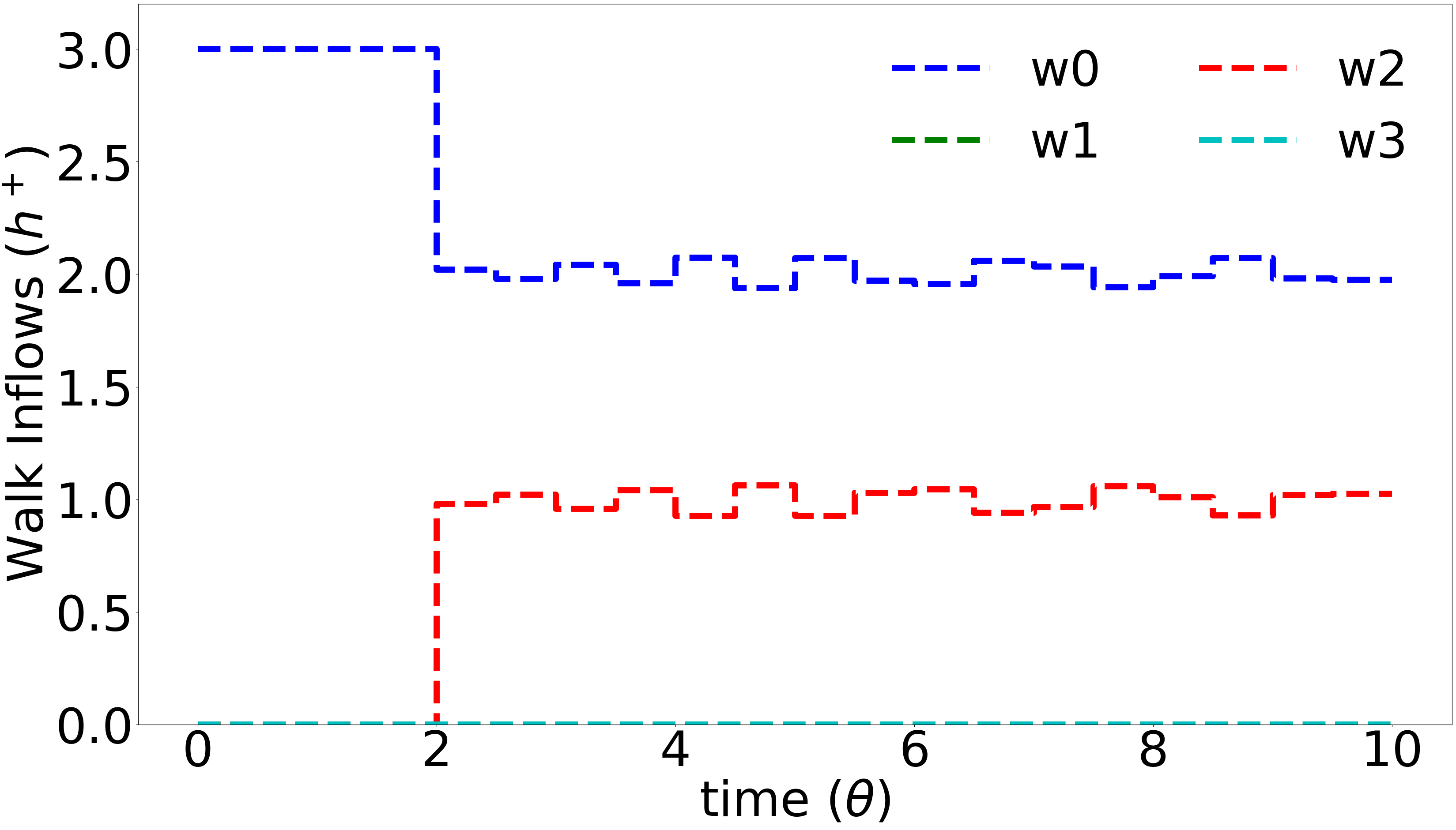}
	\includegraphics[width=0.32\textwidth]{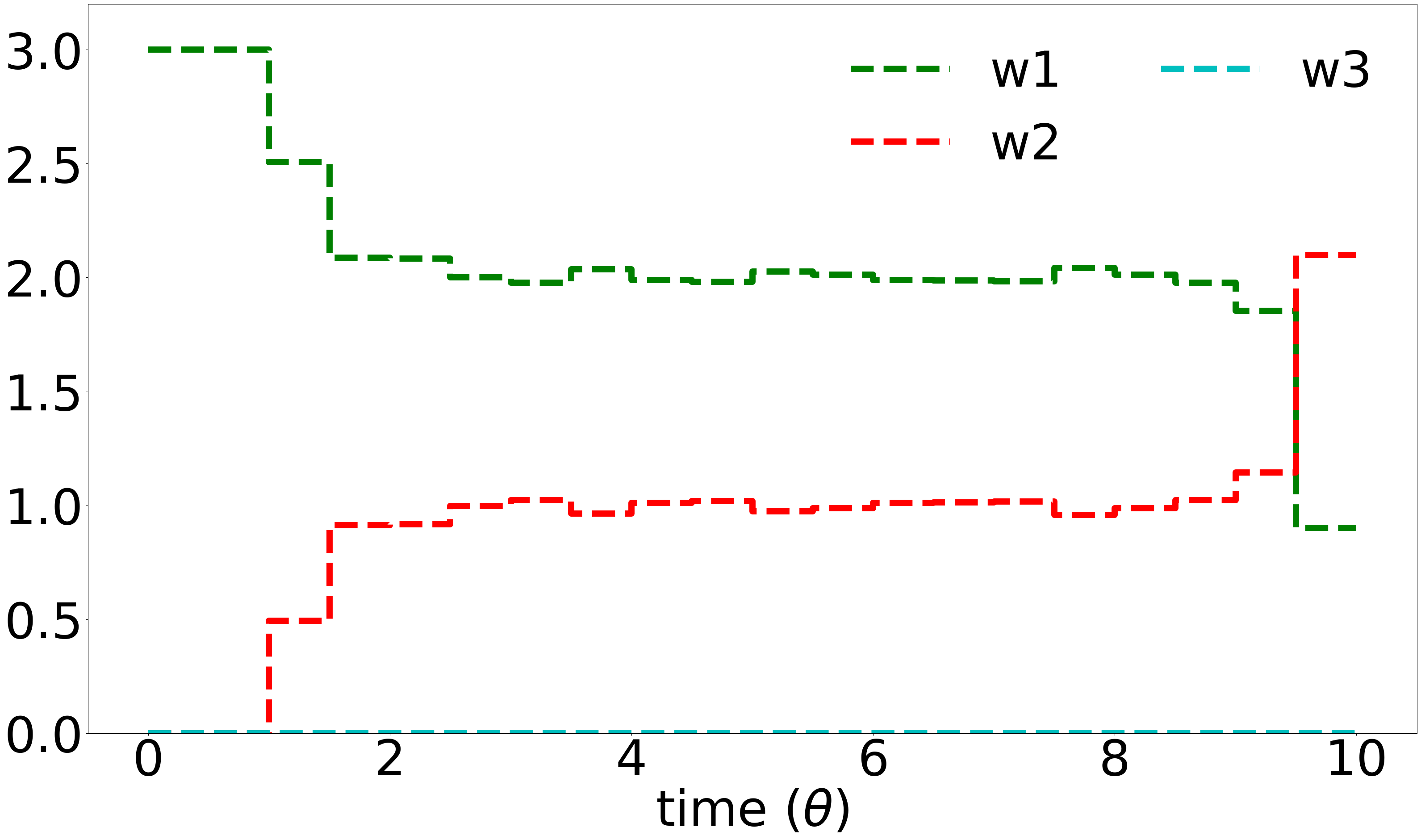}
	\includegraphics[width=0.32\textwidth]{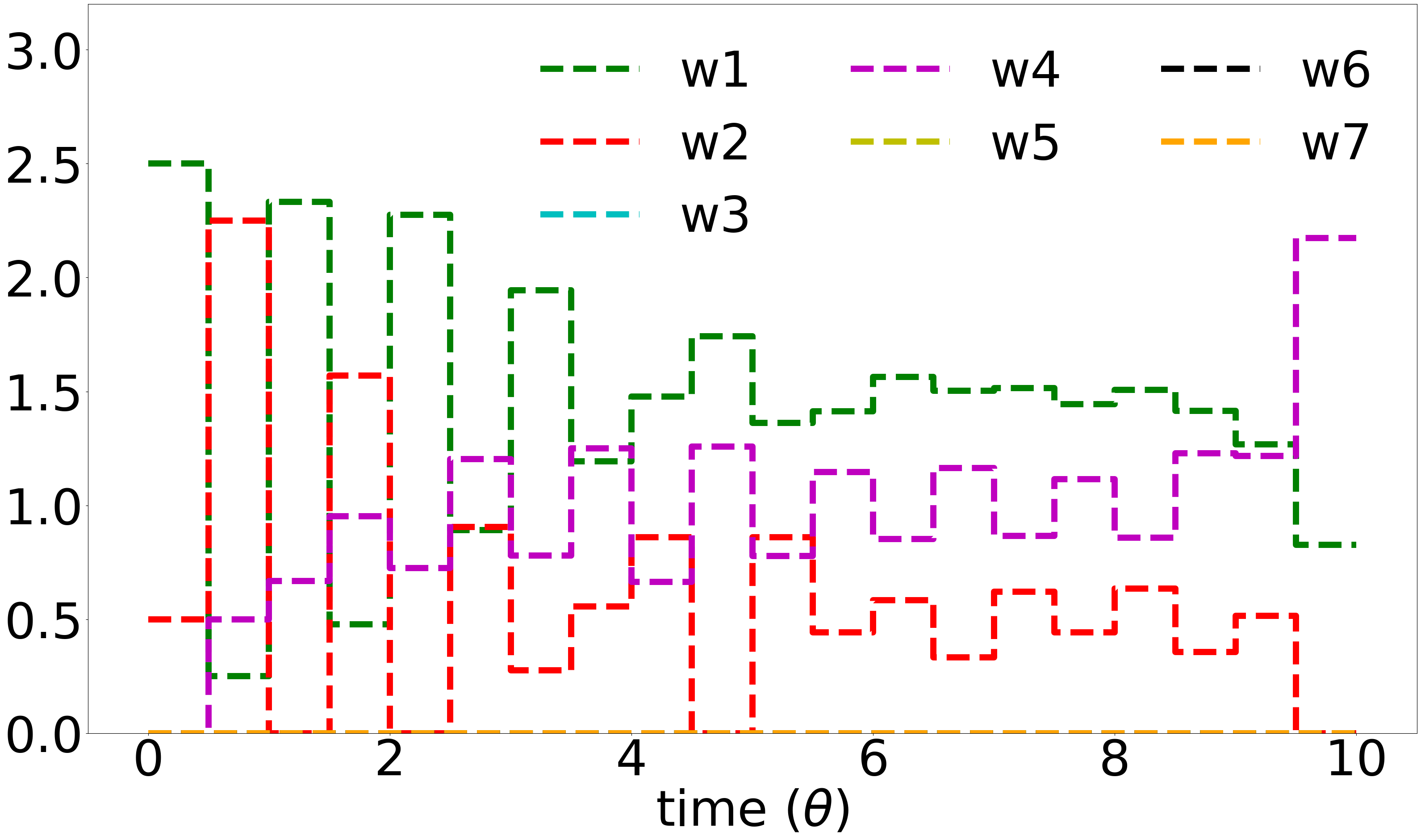}
	\includegraphics[width=0.33\textwidth]{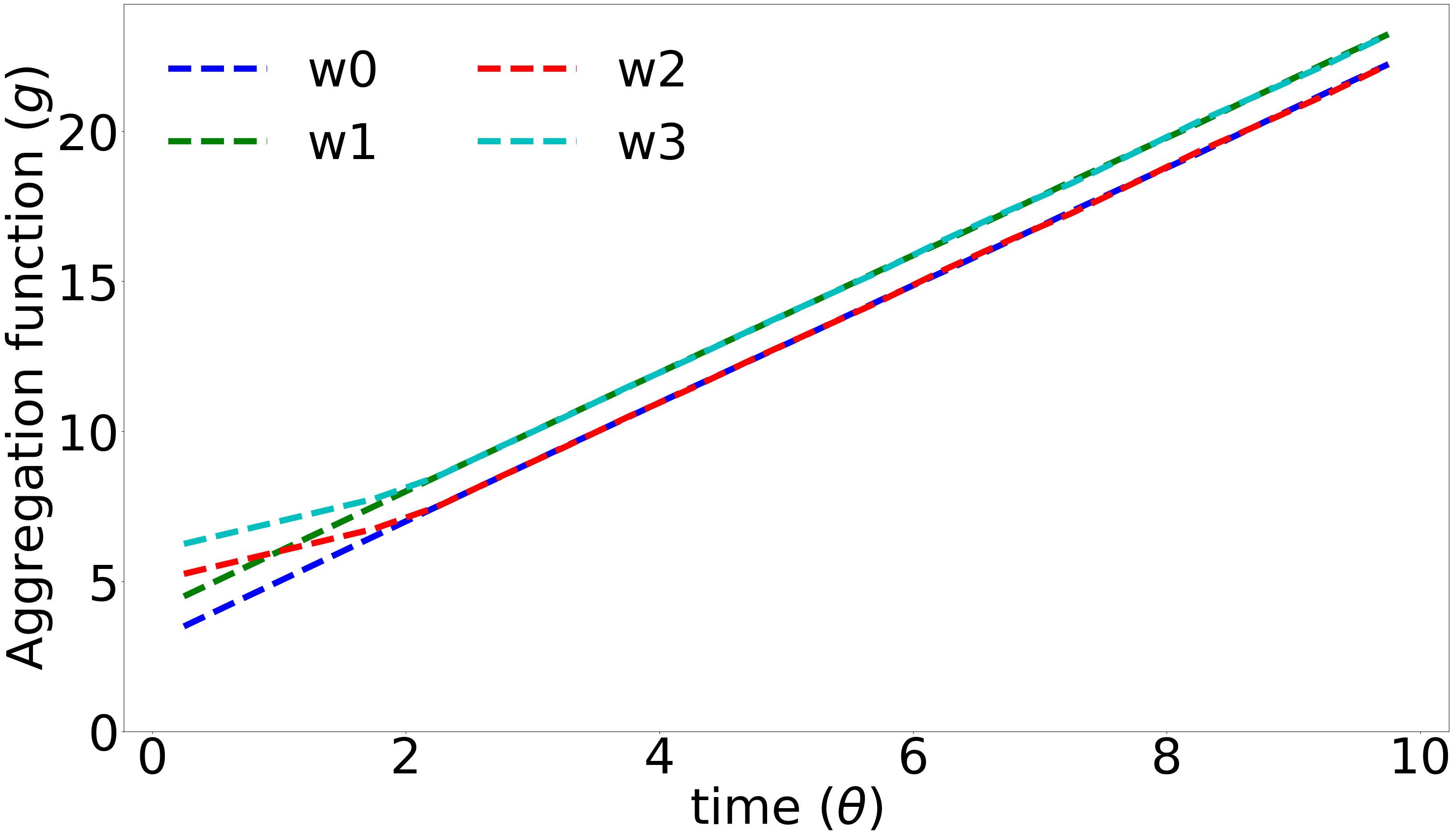}
	\includegraphics[width=0.32\textwidth]{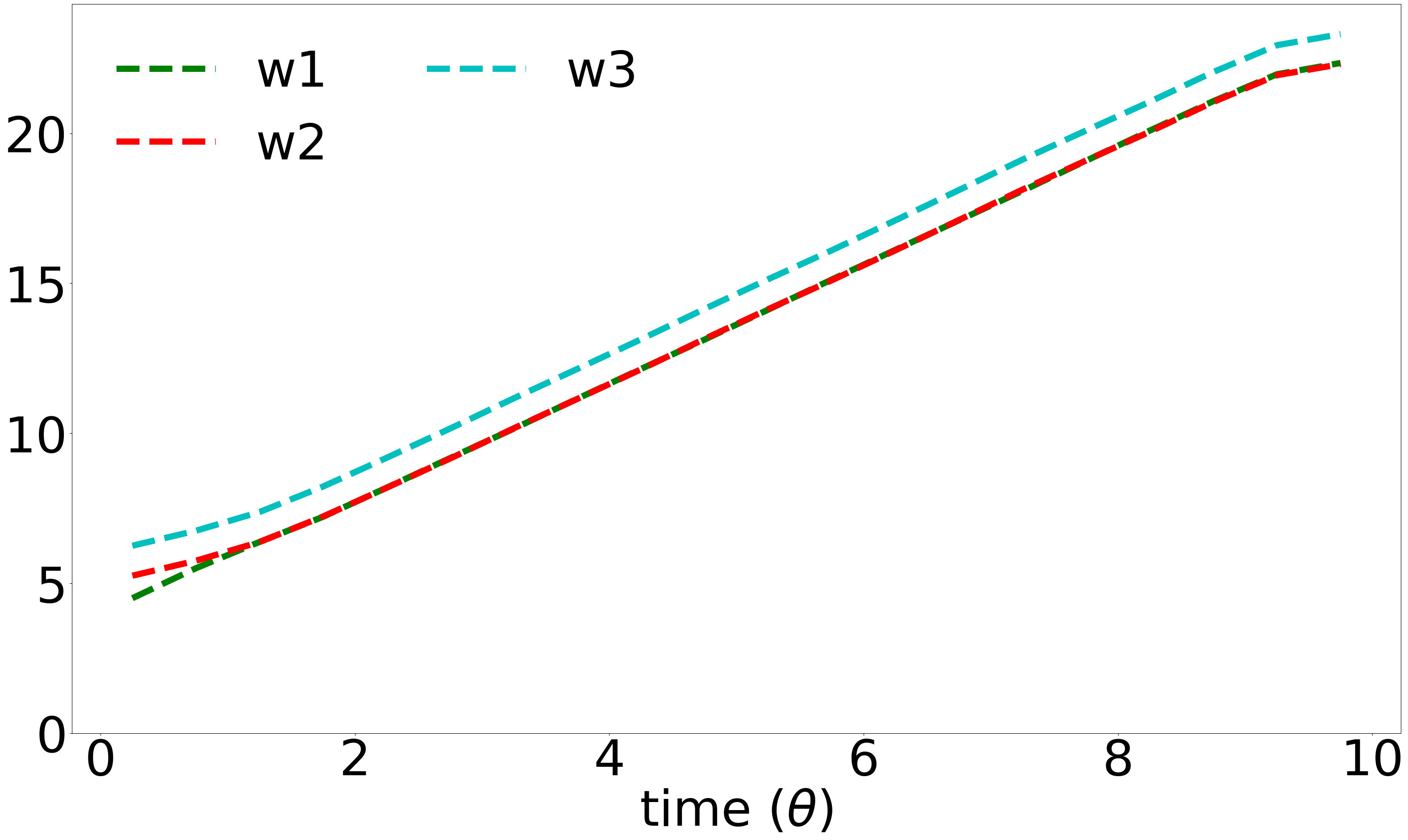}
	\includegraphics[width=0.32\textwidth]{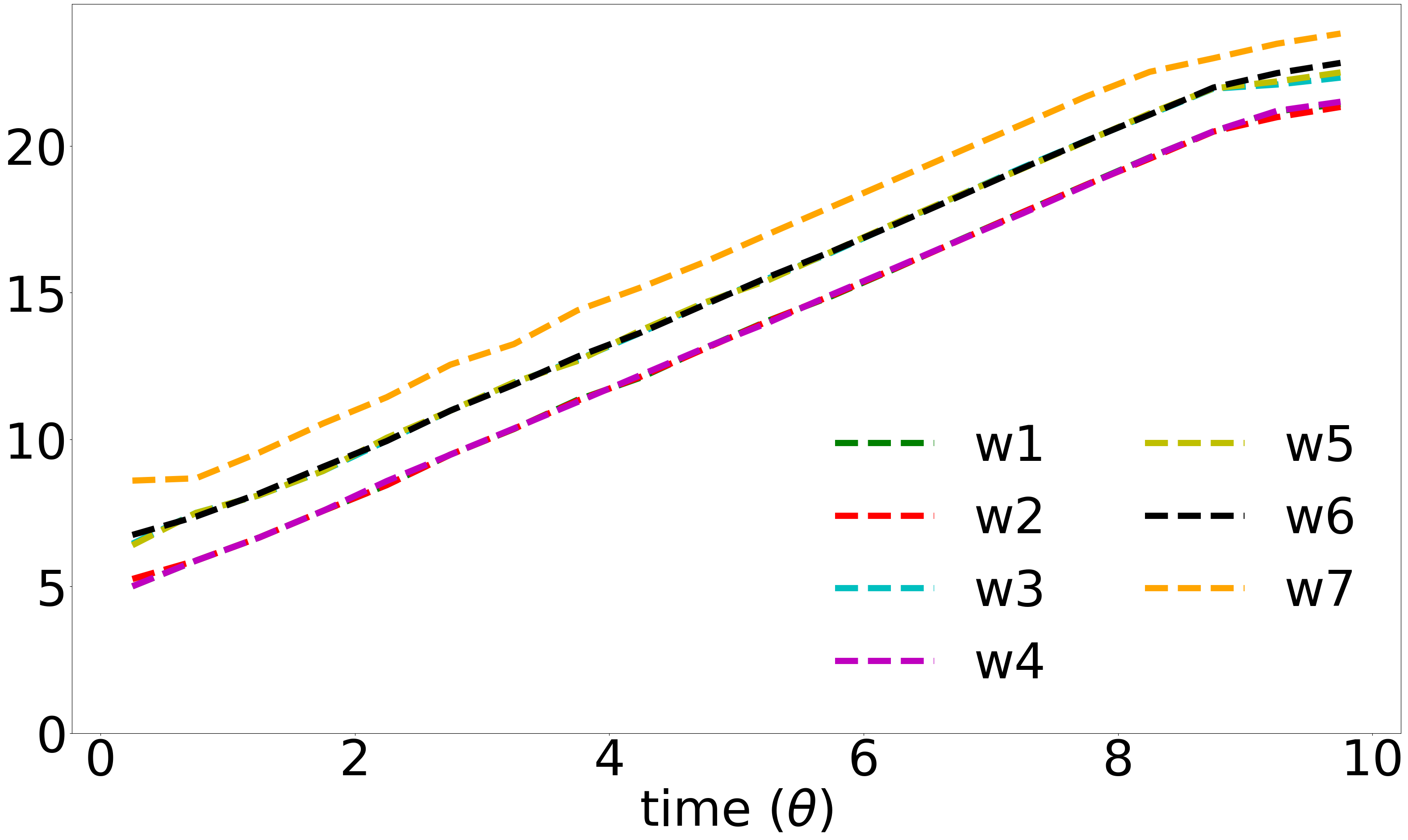}
	\includegraphics[width=0.33\textwidth]{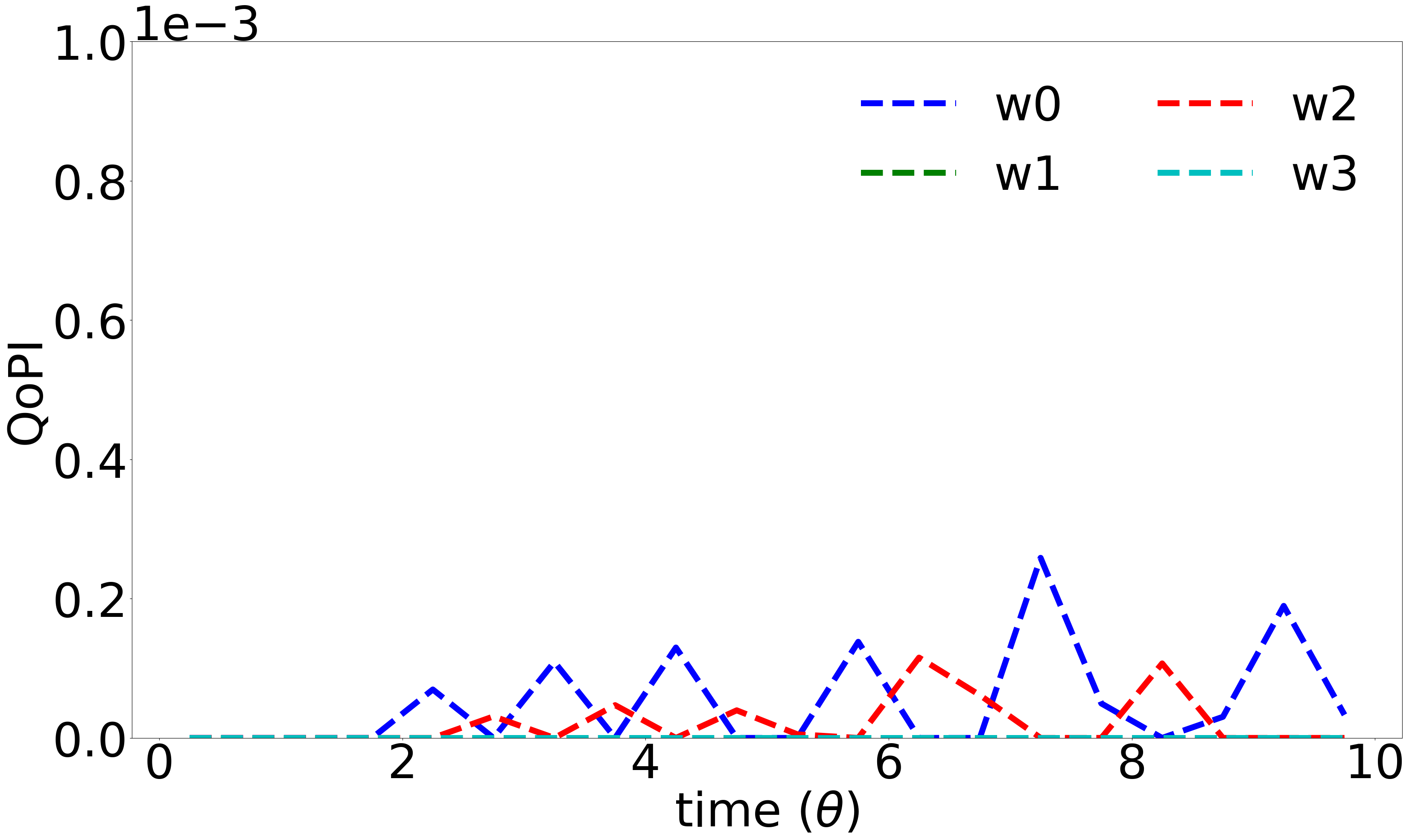}
	\includegraphics[width=0.32\textwidth]{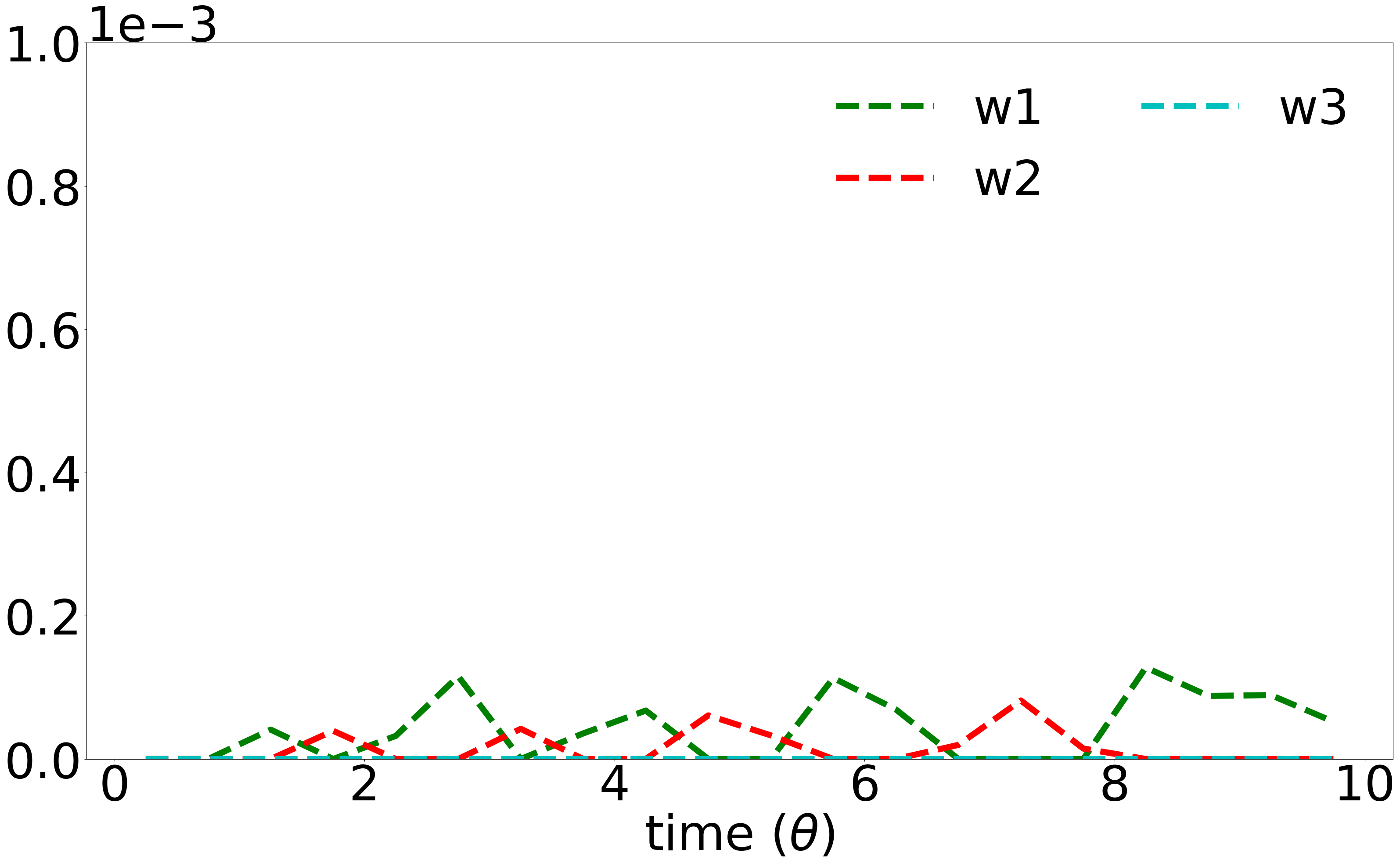}
	\includegraphics[width=0.32\textwidth]{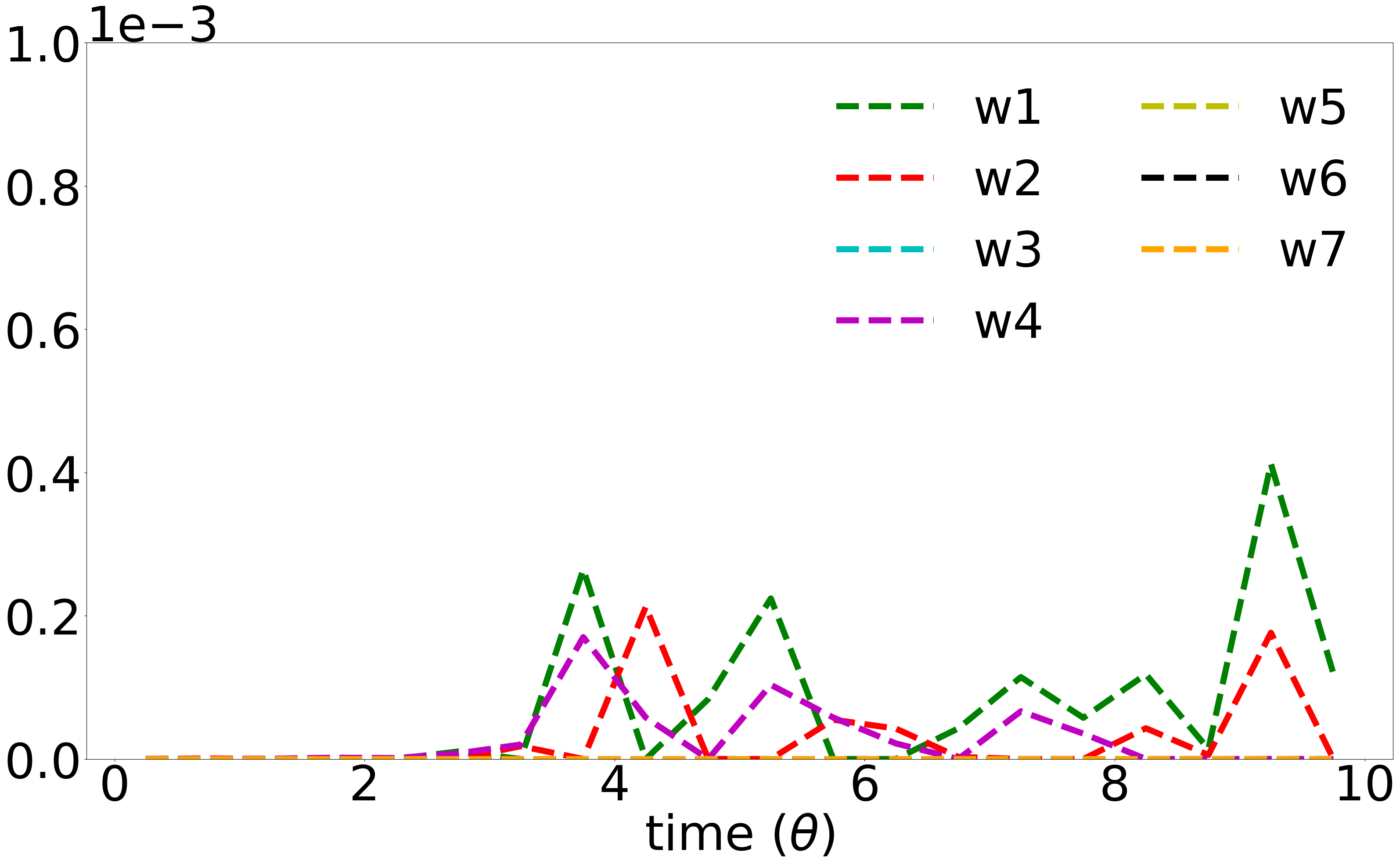}
	\caption{Walk inflow rates (top), total costs (middle) and QoPI
		(bottom) for each energy-feasible walk on the network shown in (a)
		\figref{example6WithRecharge} (top) (b) \figref{example6WithRecharge} (middle) with
		energy constraints (c) \figref{example6WithRecharge} (bottom) with energy
		constraints and a recharging station.
		\label{fig:example6Results}}
\end{figure}

\paragraph*{Energy profiles.}
Next, we can compare the energy profiles of the approximate equilibrium flow computed for the instance with and without recharging. These are depicted in \figref{combProfsToy}. We observe that the introduction of recharging stations increases the average energy consumption -- at times even above the maximum battery consumption, meaning that particles starting at these times on average use more than one complete battery charge for their travel.

\begin{figure}[h]
	\centering

	\includegraphics[width=0.49\textwidth]{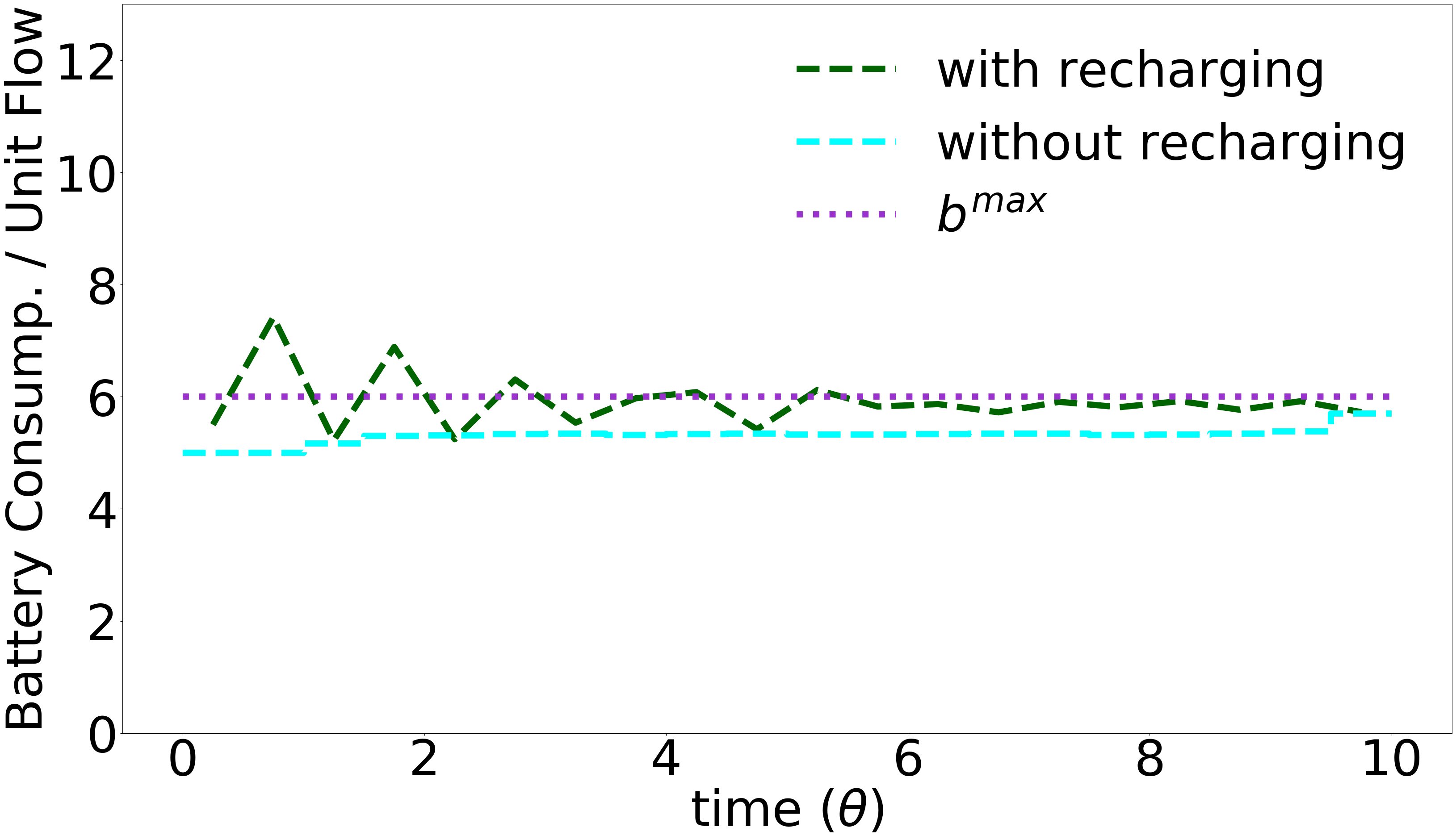}
	\caption{Energy consumption profiles for the networks Example1-B and 1-C.}
	\label{fig:combProfsToy}
\end{figure}

\paragraph*{Comparison with exact equilibria.}
As the instance considered here is quite small, we can also deduce some properties of exact equilibria in this network and compare those to the approximate equilibria computed by our algorithm. Namely, the flow computed for Example1-A (without energy constraints) is very close to a flow which sends everything into walk $W_0$ until time $\theta=2$ and splits the flow in a $2:1$ ratio between walks $W_0$ and $W_2$ after that. It is easy to verify that such flow is indeed an exact dynamic equilibrium for this network. For Example1-B (with energy consumption) we see that after around time $\theta=2$ the flow again is close to a stable split in $2:1$ ratio between two walks: This time between $W_1$ and $W_2$. One can again verify that this is an energy-feasible equilibrium. For the time before $\theta=2$, however, a simple constant split like in the case of Example1-A would not be an equilibrium. Instead a more gradual change from the initial flow split towards the one after time $2$ is necessary here, in order to also satisfy the equilibrium condition during this transition phase. We can also see this in the plot for the walk-flows of Example1-B.

\paragraph*{Overtaking in energy feasible equilibria.}
Finally, we want to point out the following effect taking place during the constant flow split after time $\theta=2$ in the (approximate) energy-feasible equilibrium computed for Example1-B which is not possible for dynamic equilibria in single-commodity networks: Namely, that simultaneous starting particles may overtake each other at intermediate nodes while still arriving at the sink at the same time. This effect can lead to much more involved structures of energy-feasible equilibria compared to dynamic equilibria. In particular, it seems unlikely that energy-feasible dynamic equilibria can be constructed by repeatedly extending a given partial equilibrium as it is possible for dynamic equilibria in single-commodity networks (cf. \cite{Koch11}). Since if one were to extend a given partial equilibrium flow, particles starting within the new extension period might overtake particles of a previous phase and then form a queue, hereby increasing the travel time of those earlier particles and possibly leading to violations of the equilibrium condition in the previously calculated part of the flow. Consequently, to directly compute an energy-feasible dynamic equilibrium the whole time-horizon $[0,T]$ has to be taken into account at once.

\subsubsection{Results for the Nguyen Network}

\begin{figure}[h!]
	\centering
	\resizebox{.5\linewidth}{!} {\input{tikz/nguyen.tex}}
	\caption{Pictorial depiction of the Nguyen-C network with recharging stations (with
		only one mode of recharging) at nodes labeled 6, 8 and 9.}
	\label{fig:nguyenNetwork}
\end{figure}
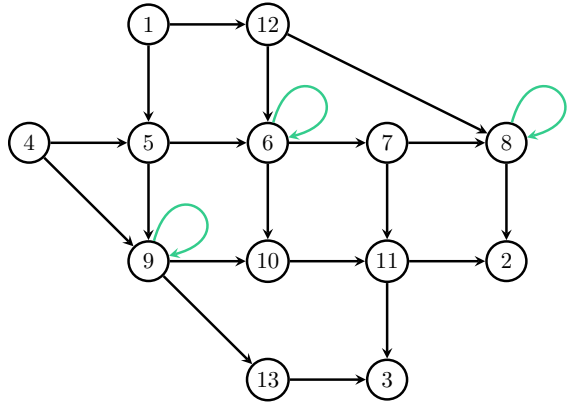

Next, we report on the performance of our algorithm on variants of the Nguyen network. These are, again, without energy constraints (-A), with energy
constraints (-B) and with energy constraints and recharging stations (-C) at nodes
labeled 6, 8, and 9 as shown in \figref{nguyenNetwork} (considering only one recharging mode per charging station). 

\paragraph*{Convergence measures.}
\tabref{nguyen} presents the performance measures of \algoref{fixedPoint} on these three variants of
Nguyen network with different number of commodities. All the instances except Nguyen-B with
twelve commodities converge to flows with QoPI values very close to zero. \figref{nguyenConv} shows the plots for $\Delta h$ and QoPI per iteration for the Nguyen network with four commodities demonstrating how the walk-flows converge to an approximate
energy-feasible dynamic equilibrium. Note here the logarithmic scale of the x-axis. Both plots show a very fast decrease for both $\Delta h$ and QoPI in the early iterations reaching a QoPI of less than $0.1$ after only about 10 iterations. On the other hand the extended tail of these plots indicate that a large number of
iterations are spent in a finer adjustment of flows among walks. It might be possible to improve this behaviour in later iterations by using a different setting of algorithmic parameters in future experiments.

\begin{table}[htbp]
\caption{Performance of \algoref{fixedPoint} on variants of Nguyen network with
different number of commodities. The following parameters are used: $\epsilon =
0.01$, $b^{max} = 5$, $\alpha^0 = 0.005$, wall clock time limit = 7200s, iteration
limit = 40000. `TimeLim' indicates that the run was terminated because the soft time limit
was exceeded.}
\scriptsize
\begin{tabular}{lccccccccc}\toprule
\# Commodities & \multicolumn{3}{c}{4} & \multicolumn{3}{c}{8} & \multicolumn{3}{c}{12} \\\cmidrule(lr){2-4}\cmidrule(lr){5-7}\cmidrule(lr){8-10}
Total inflow & \multicolumn{ 3}{c}{3600} & \multicolumn{ 3}{c}{7200} &
\multicolumn{3}{c}{10800} \\\cmidrule(lr){2-4}\cmidrule(lr){5-7}\cmidrule(lr){8-10}
Variant & A & B & C & A & B & C & A & B & C \\
\cmidrule(lr){2-2}\cmidrule(lr){3-3}\cmidrule(lr){4-4}\cmidrule(lr){5-5}\cmidrule(lr){6-6}
\cmidrule(lr){7-7}\cmidrule(lr){8-8}\cmidrule(lr){9-9}\cmidrule(lr){10-10}
Total \# walks & 25 & 17 & 25 & 35 & 27 & 35 & 44 & 36 & 44 \\
Wall clock time taken & 2634.712 & 275.452 & 1891.441 & 21.762 & 418.538 & 1163.050 & TimeLim & TimeLim & 3902.077 \\
Mean time/DNL & 0.061 & 0.042 & 0.060 & 0.190 & 0.125 & 0.236 & 0.246 & 0.156 & 0.349 \\
Mean time/FP-Update & 0.008 & 0.007 & 0.010 & 0.015 & 0.010 & 0.014 & 0.019 & 0.018 & 0.019 \\
\# iterations & 34349 & 5072 & 24376 & 100 & 2923 & 4468 & 26042 & 39048 & 10246 \\
\# walks with h > 0 & 18 & 17 & 21 & 29 & 24 & 28 & 27 & 25 & 31 \\
$\Delta h$ & 0.010 & 0.010 & 0.011 & 0.012 & 0.011 & 0.012 & 0.060 & 8315.566 & 0.010 \\
$\Delta h$ (relative) & 0.000 & 0.000 & 0.000 & 0.000 & 0.000 & 0.000 & 0.000 & 0.000 & 0.000 \\
QoPI (absolute) & 11.242 & 19.854 & 44.903 & 2386.938 & 268.378 & 863.509 & 1153.254 & 251934.439 & 2698.931 \\
QoPI & 0.000 & 0.000 & 0.001 & 0.059 & 0.003 & 0.012 & 0.008 & 5.070 & 0.040 \\\bottomrule
\end{tabular}
\begin{tabular}{lccccccccc}
\# Commodities & \multicolumn{ 3}{c}{16} & \multicolumn{3}{c}{20} \\ \cmidrule(lr){2-4}\cmidrule(lr){5-7}
Total inflow & \multicolumn{ 3}{c}{14400} & \multicolumn{ 3}{c}{18000} \\ \cmidrule(lr){2-4}\cmidrule(lr){5-7}
Variant & A & B & C & A & B & C \\
\cmidrule(lr){2-2}\cmidrule(lr){3-3}\cmidrule(lr){4-4}\cmidrule(lr){5-5}\cmidrule(lr){6-6}
\cmidrule(lr){7-7}
Total \# walks & 58 & 49 & 58 & 67 & 58 & 67 \\
Wall clock time taken & TimeLim & TimeLim & 462.7341 & TimeLim & TimeLim & TimeLim \\
Mean time/DNL & 0.362 & 0.449 & 0.462 & 0.346 & 0.451 & 0.456 \\
Mean time/FP-update & 0.024 & 0.022 & 0.026 & 0.028 & 0.025 & 0.030 \\
\# iterations & 17880 & 14749 & 917 & 18368 & 14582 & 14284 \\
\# walks with h > 0 & 36 & 42 & 34 & 36 & 49 & 36 \\
$\Delta h$ & 0.093 & 0.166 & 0.010 & 0.105 & 0.102 & 0.091 \\
$\Delta h$ (relative) & 0.000 & 0.000 & 0.000 & 0.000 & 0.000 & 0.000 \\
QoPI (absolute) & 368.589 & 910.713 & 2372.727 & 407.333 & 455.671 & 477.226 \\
QoPI & 0.002 & 0.006 & 0.017 & 0.002 & 0.002 & 0.003 \\ \bottomrule
\end{tabular}
\label{tab:nguyen}
\end{table}
\begin{figure}[h]
	\centering
	\includegraphics[width=0.49\textwidth]{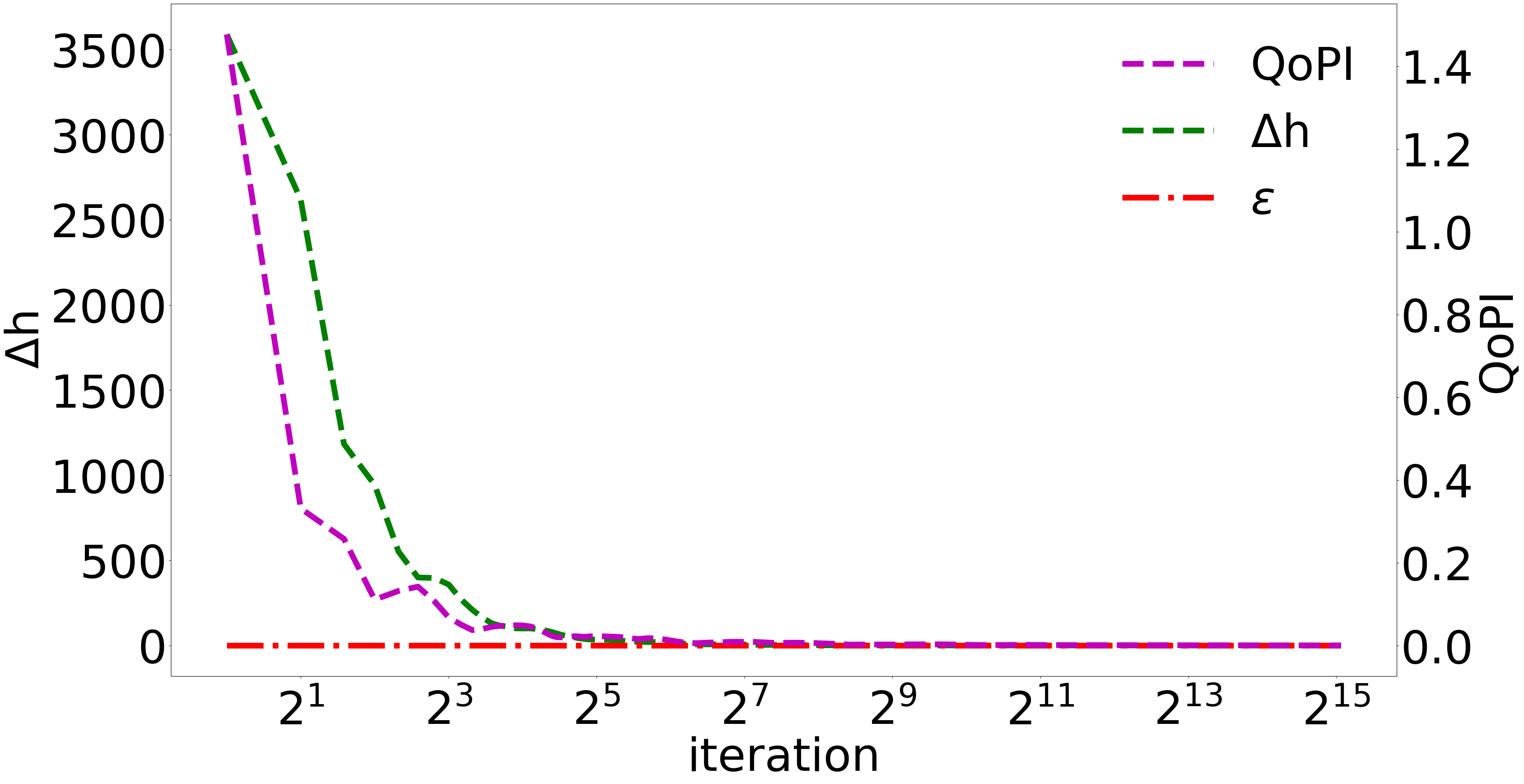}
	\includegraphics[width=0.49\textwidth]{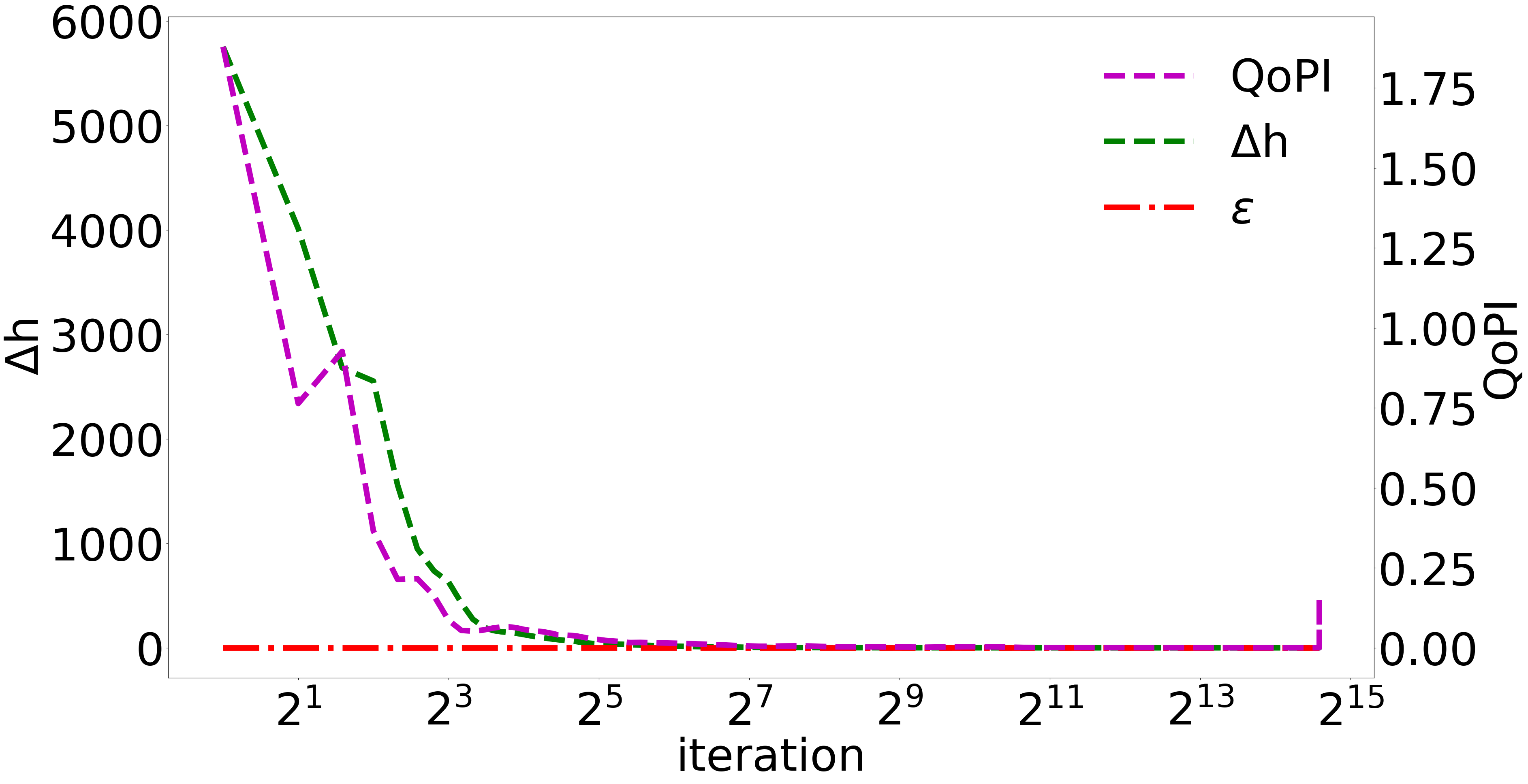}
	\caption{Change in the norm of walk-flows ($\Delta h$) and QoPI over iterations for
		the networks Nguyen-A (left) and Nguyen-C (right) for four commodities.}
	\label{fig:nguyenConv}
\end{figure}

\paragraph*{Travel times and energy profile.}
For the Nguyen network with four commodities we give some more details on the computed approximate equilibria in \Cref{fig:nguyen4minTravTimes,fig:nguyen4EnergyProfs,fig:combProfsNguyen}. \figref{nguyen4minTravTimes} shows the travel times of the four commodities. We can see here that the introduction of energy-constraints in Nguyen-B leads to higher travel times compared to Nguyen-A for three of the four commodities (as some of the shorter walks become infeasible) The addition of charging stations in Nguyen-C then reduces the minimum walk travel times for all these commodities again. Compared to the respective travel time profiles in Nguyen-A network (the case without energy-constraints), the commodity denoted `comm0' still incurs a higher travel time while the other commodities exhibit very similar profiles.

\figref{nguyen4EnergyProfs} shows the corresponding minimum, maximum and mean energy
consumption profiles, respectively, for Nguyen-B (top row)
and Nguyen-C (bottom row).  \figref{nguyen4minTravTimes} shows the energy profiles aggregated over all four commodities. We observe that in particular for commodity `comm3' the decrease in travel times from Nguyen-B to Nguyen-C comes with a significant increase in the energy consumption (especially evident from the plots for mean energy consumption). At the same time `comm2' also profits from the introduction of recharging stations in terms of travel time while incurring only a very modest increase in energy consumption: The maximum energy consumption increases significantly, but the average energy consumption stays almost the same, suggesting that while new (more energy intensive) walks become available, only a small percentage of this commodity's particles actually take these less energy efficient walks.
Altogether this shows that the effects of adding recharging stations on travel time and energy consumption can vary significantly between different agents. 

\begin{figure}[h]
	\centering
	\includegraphics[width=0.32\textwidth]{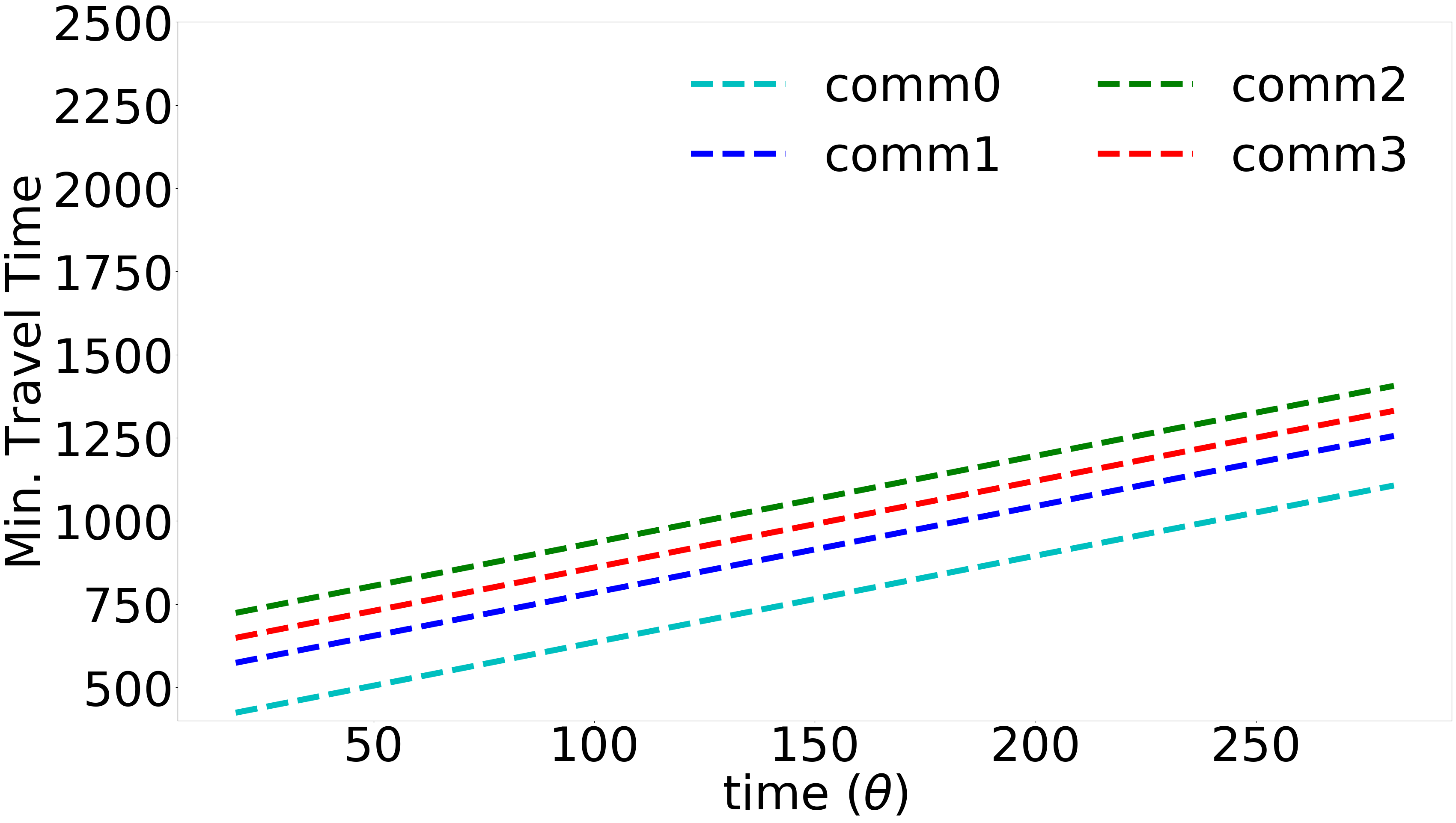}
	\includegraphics[width=0.32\textwidth]{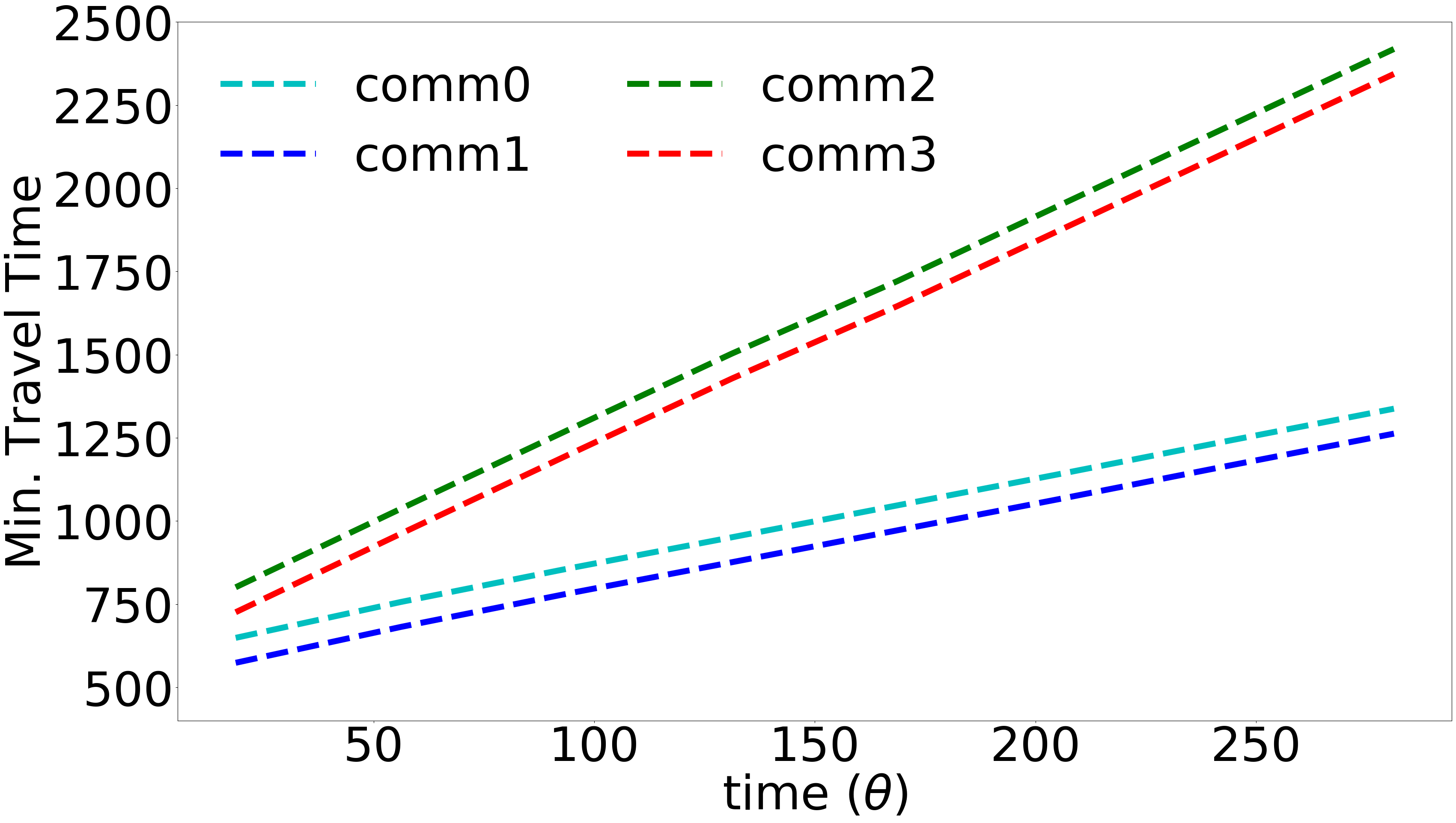}
	\includegraphics[width=0.32\textwidth]{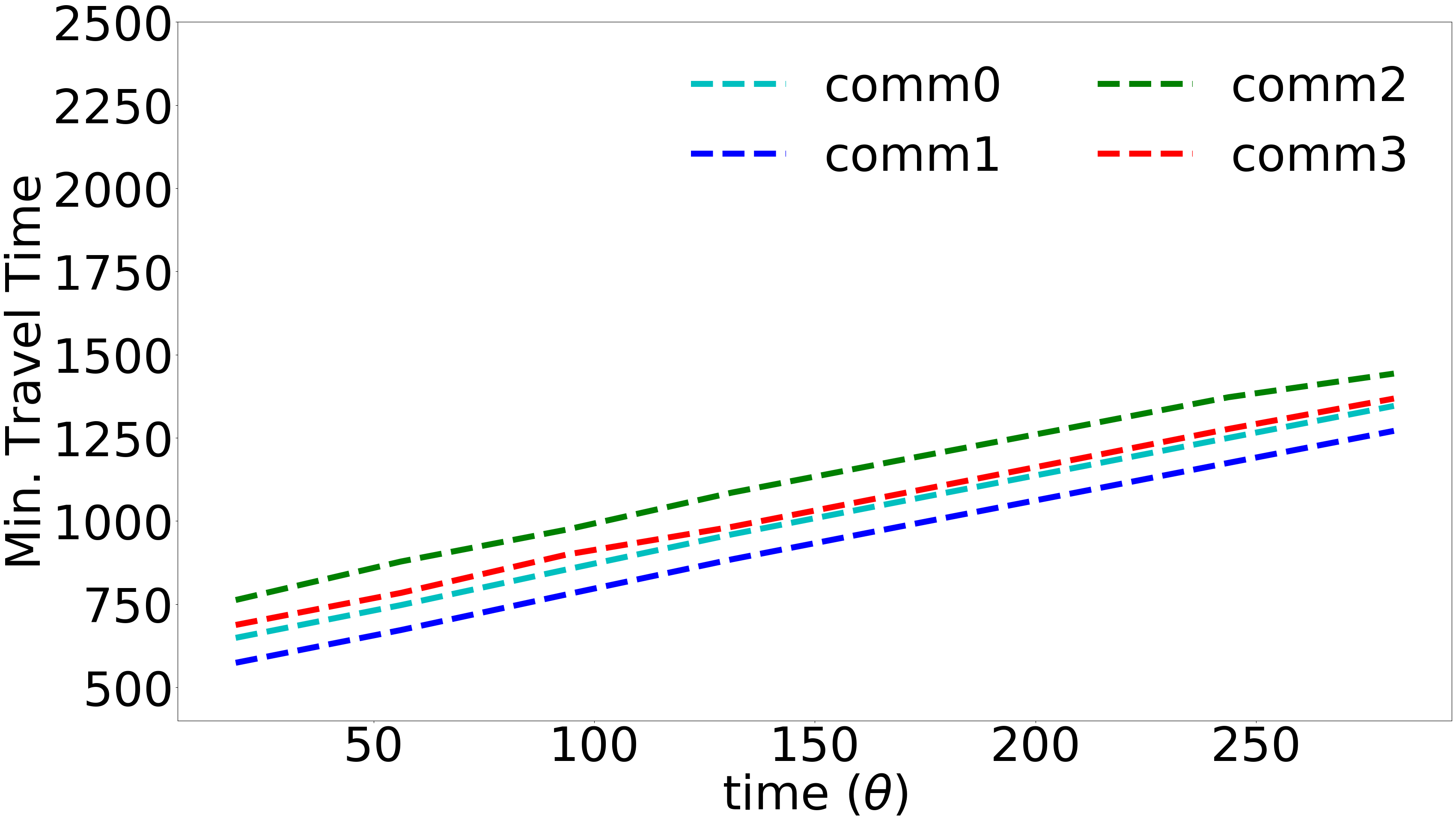}
	\caption{Minimum walk travel times for (four) different commodities for the Nguyen-A
		(left), Nguyen-B (centre) and Nguyen-C networks (right).}
	\label{fig:nguyen4minTravTimes}
\end{figure}
\begin{figure}[h]
\centering
\includegraphics[width=0.49\textwidth]{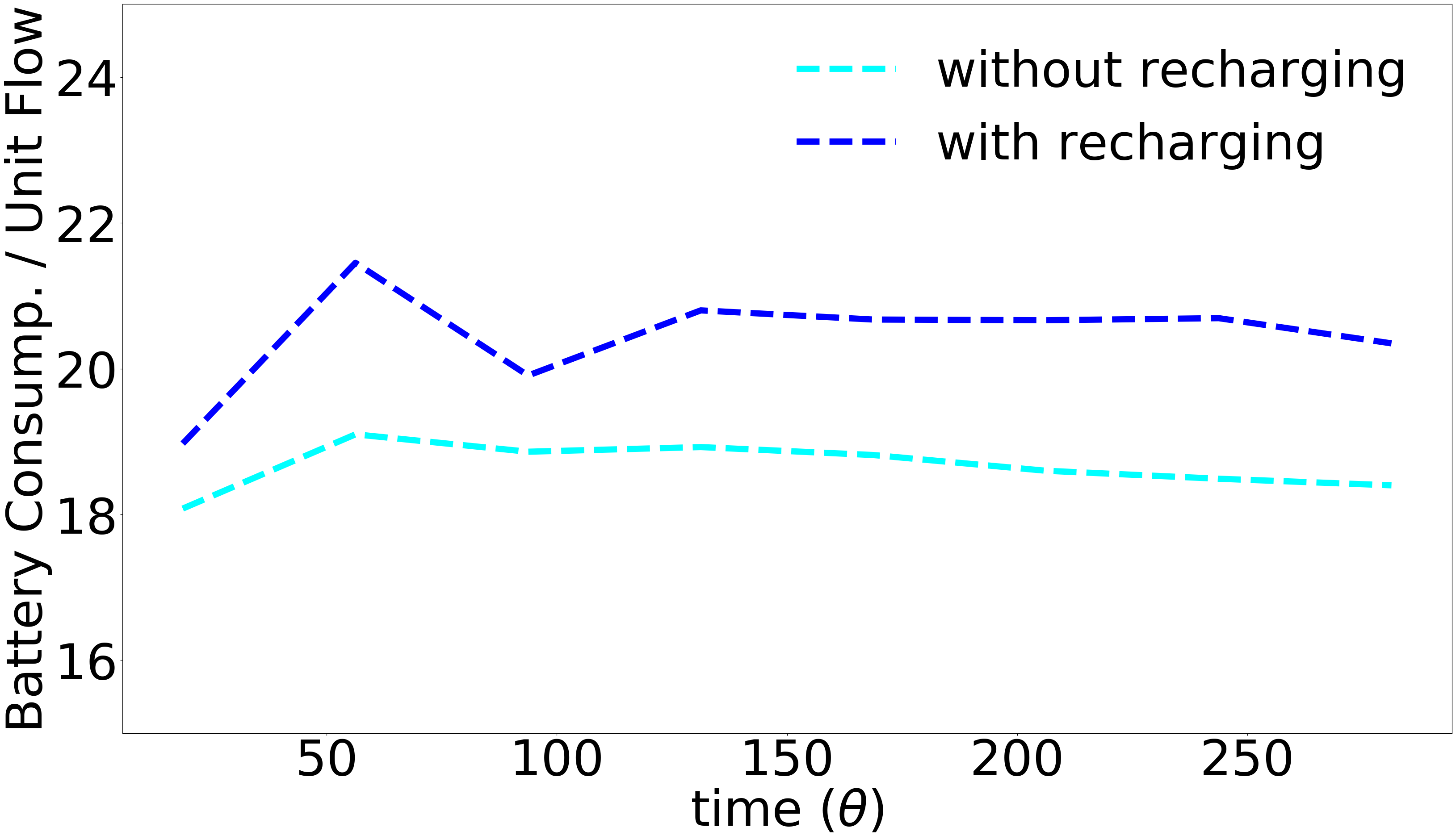}
\caption{Energy consumption profiles for the networks Nguyen-B and Nguyen-C four commodities.}
\label{fig:combProfsNguyen}
\end{figure}
\begin{figure}[h]
\centering
\includegraphics[width=0.32\textwidth]{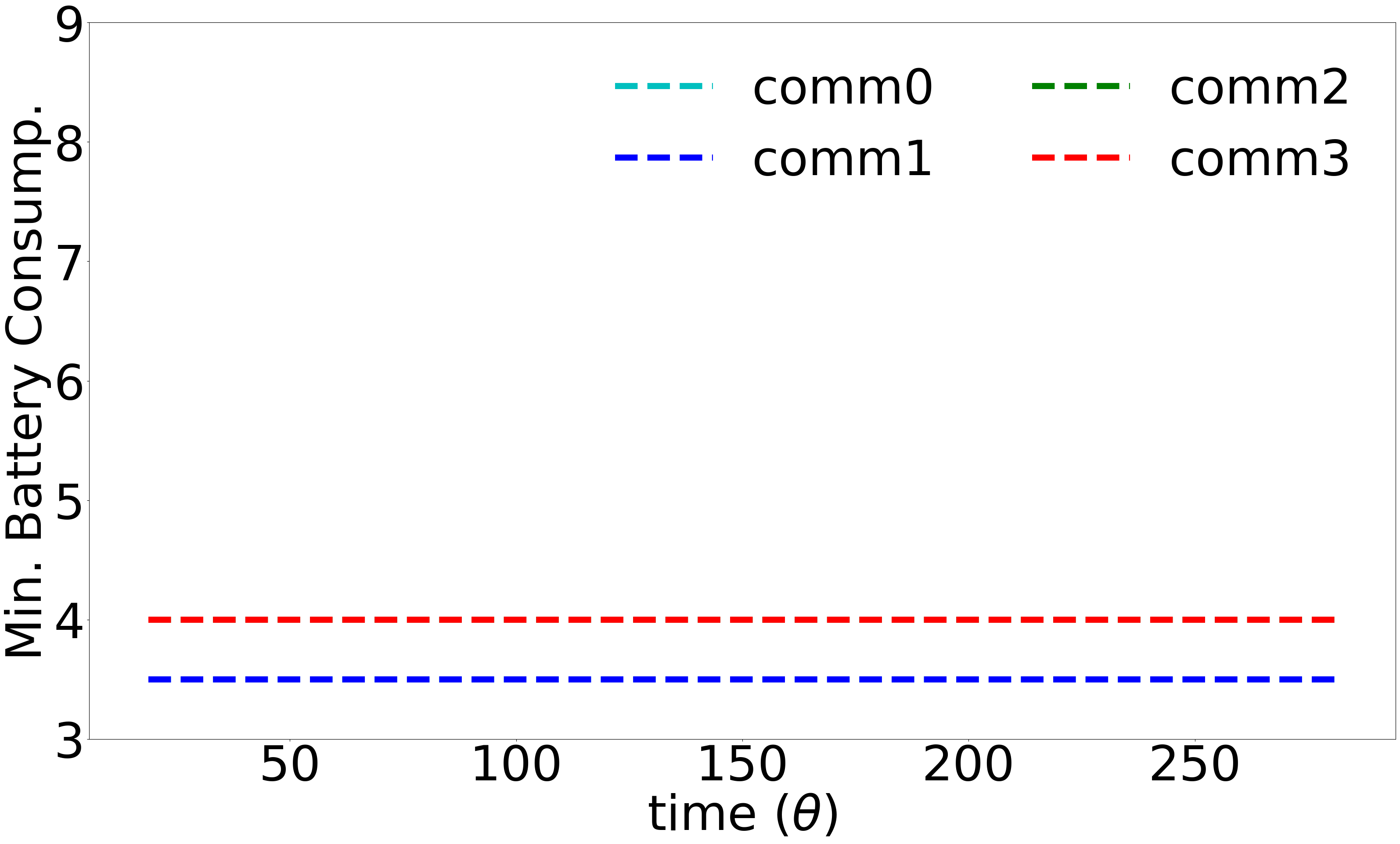}
\includegraphics[width=0.32\textwidth]{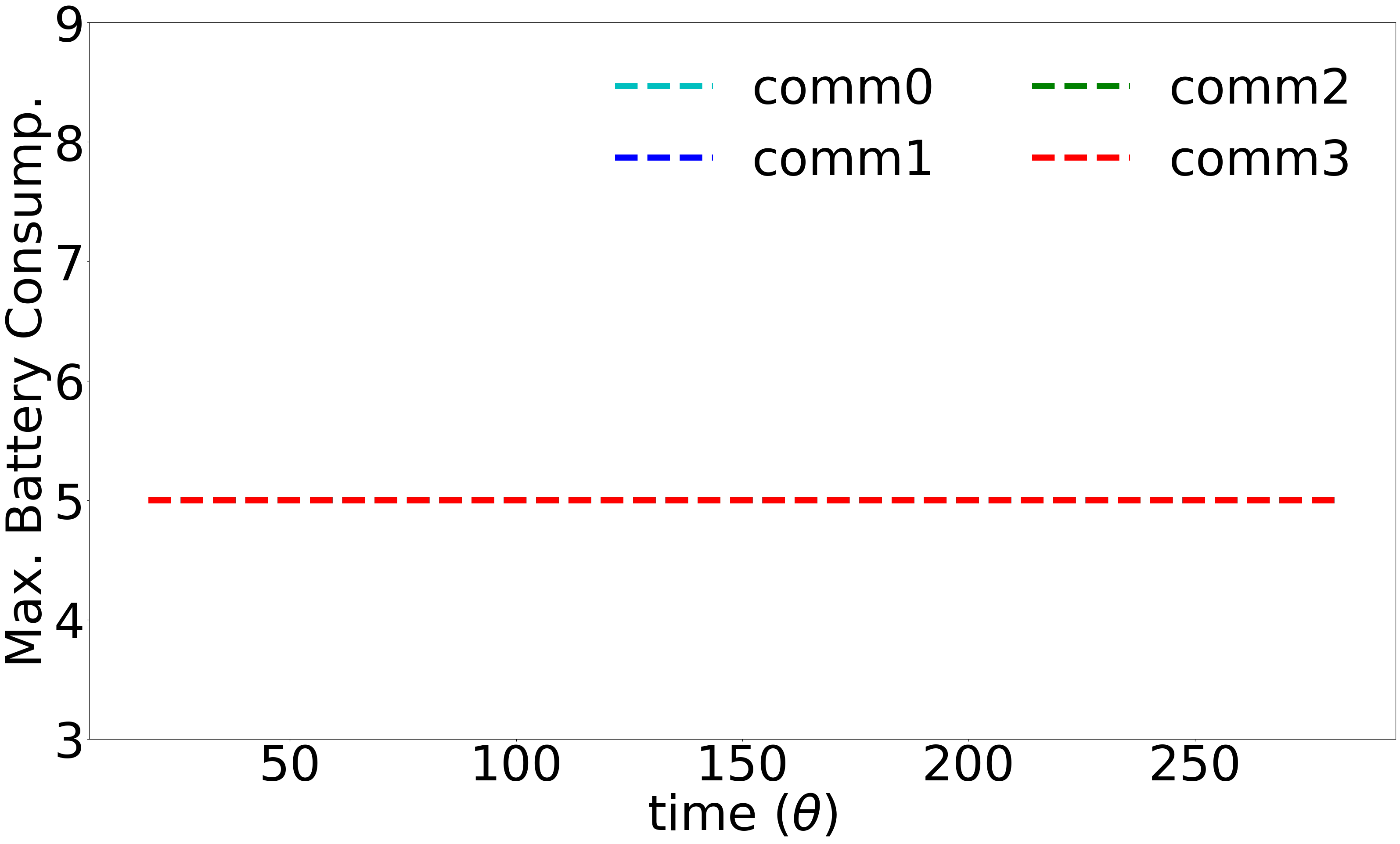}
\includegraphics[width=0.32\textwidth]{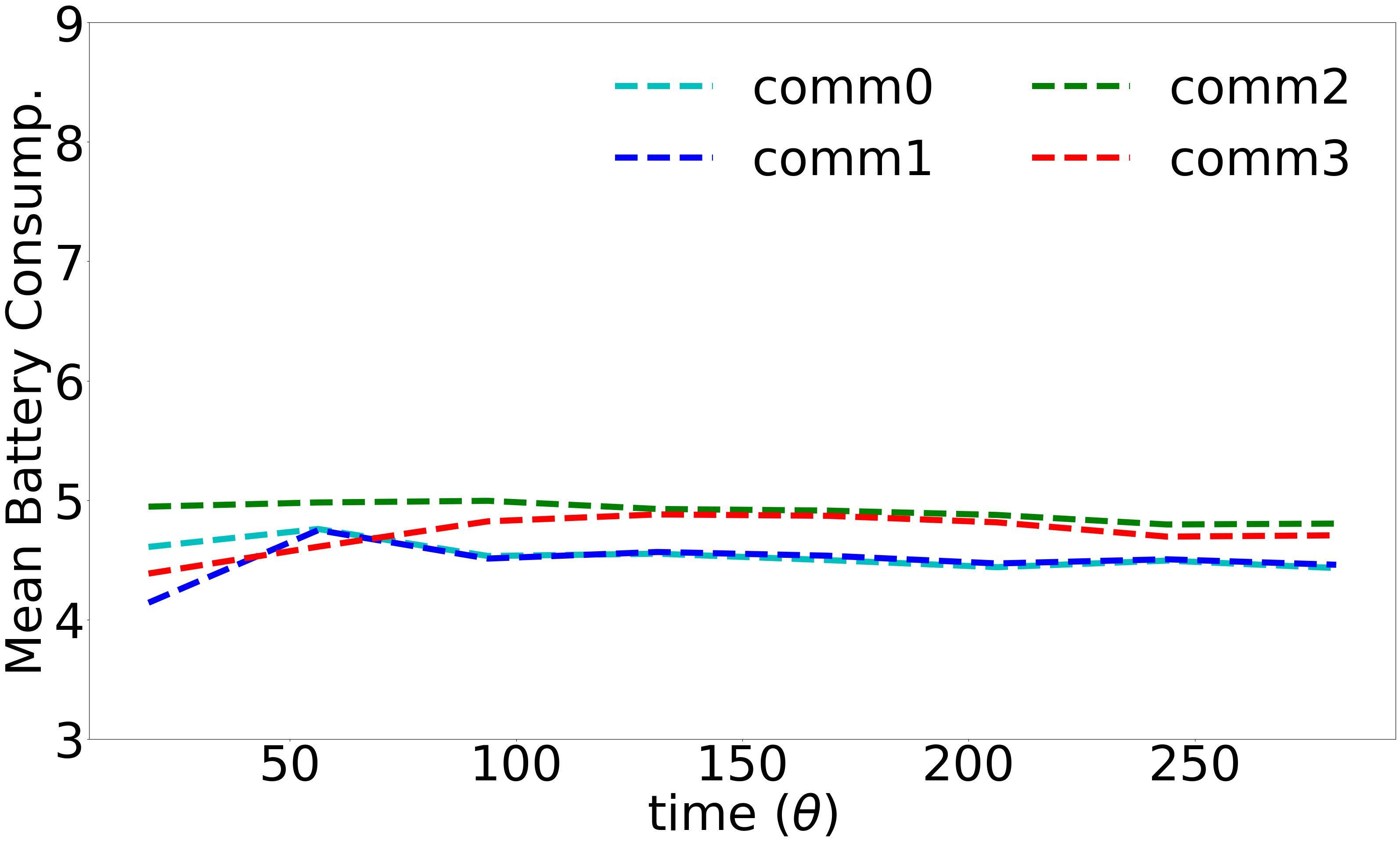}
\includegraphics[width=0.32\textwidth]{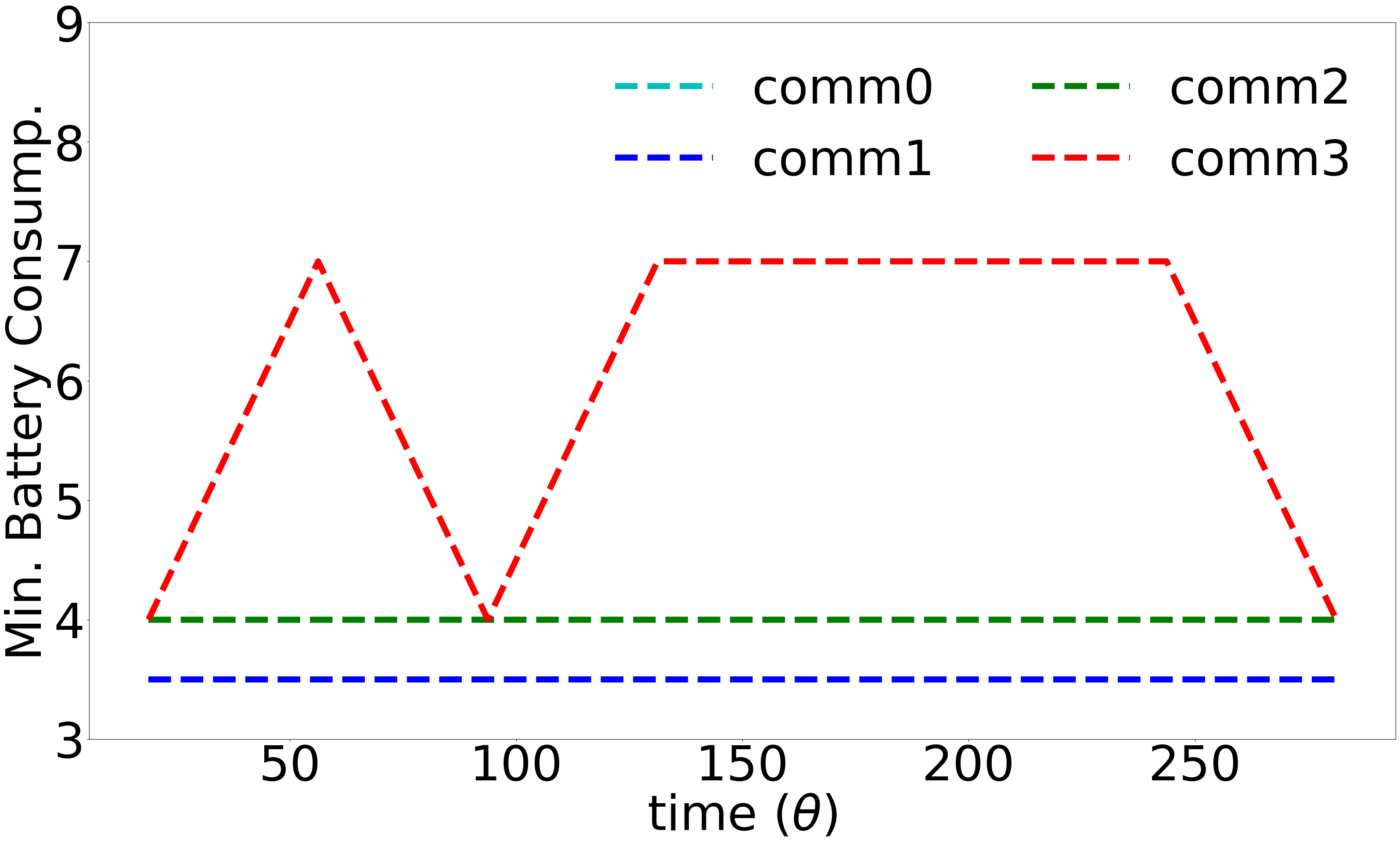}
\includegraphics[width=0.32\textwidth]{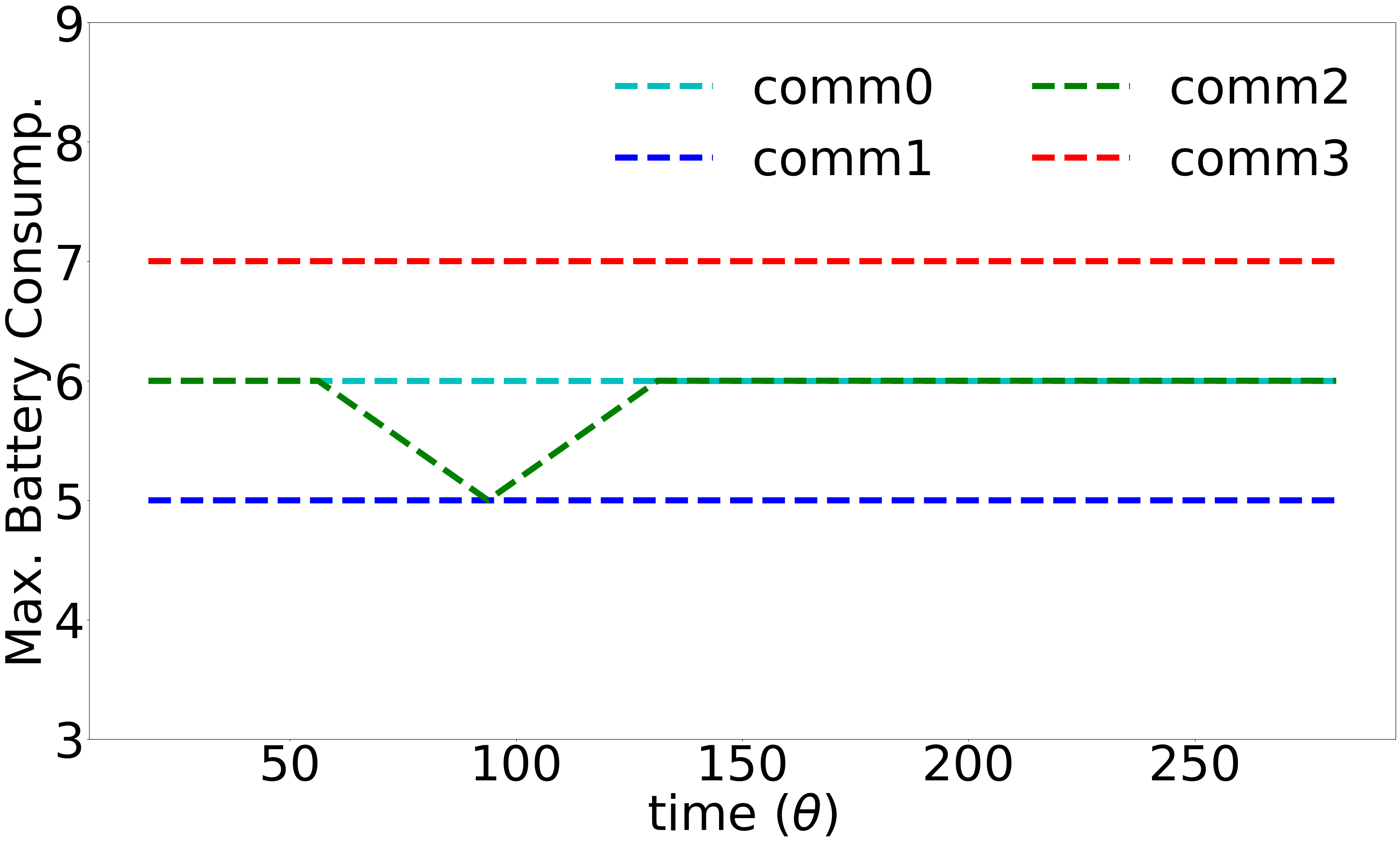}
\includegraphics[width=0.32\textwidth]{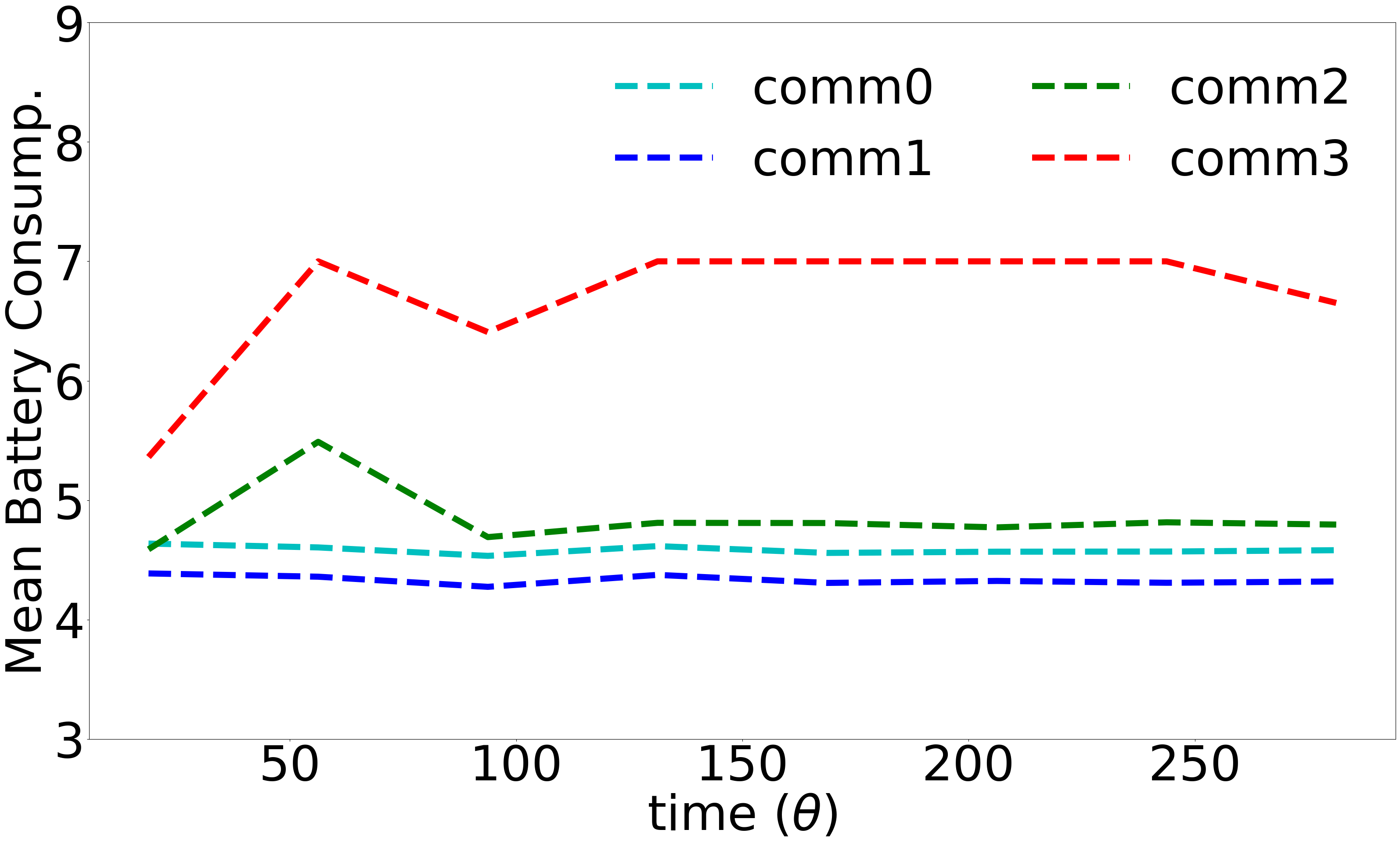}
\caption{Minimum, maximum and the mean energy consumption profiles for (four)
different commodities for the Nguyen-B (top) and Nguyen-C (bottom) networks.}
\label{fig:nguyen4EnergyProfs}
\end{figure}

\paragraph*{Effects of recharging prices.}
\figref{nguyenEnergy} (left) shows the effect of different prices for recharging
for the Nguyen-C network
on the energy consumption profiles via use of different values of the parameter $\widetilde{\lambda}_i$. The energy
consumption is initially higher when the
recharging prices are low (corresponding to smaller values of $\widetilde{\lambda}_i$), but
eventually, higher prices discourage the agents to use walks that require recharging, and the
respective energy profiles converge to the one with no recharging. 
The effect of prices on the mean (\figref{nguyenEnergy}, right), minimum
(\figref{nguyenPricedTimes}, left), and the maximum
(\figref{nguyenPricedTimes}, right) travel times taken for walks with a positive flow at
times $\theta \in [0,T]$ over all commodities is also shown. The maximum travel times are the lowest when there is no
price for recharging and increase with prices. For very high prices, the plots for maximum travel
time converge to the one with no recharging. The effect on the minimum travel times is
not as stark, however, the prices corresponding to values $\widetilde{\lambda}_i=10$
and $20$ seem to benefit the agents that start after specific time-points.
\begin{figure}[h]
\centering
\includegraphics[width=0.48\textwidth]{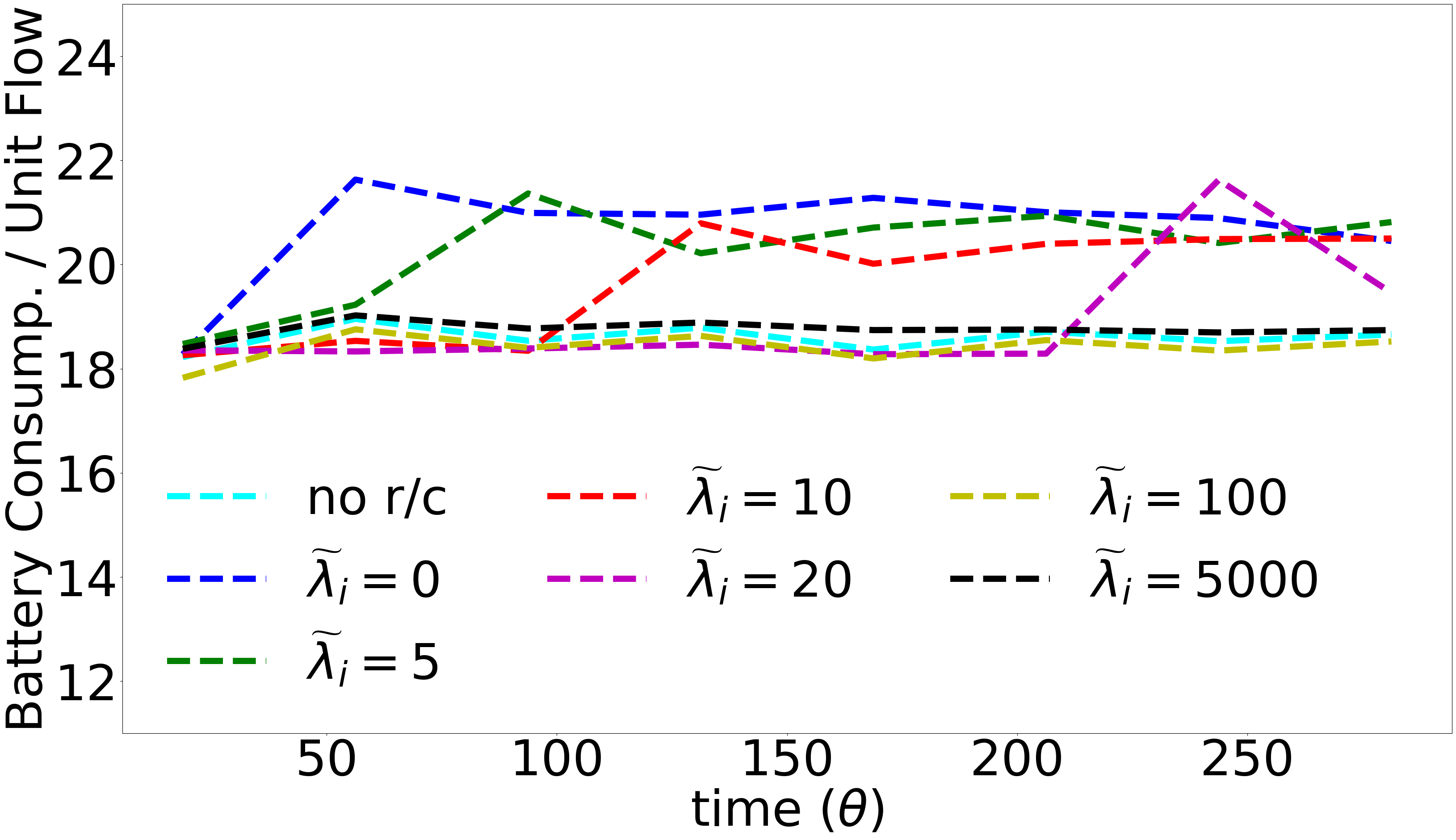}
\includegraphics[width=0.50\textwidth]{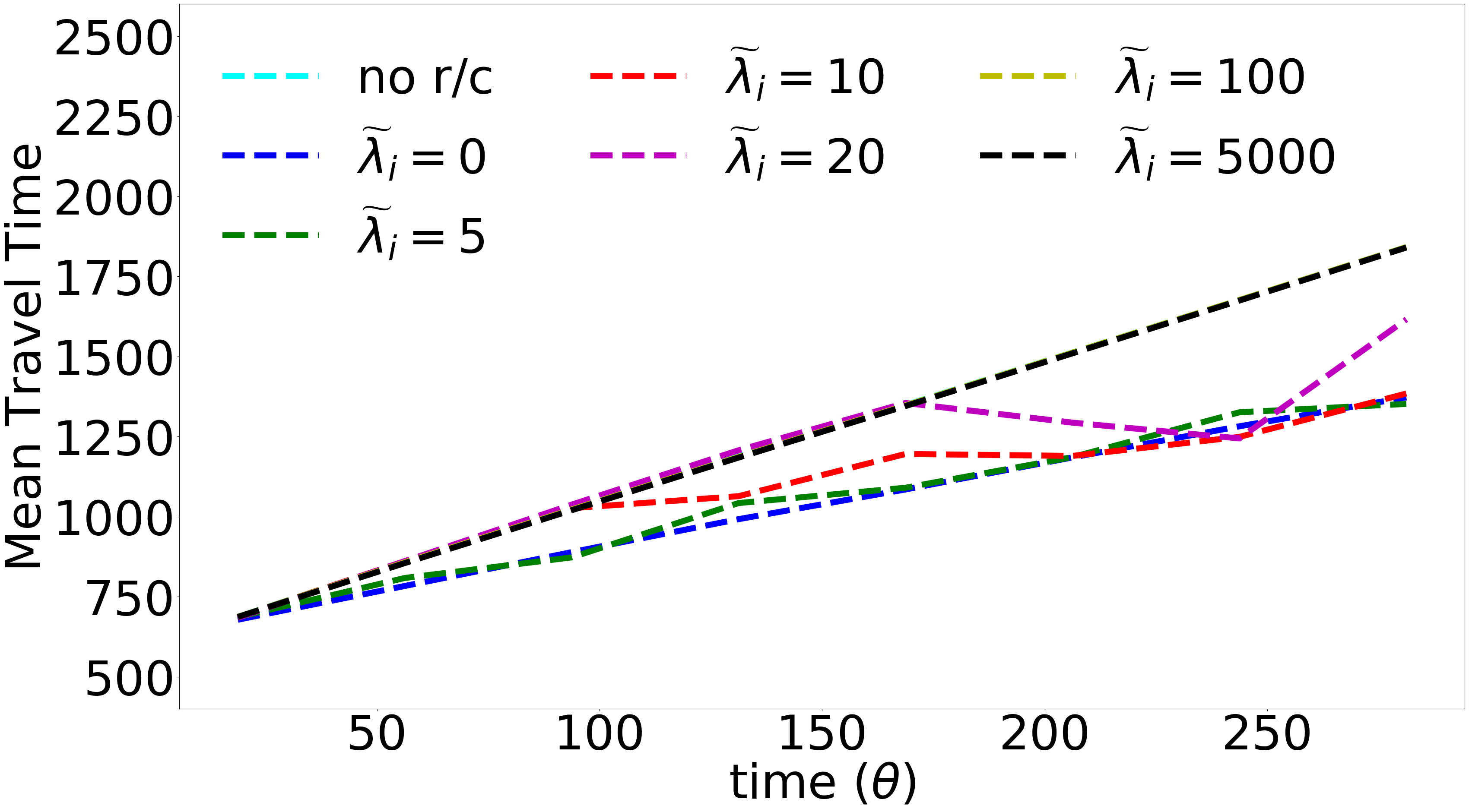}
\caption{Variation in the energy consumption profiles with different values of the
parameter $\tilde{\lambda}_i$ for recharing (left) and the mean travel times taken
for Nguyen-C (right) network with four commodities. $\alpha^0 = 0.001$ has been used
for these experiments.}
\label{fig:nguyenEnergy}
\end{figure}
\begin{figure}[h]
\centering
\includegraphics[width=0.49\textwidth]{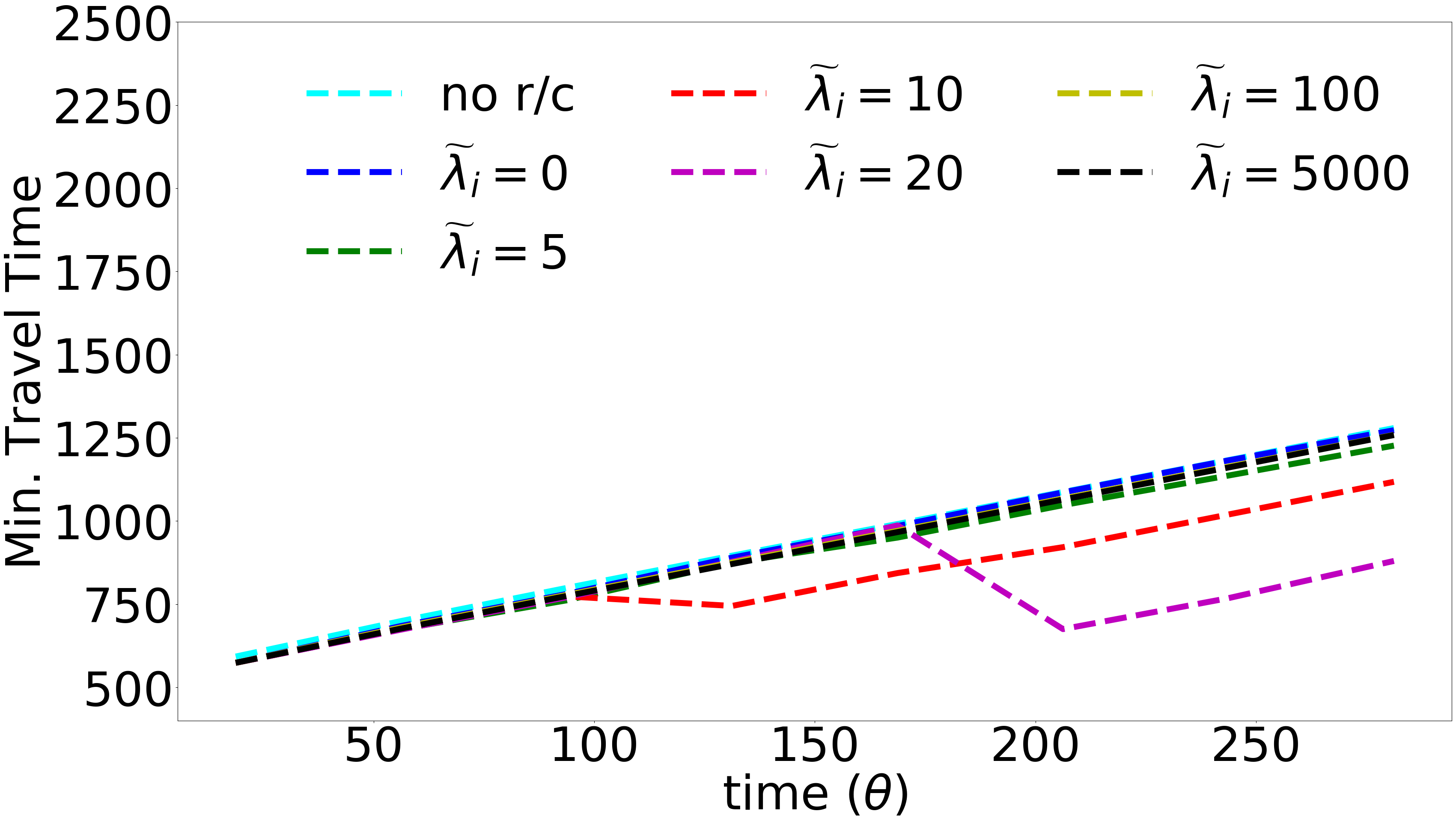}
\includegraphics[width=0.49\textwidth]{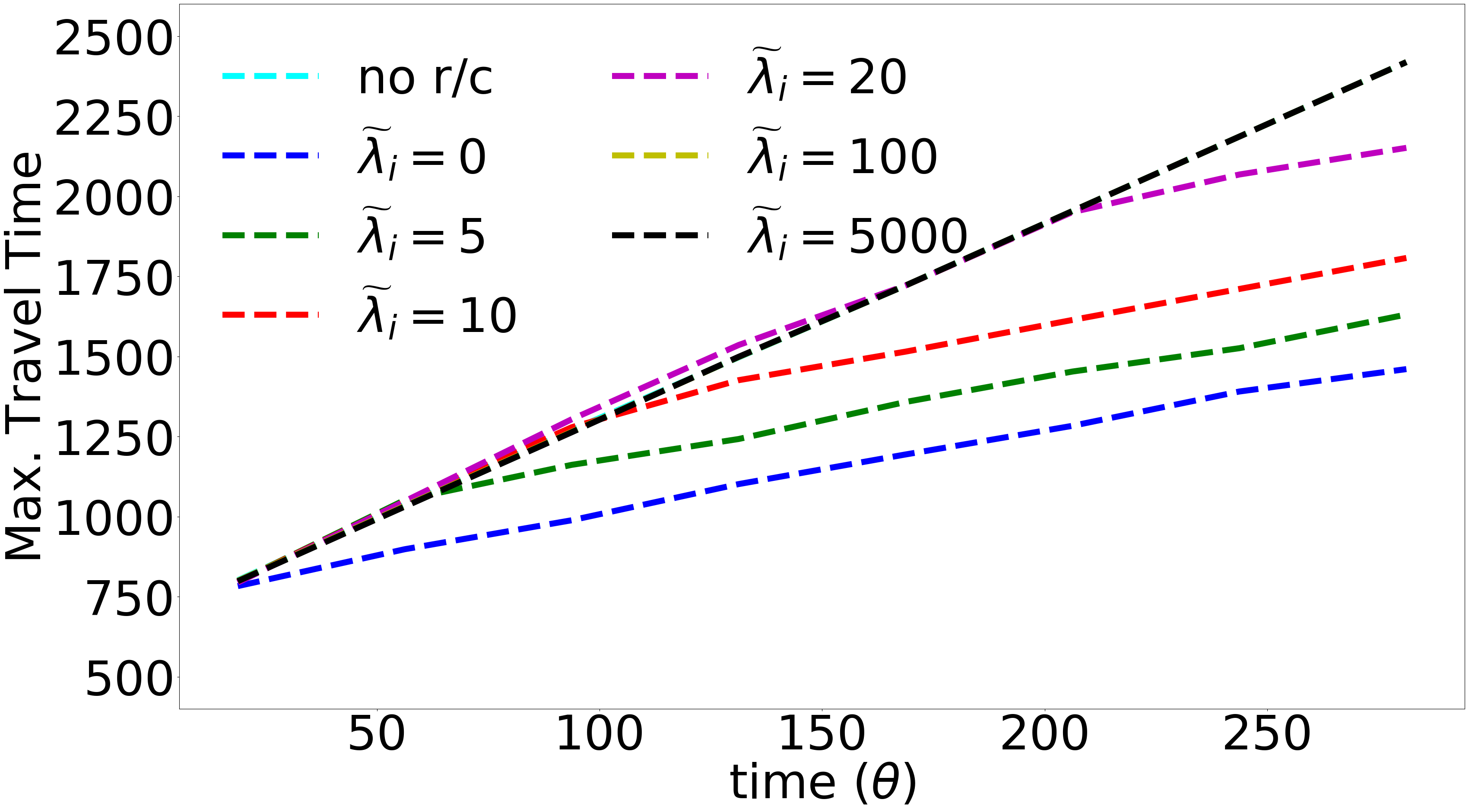}
\caption{The minimum (left) and the maximum (right) travel times taken over four
commodities for walks with a positive flow for the Nguyen-C network with different
values of the parameter $\widetilde{\lambda}_i$.}
\label{fig:nguyenPricedTimes}
\end{figure}
The plots in \figref{nguyenPricedTimes1} for the mean of the minimum (left) and the
mean of the maximum travel times (right) exhibit an effect similar to that on the
minimum and the maximum travel times.
\begin{figure}[h]
\centering
\includegraphics[width=0.49\textwidth]{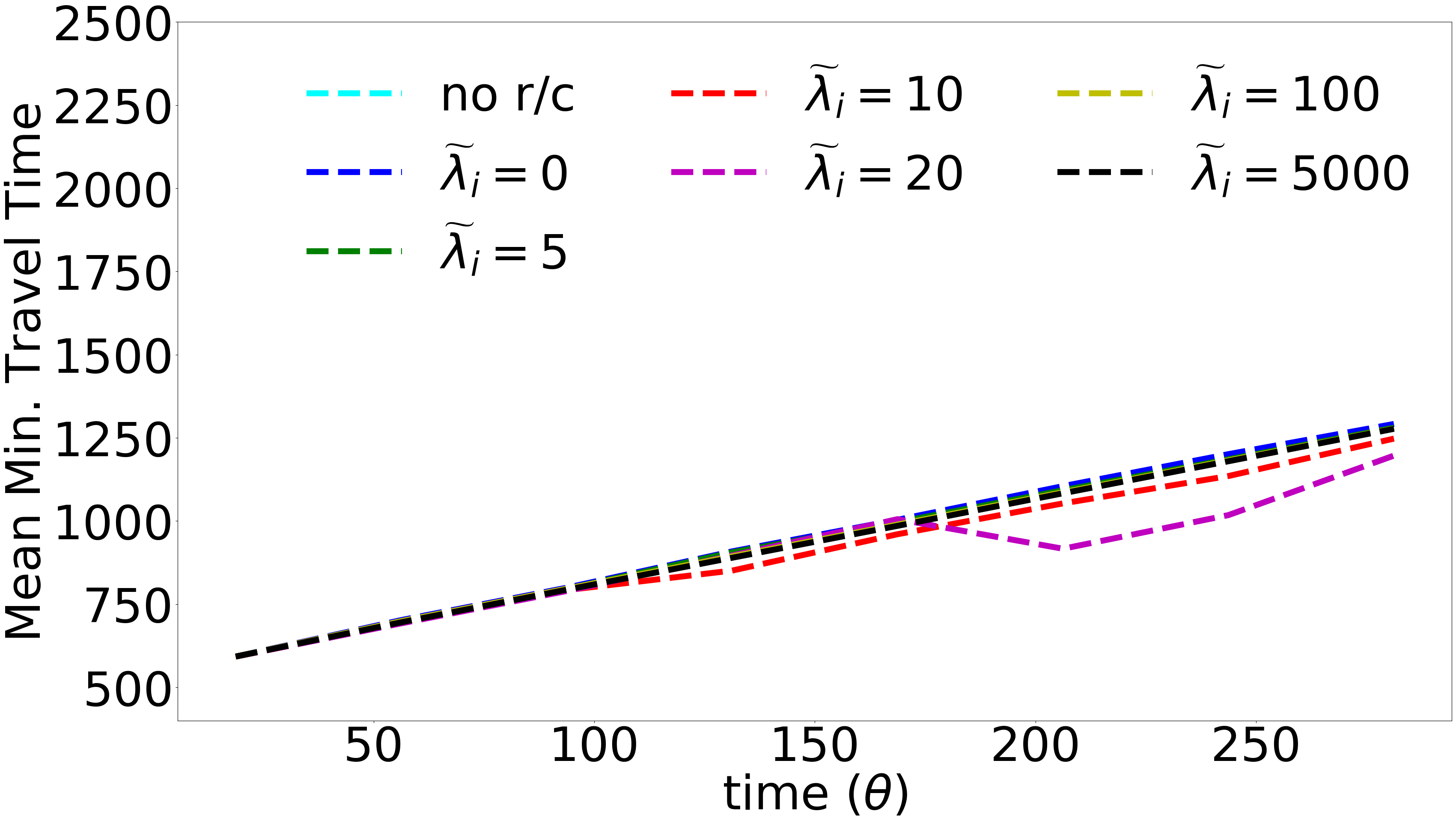}
\includegraphics[width=0.49\textwidth]{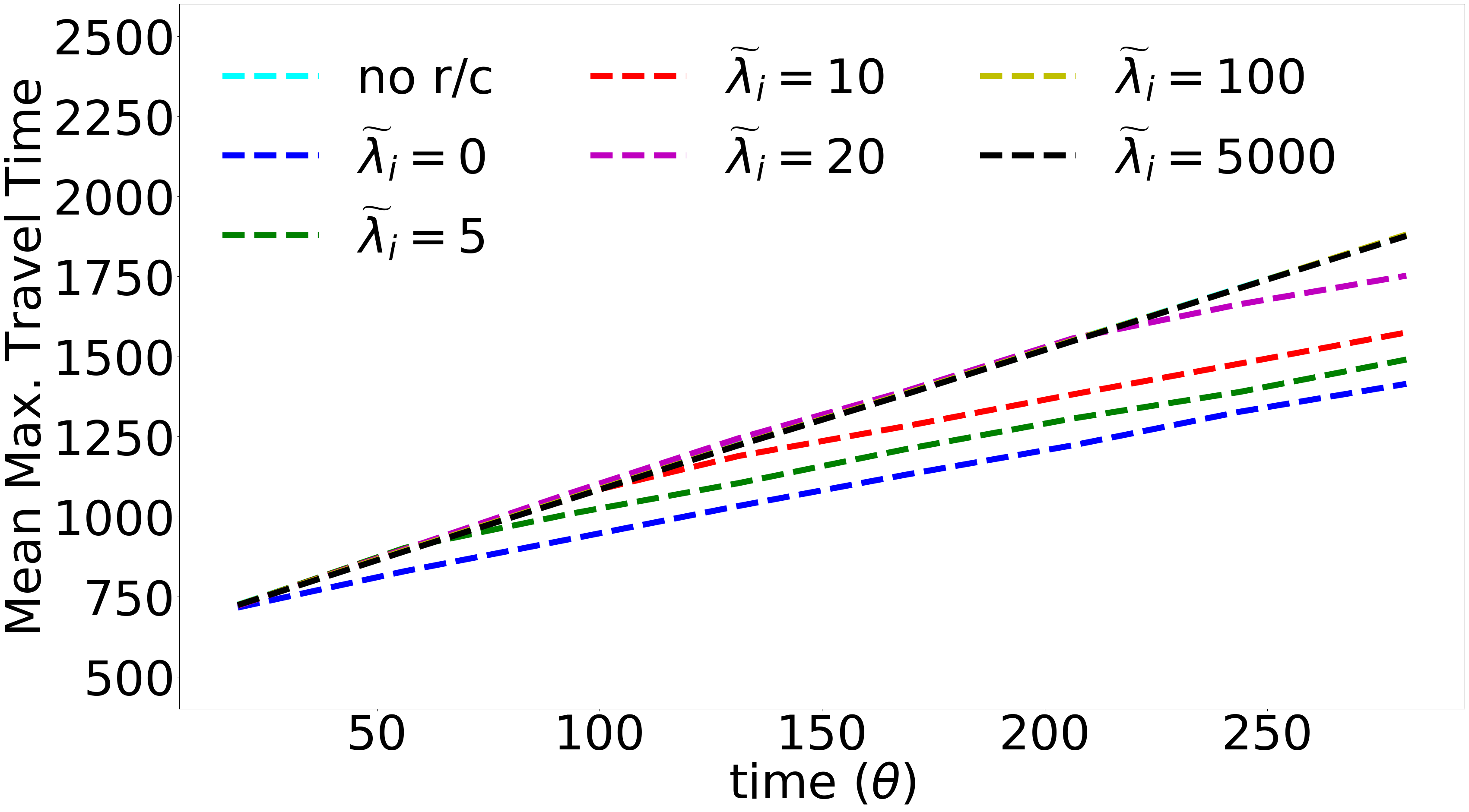}
\caption{The mean of minimum (left) and the mean of the maximum (right) travel times
over all four commodities for walks with a positive flow for the Nguyen-C network
with different values of the parameter $\widetilde{\lambda}_i$.}
\label{fig:nguyenPricedTimes1}
\end{figure}

\paragraph*{Effects of recharging station placement.}
To see the impact of different number and location of recharging stations, we used
recharging stations at nodes 6, 8 and 9 (with a zero price for recharging) for the
Nguyen-C network as shown in \figref{nguyenNetwork}.
We use the notation shown in \tabref{nguyenNotation} to indicate a particular
combination of the operating recharging stations.
\begin{table}[htbp]
\caption{Notation indicating the number and location of recharging stations for the
Nguyen-C network.}
\centering
\begin{tabular}{cccccccc}\toprule
\# r/c stations & \multicolumn{ 3}{c}{1} & \multicolumn{ 3}{c}{2} & 3 \\\cmidrule(lr){2-4}\cmidrule(lr){5-7}\cmidrule(lr){8-8}
r/c station(s) at node(s) & \{6\} & \{8\} & \{9\} & \{6, 8\} & \{6, 9\} & \{8, 9\} & \{6, 8, 9\} \\
Notation & R1s1 & R1s2 & R1s3 & R2s1 & R2s2 & R2s3 & R3 \\ \bottomrule
\end{tabular}
\label{tab:nguyenNotation}
\end{table}
\figref{rnguyen4Energy} (left) shows the energy consumption profiles with the
corresponding mean travel times in \figref{rnguyen4Energy} (right). The highest
energy consumption corresponds to combinations R2s3, R3, and R1s3 (in the mentioned
order) that include the recharging
station at node 9. \figref{rnguyen4Times} shows the minimum and maximum travel times. 
For clarity, the mean times taken using 1, 2 and 3 recharging stations are
shown in \figref{rnguyen4TimesR123} which highlight that the mentioned combinations
corresponding to high energy use involving the recharging
station at node 9 lead to shorter travel times on average. Similar patterns are observed for
the means of the minimum and the maximum travel times as shown in \figref{rnguyen4Times1}.
When using eight commodities for the Nguyen-C network, we observe patterns similar to
the case of four commodities in the energy profiles and the mean travel times as
shown in \figref{rnguyen8}.
\begin{figure}[h]
\centering
\includegraphics[width=0.48\textwidth]{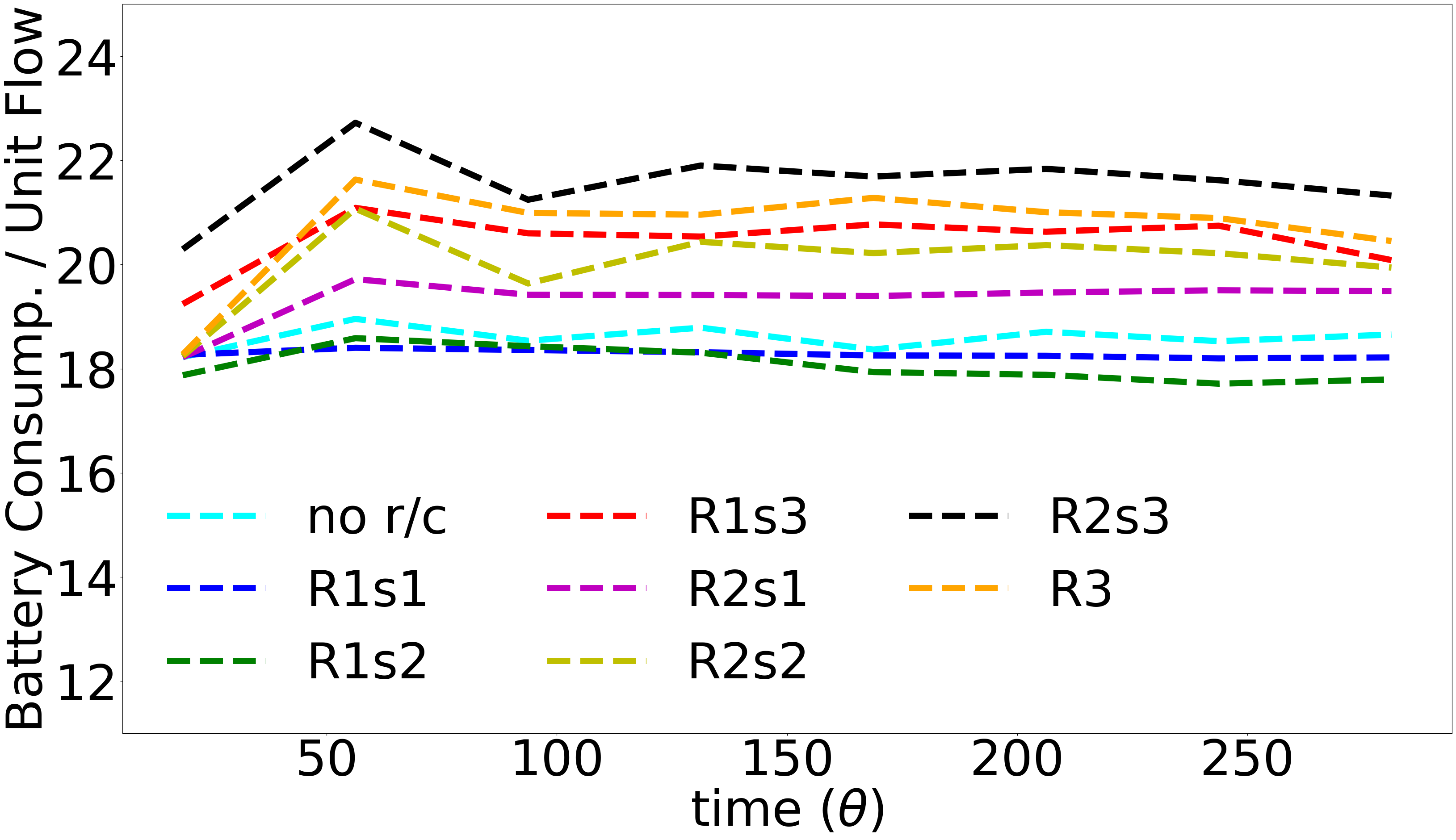}
\includegraphics[width=0.50\textwidth]{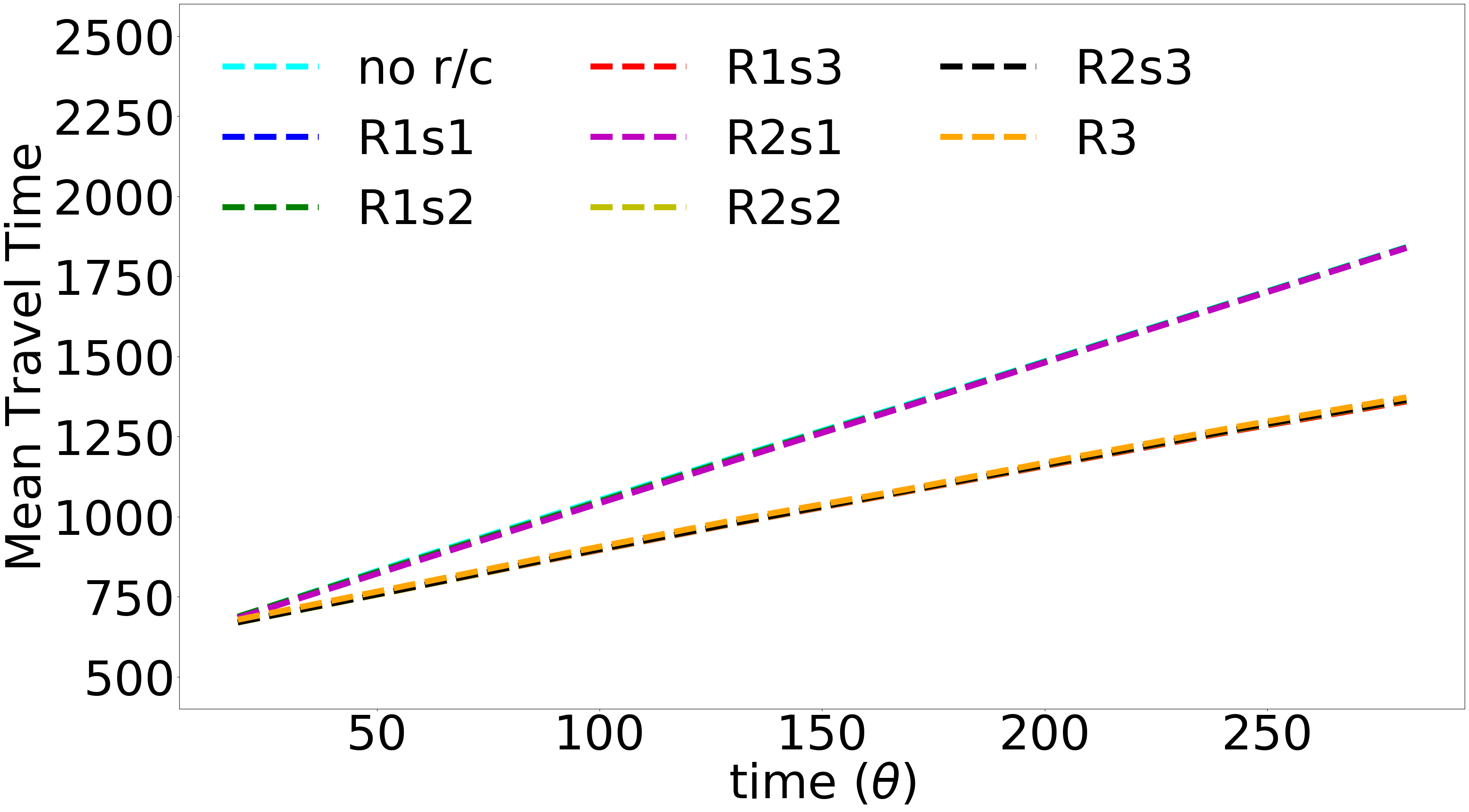}
\caption{The energy consumption profiles (left) and the mean travel times (right)
taken over all four commodities for walks with a positive flow for the Nguyen-C
network with different different combinations of recharging stations at nodes 6, 8 and/or 9.
}
\label{fig:rnguyen4Energy}
\end{figure}
\begin{figure}[h]
\centering
\includegraphics[width=0.49\textwidth]{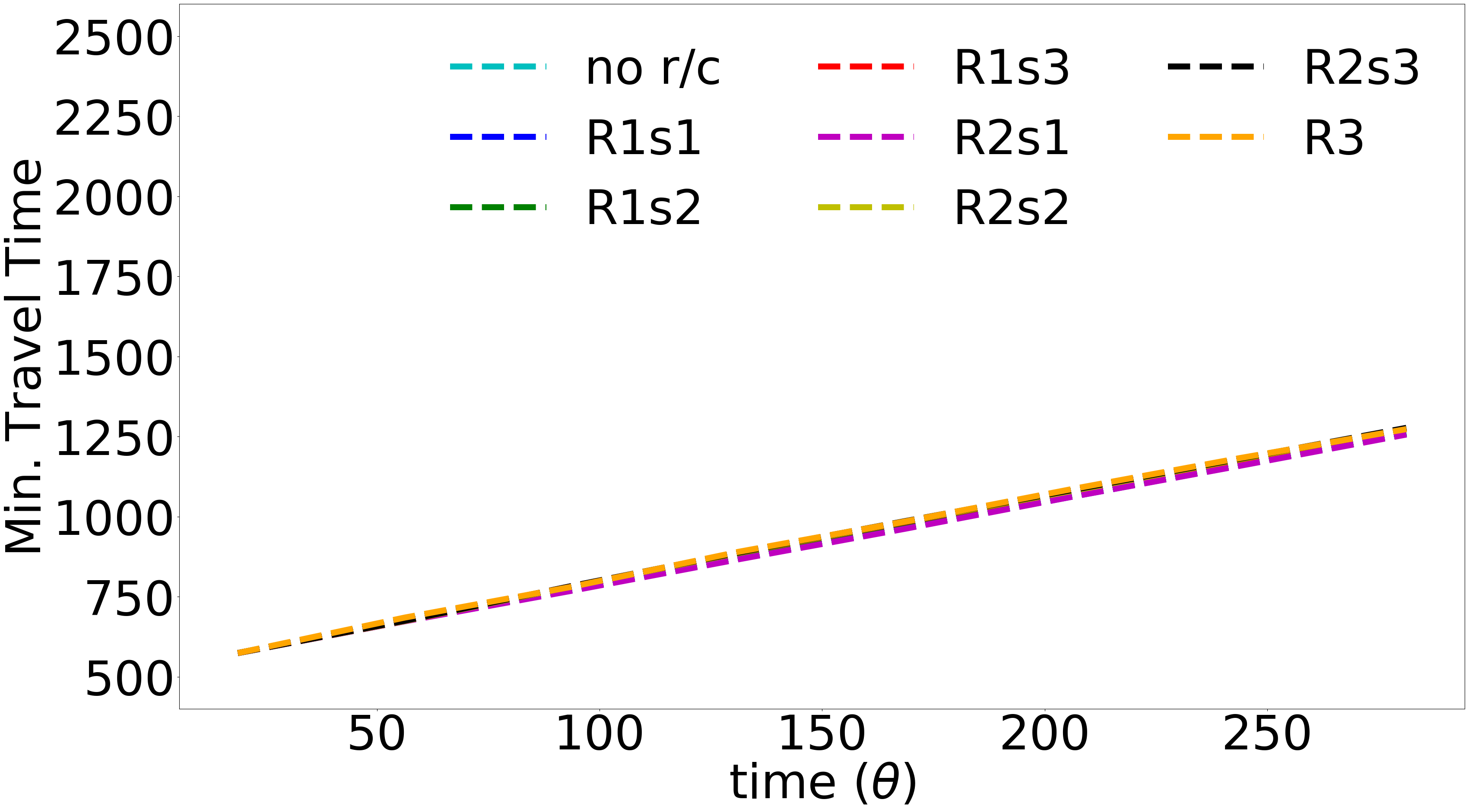}
\includegraphics[width=0.49\textwidth]{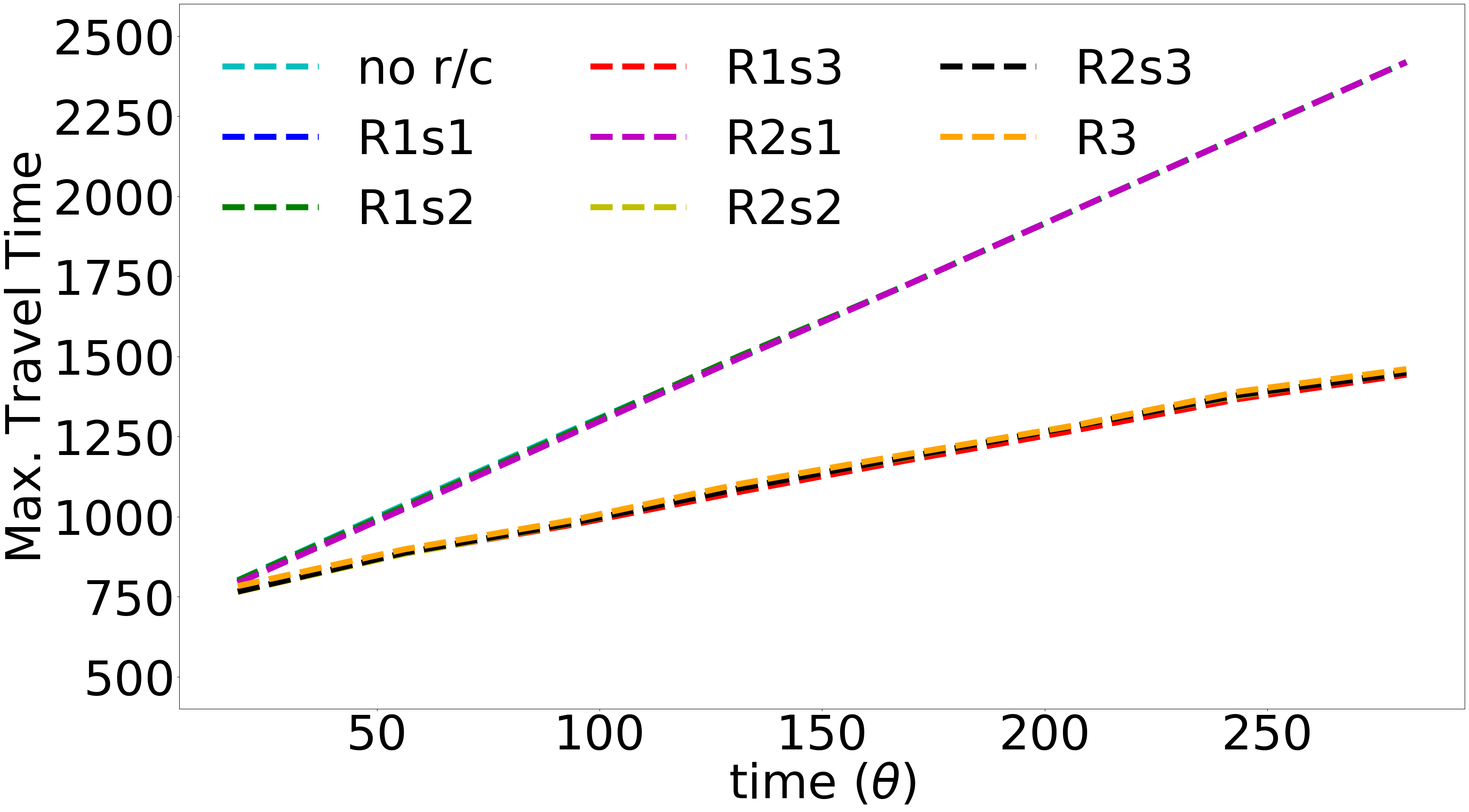}
\caption{The minimum (left) and the maximum (right) travel times
over all four commodities for walks with a positive flow for the Nguyen-C network
with different combinations of recharging stations at nodes 6, 8 and/or 9.}
\label{fig:rnguyen4Times}
\end{figure}
\begin{figure}[h]
\centering
\includegraphics[width=0.32\textwidth]{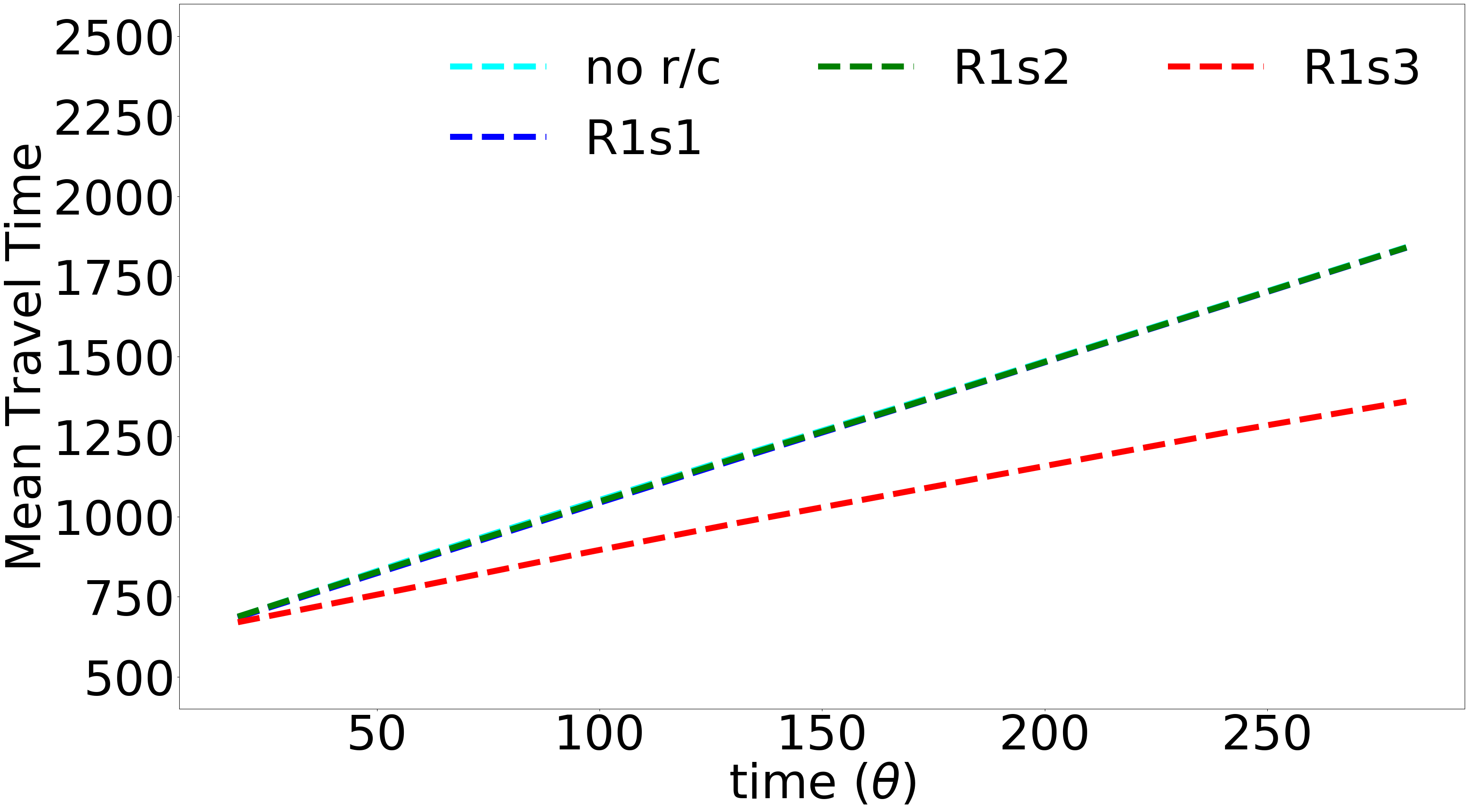}
\includegraphics[width=0.32\textwidth]{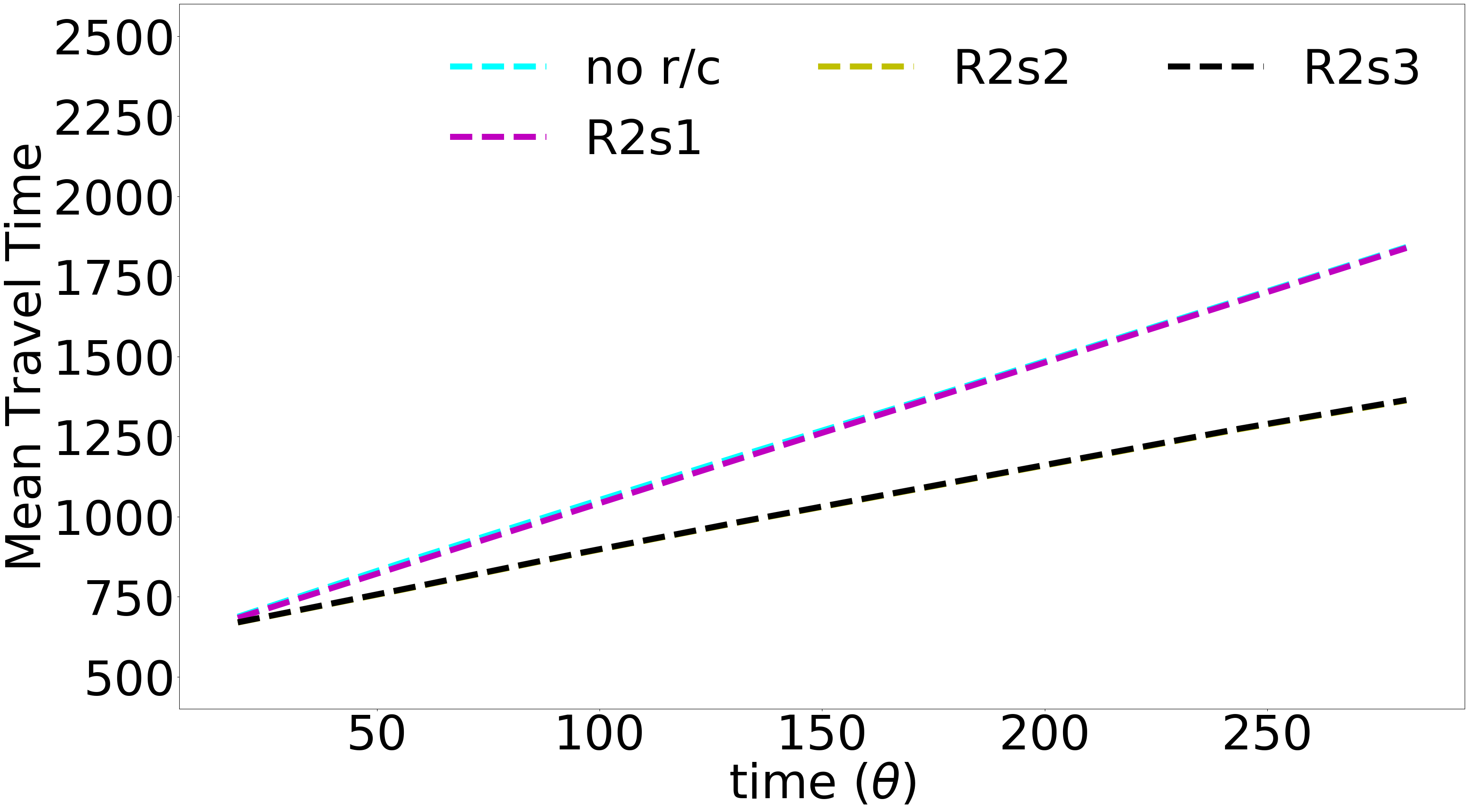}
\includegraphics[width=0.32\textwidth]{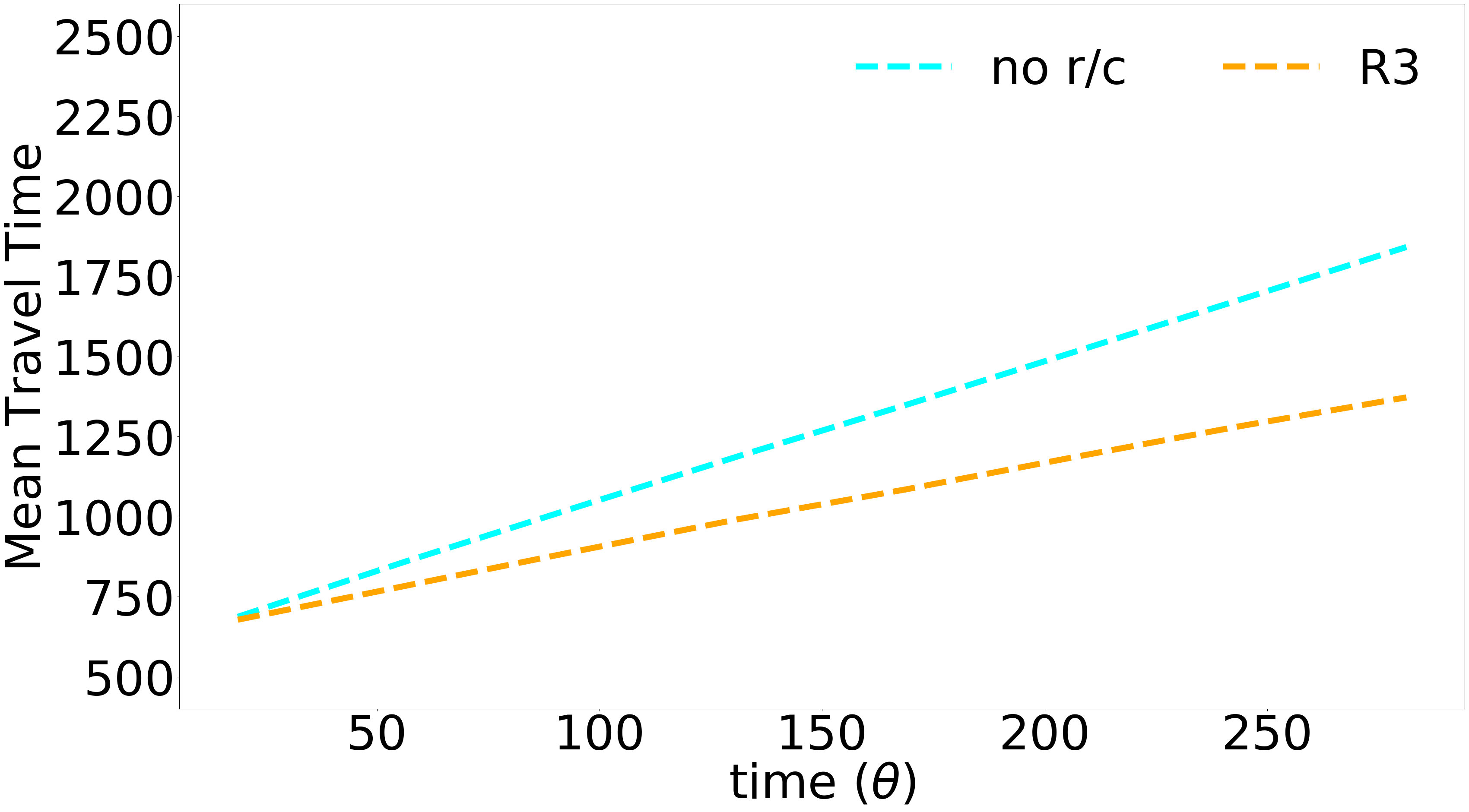}
\caption{The mean travel times for walks with a positive flow for Nguyen-C network
with one (left), two (centre) and three recharging stations (right).}
\label{fig:rnguyen4TimesR123}
\end{figure}
\begin{figure}[h]
\centering
\includegraphics[width=0.49\textwidth]{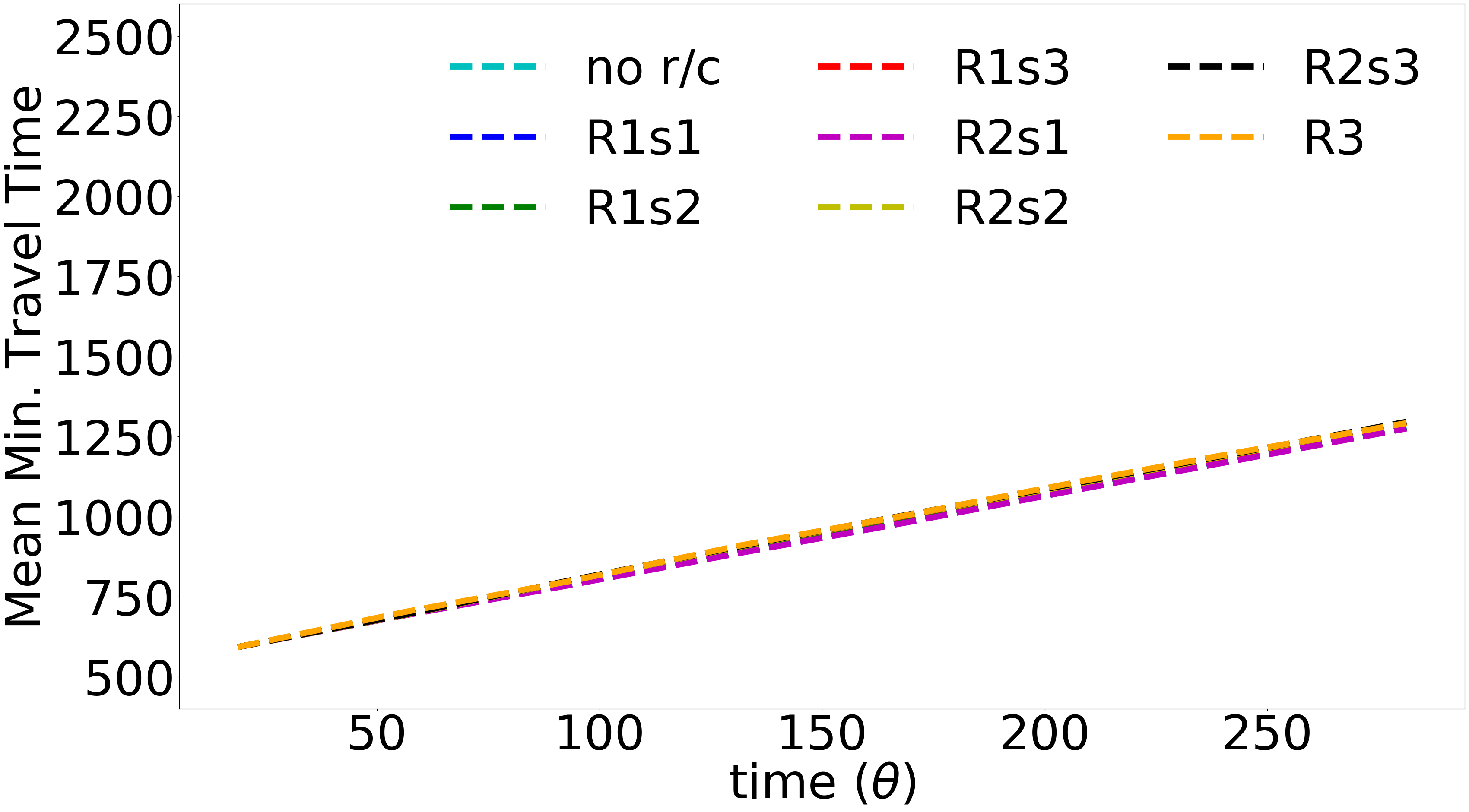}
\includegraphics[width=0.49\textwidth]{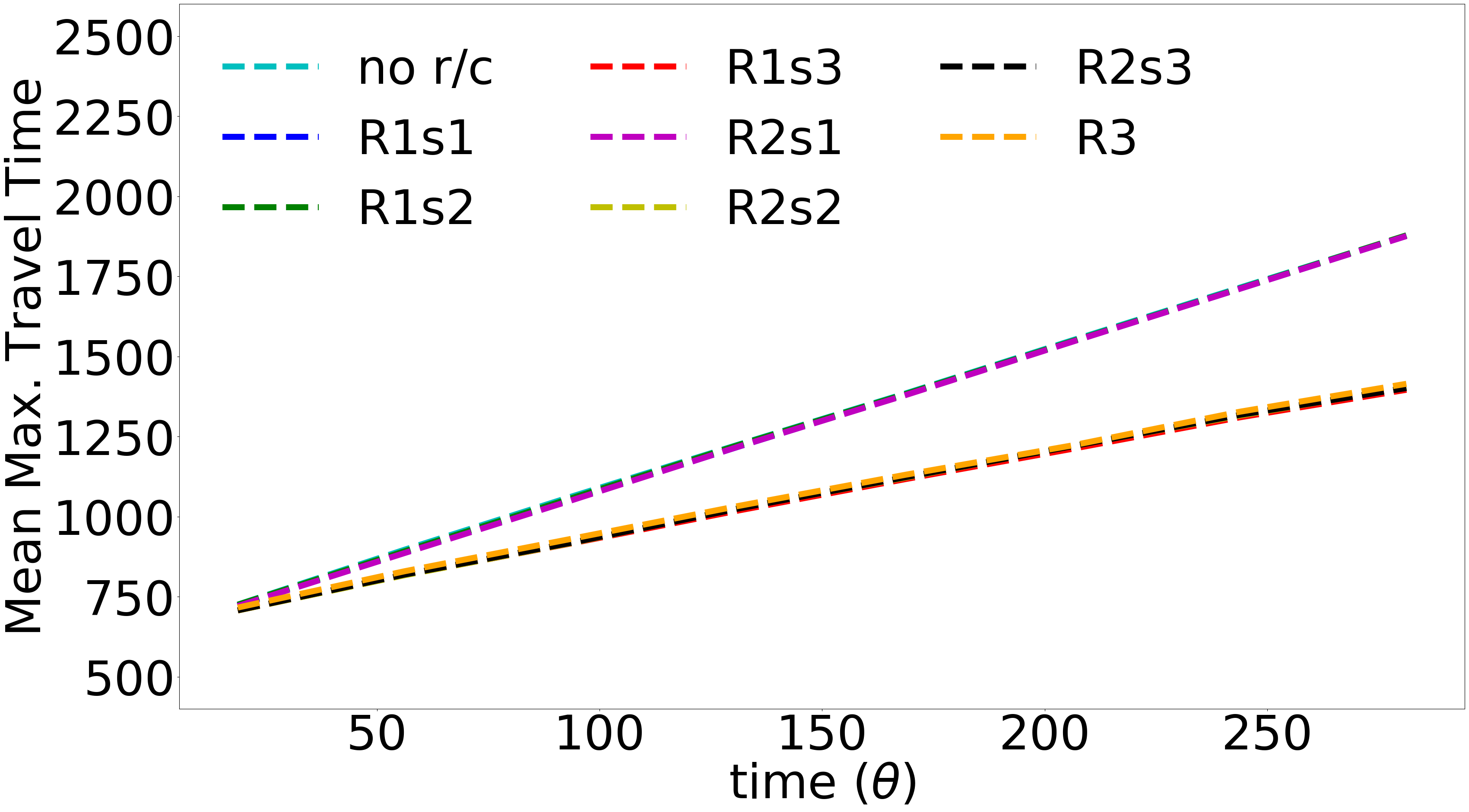}
\caption{The mean of minimum (left) and the mean of the maximum (right) travel times
over all four commodities for walks with a positive flow for the Nguyen-C network
with different combinations of recharging stations at nodes 6, 8 and/or 9.}
\label{fig:rnguyen4Times1}
\end{figure}
\begin{figure}[h]
\centering
\includegraphics[width=0.49\textwidth]{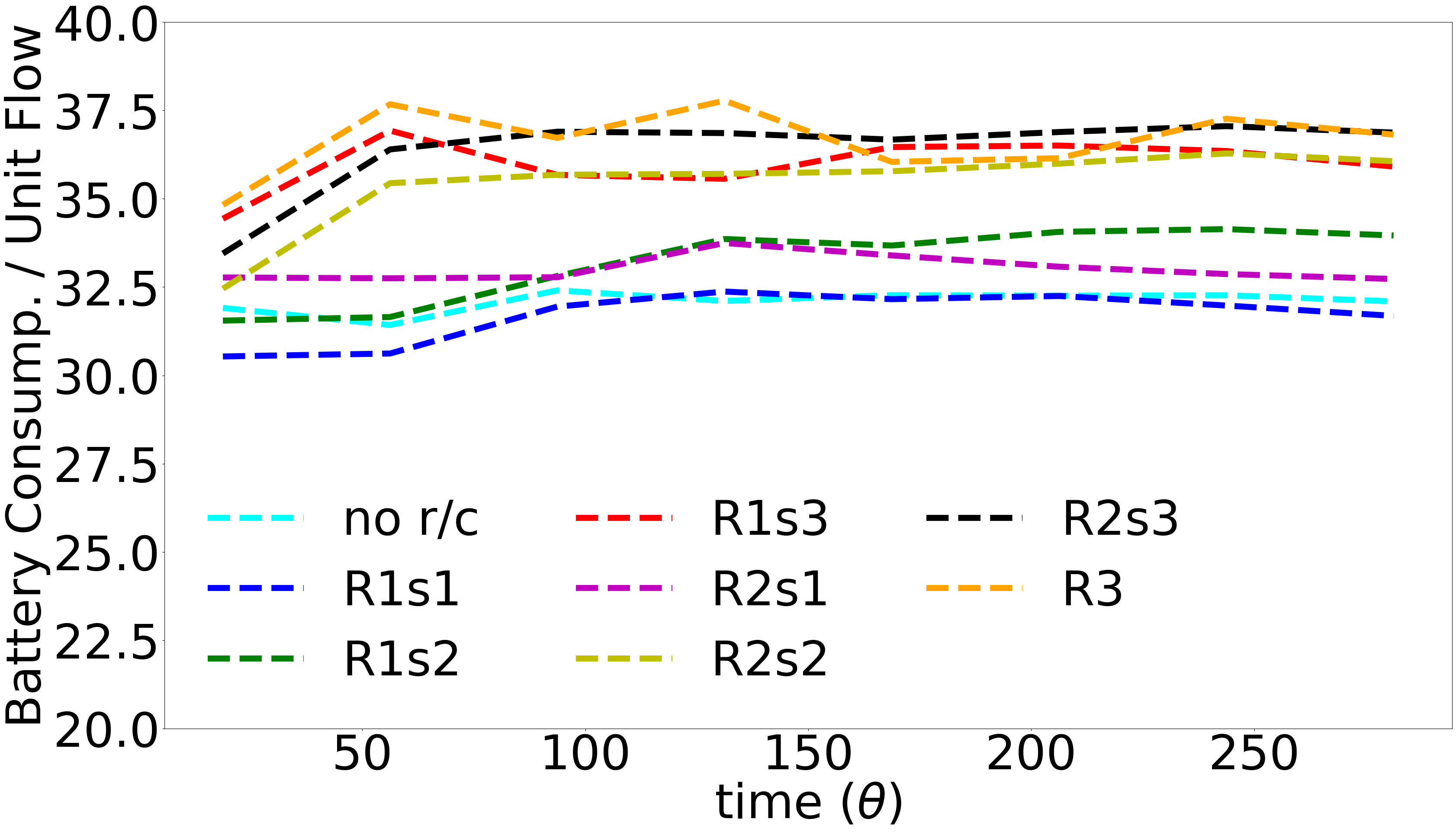}
\includegraphics[width=0.49\textwidth]{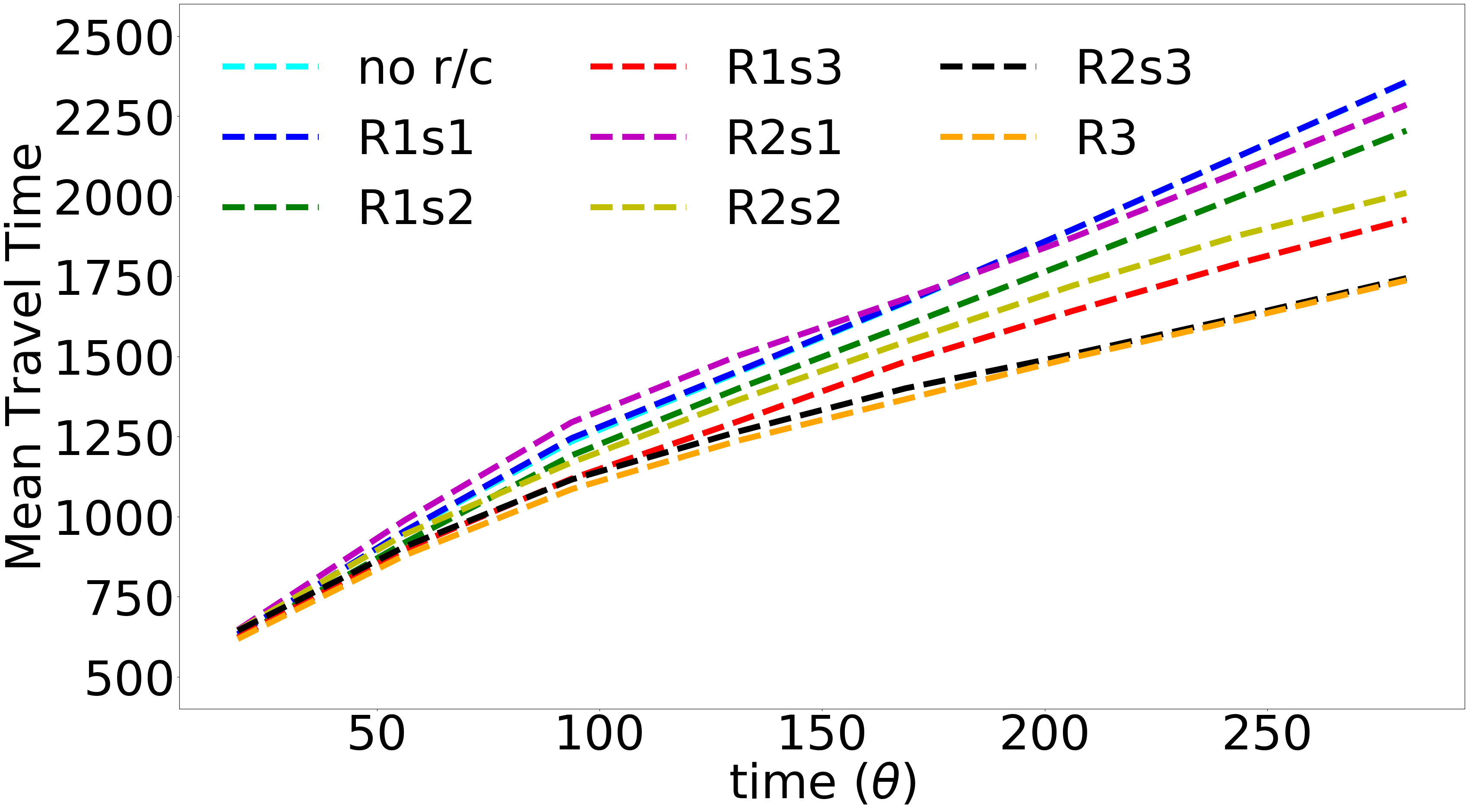}
\caption{The energy consumption profiles (left) and the mean travel times (right)
over eight commodities for walks with a positive flow for the Nguyen-C network
with different combinations of recharging stations at nodes 6, 8 and/or 9.}
\label{fig:rnguyen8}
\end{figure}

\FloatBarrier
\subsubsection{Results for the Sioux Falls Network}

For the Sioux Falls network, the number of walks to be evaluated increases
drastically with the number of recharging stations. Thus, we only include a single recharging station for this network. 
\figref{siouxNetwork} (left) depicts the Sioux Falls-C network with a recharging
station at node $8$. 

\tabref{Sioux} then describes the parameters used in the algorithm and the results for two different battery capacities $b^{max}$. For one of the variants (with energy consumption but without recharging and a battery capacity of $b^{max} = 10$) the algorithm terminated because it reached the desired precision. For three other variants the algorithm terminated after between $10$ and $16$ iteration as it reached the time limit of the runs. This is due to a large fraction of
computational time spent in the network loading and the fixed-point update step
(this is in turn due to the large number of feasible walks in the network).  However, even for those instances the achieved precision was quite small. For one variant (with energy consumption but without recharging and a battery capacity of $b^{max}=6$) no flow could be computed as there was one commodity without an energy feasible walk between source and sink.
\figref{siouxNetwork} (right) shows the convergence trend (in terms of $\Delta h$ and QoPI) for
the Sioux Falls network with recharging and $b^{max}=10$ on a logarithmic scale for the first $10$
iterations. Here, we can again see that most of the progress already happens in the first few iterations after which further progress happens much more slowly.

\begin{table}[htbp]
\caption{Performance of \algoref{fixedPoint} on variants of Sioux network with four
commodities. The following parameters are used: $\epsilon = 0.01$, wall clock time
limit = 7200s, iteration limit = 5000. The abbreviation `NF' indicates that no feasible
walk could be found for at least one of the commodities.}
\scriptsize
\centering
\begin{tabular}{lcccccc}\toprule
 & & \multicolumn{2}{c}{$b^{max} = 6$} & \multicolumn{2}{c}{$b^{max} = 10$} \\\cmidrule(lr){3-4}\cmidrule(lr){5-6}
Variant & A & B & C & B & C \\\cmidrule(lr){2-2}\cmidrule(lr){3-3}\cmidrule(lr){4-4}\cmidrule(lr){5-5}\cmidrule(lr){6-6}
Total \# walks              & 14843             & NF            & 27920           & 90                       & 68674      \\
Wall clock time             & TimeLim           & --            & TimeLim         & 584.771                  & TimeLim    \\
Mean time/DNL               & 85.074            & --            & 257.964         & 0.339                    & 621.421    \\
Mean time/FP-Update         &  8.695            & --            & 20.764          & 0.034                    & 58.764     \\
\# iterations               & 71                & --            & 25              & 1457                     & 10         \\
\# walks with $h>0$         & 10                & --            & 12              & 11                       & 16         \\ 
$\Delta h$                  & 10.745            & --            & 27.236          & 0.100                    & 258.073    \\
$\Delta h$ (relative)       & 0.002             & --            & 0.005           & 0.000                    & 0.045      \\
QoPI (absolute)             & 355.168           & --            & 364.975         & 40.345                   & 1579.097   \\
QoPI                        & 0.007             & --            & 0.003           & 0.001                    & 0.027      \\ \bottomrule
\end{tabular}
\label{tab:Sioux}
\end{table}
\begin{figure}[h!]
	\centering
    \resizebox{.25\linewidth}{!} {\input{tikz/sioux.tex}}
    \includegraphics[width=0.49\textwidth]{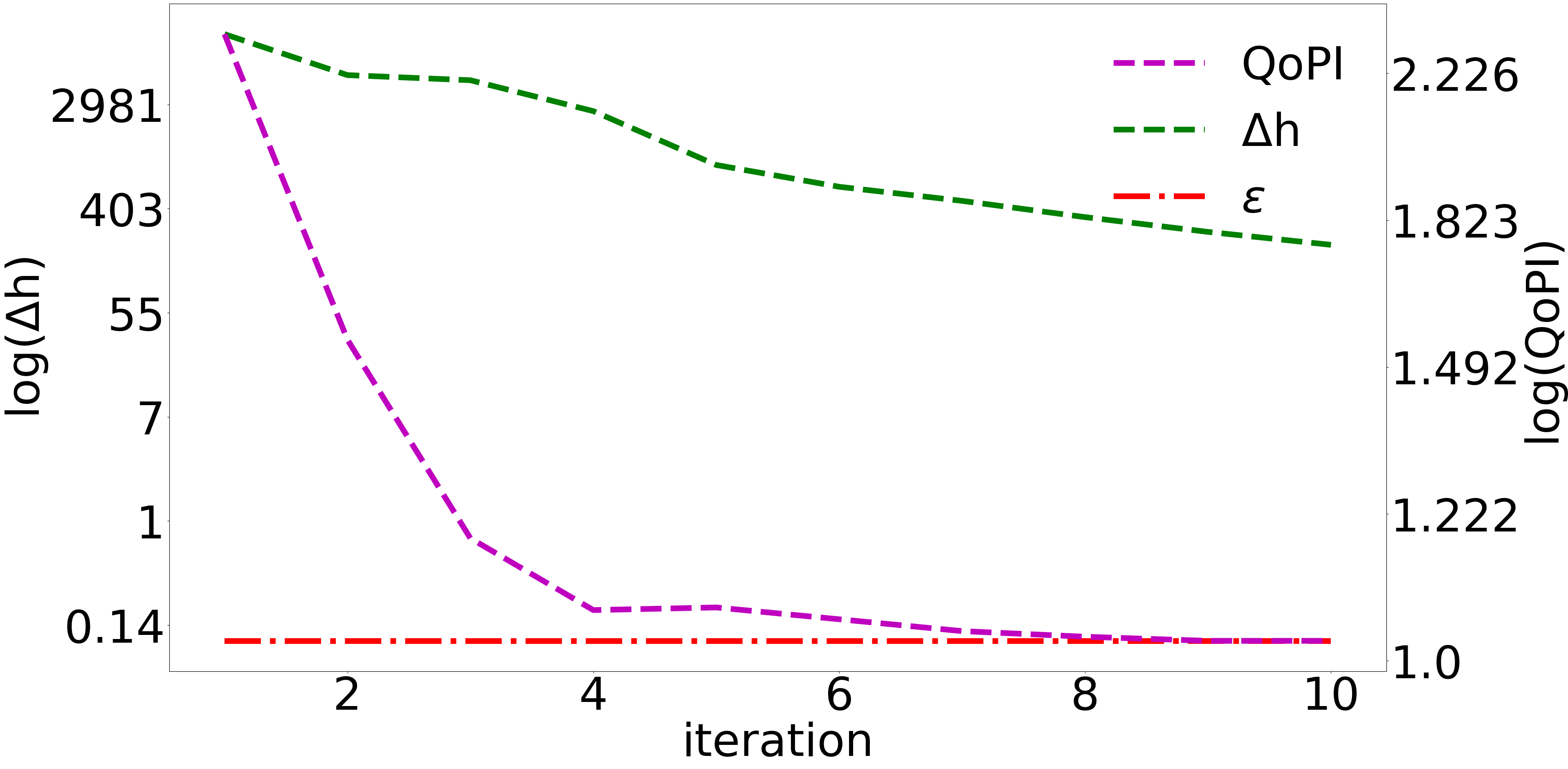}
	\caption{Pictorial depiction of the Sioux Falls network with a recharging station
    at node $8$ (left) and the change in the norm of walk-flows ($\Delta h$) and QoPI
    over iterations for this network with four commodities.}
	\label{fig:siouxNetwork}
\end{figure}

%% file: tikz/example6NBC.tex
\begin{tikzpicture}[node distance={50mm}, thick, main/.style = {draw,
circle, normalEdge},initial text={}] 
\node[initial,main] (s)      {$s$};
\draw[<-] (s) -- node[above] {$u=3$} ++(-2cm,0);
\node[main] (u) [right of=s] {$u$}; 
\node[main] (v) [right of=u, node distance=40mm] {$v$}; 
\node[main] (t) [right of=v] {$t$}; 
\path[->] (s) edge [bend left, normalEdge]  node[near start, above] (e1){$e_1$} node[midway,above] {\color{red} $\tau = 1$} node[midway,below, sloped] {\color{blue} $\nu = 2$} (u);
\path[->] (s) edge [bend right, normalEdge] node[near start, below] (e2){$e_2$} node[midway,above] {\color{red} $\tau = 2$} node[midway,below, sloped] {\color{blue} $\nu = 1$} (u);
\path[->] (u) edge [normalEdge] node[near start, above] (e3){$e_3$} node[midway,above] {\color{red} $\tau = 1$} node[midway,below, sloped] {\color{blue} $\nu = 1$} (v);
\path[->] (v) edge [bend left, normalEdge]  node[near start, above] (e4){$e_4$} node[midway,above] {\color{red} $\tau = 1$} node[midway,below, sloped] {\color{blue} $\nu \geq 1$} (t);
\path[->] (v) edge [bend right, normalEdge] node[near start, below] (e5){$e_5$} node[midway,above] {\color{red} $\tau = 2$} node[midway,below, sloped] {\color{blue} $\nu = 1/2$} (t);
\end{tikzpicture}

%% file: tikz/example6.tex
\begin{tikzpicture}[node distance={50mm}, thick, main/.style = {draw,
circle, normalEdge},initial text={}] 
\node[initial,main] (s)      {$s$};
\draw[<-] (s) -- node[below] {\color{darkgreen} $b^{max}=6$} ++(-2cm,0);
\draw[<-] (s) -- node[above] {$u=3$} ++(-2cm,0);
\node[main] (u) [right of=s] {$u$}; 
\node[main] (v) [right of=u, node distance=40mm] {$v$}; 
\node[main] (t) [right of=v] {$t$}; 
\path[->] (s) edge [bend left, normalEdge]  node[near start, above] (e1){$e_1$} node[above of=e1, node distance=6mm] {\color{darkgreen} $b=4$} node[midway,above] {\color{red} $\tau = 1$} node[midway,below, sloped] {\color{blue} $\nu = 2$} (u);
\path[->] (s) edge [bend right, normalEdge] node[near start, below] (e2){$e_2$} node[above of=e2, node distance=6mm] {\color{darkgreen} $b=2$} node[midway,above] {\color{red} $\tau = 2$} node[midway,below, sloped] {\color{blue} $\nu = 1$} (u);
\path[->] (u) edge [normalEdge] node[near start, above] (e3){$e_3$} node[above of=e3, node distance=6mm] {\color{darkgreen} $b=0$} node[midway,above] {\color{red} $\tau = 1$} node[midway,below, sloped] {\color{blue} $\nu = 1$} (v);
\path[->] (v) edge [bend left, normalEdge]  node[near start, above] (e4){$e_4$} node[above of=e4, node distance=6mm] {\color{darkgreen} $b=4$} node[midway,above] {\color{red} $\tau = 1$} node[midway,below, sloped] {\color{blue} $\nu \geq 1$} (t);
\path[->] (v) edge [bend right, normalEdge] node[near start, below] (e5){$e_5$} node[above of=e5, node distance=6mm] {\color{darkgreen} $b=2$} node[midway,above] {\color{red} $\tau = 2$} node[midway,below, sloped] {\color{blue} $\nu = 1/2$} (t);
\end{tikzpicture}

%% file: tikz/example6WRModes.tex
\begin{tikzpicture}[node distance={50mm}, thick, main/.style = {draw,
circle, normalEdge},initial text={}] 
\node[initial,main] (s)      {$s$};
\draw[<-] (s) -- node[below] (b) {\color{darkgreen} $b^{max}=6$} ++(-2cm,0)
node[above of=b, node distance=6mm] (u0) {$u=3$}
node[below of=b, node distance=4.5mm] (p) {\color{cyan} $p^{max}=6$};

\node[main] (u) [right of=s] {$u$}; 
\node[main] (v) [right of=u, node distance=40mm] {$v$}; 
\node[main] (t) [right of=v] {$t$}; 
\path[->] (s) edge [bend left, normalEdge]  node[near start, above] (e1){$e_1$} node[above of=e1, node distance=6mm] {\color{darkgreen} $b=4$} node[midway,above] {\color{red} $\tau = 1$} node[midway,below, sloped] {\color{blue} $\nu = 2$} (u);
\path[->] (s) edge [bend right, normalEdge] node[near start, below] (e2){$e_2$} node[above of=e2, node distance=6mm] {\color{darkgreen} $b=2$} node[midway,above] {\color{red} $\tau = 2$} node[midway,below, sloped] {\color{blue} $\nu = 1$} (u);
\path[->] (u) edge [normalEdge] node[near start, above] (e3){$e_3$} node[above of=e3, node distance=6mm] {\color{darkgreen} $b=0$} node[midway,above] {\color{red} $\tau = 1$} node[midway,below, sloped] {\color{blue} $\nu = 1$} (v);
\path[->] (v) edge [bend left, normalEdge]  node[near start, above] (e4){$e_4$} node[above of=e4, node distance=6mm] {\color{darkgreen} $b=4$} node[midway,above] {\color{red} $\tau = 1$} node[midway,below, sloped] {\color{blue} $\nu \geq 1$} (t);
\path[->] (v) edge [bend right, normalEdge] node[near start, below] (e5){$e_5$} node[above of=e5, node distance=6mm] {\color{darkgreen} $b=2$} node[midway,above] {\color{red} $\tau = 2$} node[midway,below, sloped] {\color{blue} $\nu = 1/2$} (t);
\path[->,every loop/.style={looseness=12}] (v) edge [in=45,out=135, loop, normalEdge,darkgreen] node[near start, left, black] (m1){$m_1$} node[above of=m1, sloped, node distance=7mm] {\color{darkgreen} $b=b_{W}-b^{max}$} node[above of=m1, node distance=11mm] {\color{cyan} $p = 0$} node[midway,above, sloped] {\color{red} $\tau = 3/2$} node[midway,below, sloped] {\color{blue} $\nu \geq 1$} (v);
\path[->,every loop/.style={looseness=12}] (v) edge [in=-45,out=-135, loop, normalEdge,darkgreen] node[near start, left, black] (m2){$m_2$} node[below of=m2, sloped, node distance=7mm] {\color{darkgreen} $b=b_{W}-b^{max}$} node[below of=m2, node distance=11mm] {\color{cyan} $p = 7$} node[midway,below, sloped] {\color{red} $\tau = 1$} node[midway,above, sloped] {\color{blue} $\nu \geq 1$} (v);
\end{tikzpicture}

%% file: tikz/nguyen.tex
\begin{tikzpicture}[node distance={20mm}, thick, main/.style = {draw,
circle, normalEdge}, initial text={}]
\node[main] (1) {$1$};
\node[main, inner sep=2.5pt] (12) [right of=1] {$12$};
\node[main] (5) [below of=1] {$5$};
\node[main] (9) [below of=5] {$9$};
\node[main] (4) [left of=5] {$4$};
\node[main] (6) [right of=5] {$6$};
\node[main] (7) [right of=6] {$7$};
\node[main] (8) [right of=7] {$8$};
\node[main, inner sep=2.5pt] (10) [right of=9] {$10$};
\node[main, inner sep=2.5pt] (11) [right of=10] {$11$};
\node[main] (2) [right of=11] {$2$};
\node[main, inner sep=2.5pt] (13) [below of=10] {$13$};
\node[main] (3) [right of=13] {$3$};
\path[->] (1) edge [normalEdge] (12);
\path[->] (1) edge [normalEdge] (5);
\path[->] (12) edge [normalEdge] (6);
\path[->] (12) edge [normalEdge] (8);
\path[->] (4) edge [normalEdge] (5);
\path[->] (4) edge [normalEdge] (9);
\path[->] (5) edge [normalEdge] (6);
\path[->] (5) edge [normalEdge] (9);
\path[->] (6) edge [normalEdge] (7);
\path[->] (6) edge [normalEdge] (10);
\path[->] (7) edge [normalEdge] (8);
\path[->] (7) edge [normalEdge] (11);
\path[->] (8) edge [normalEdge] (2);
\path[->] (9) edge [normalEdge] (10);
\path[->] (9) edge [normalEdge] (13);
\path[->] (10) edge [normalEdge] (11);
\path[->] (11) edge [normalEdge] (2);
\path[->] (11) edge [normalEdge] (3);
\path[->] (13) edge [normalEdge] (3);
\path[->,every loop/.style={looseness=10}] (6) edge [in=15,out=75, loop,normalEdge,darkgreen] (6);
\path[->,every loop/.style={looseness=10}] (8) edge [in=15,out=75, loop,normalEdge,darkgreen] (8);
\path[->,every loop/.style={looseness=10}] (9) edge [in=15,out=75, loop,normalEdge,darkgreen] (9);
\end{tikzpicture}

%% file: tikz/sioux.tex
\begin{tikzpicture}[node distance={20mm}, thick, main/.style = {draw,
circle, normalEdge}, initial text={}]
\node[main] (1) {$1$};
\node[main] (3) [below of=1] {$3$};
\node[main] (4) [right of=3] {$4$};
\node[main] (5) [right of=4] {$5$};
\node[main] (6) [right of=5] {$6$};
\node[main] (2) [above of=6] {$2$};
\node[main] (9) [below of = 5] {$9$};
\node[main] (8) [right of=9] {$8$};
\node[main] (7) [right = 2cm of 8] {$7$};
\node[main, inner sep=2.5pt] (12) [below = 3cm of 3] {$12$};
\node[main, inner sep=2.5pt] (11) [right of=12] {$11$};
\node[main, inner sep=2.5pt] (10) [right of=11] {$10$};
\node[main, inner sep=2.5pt] (16) [right of=10] {$16$};
\node[main, inner sep=2.5pt] (18) [right = 2cm of 16] {$18$};
\node[main, inner sep=2.5pt] (17) [below of= 16] {$17$};
\node[main, inner sep=2.5pt] (19) [below of= 17] {$19$};
\node[main, inner sep=2.5pt] (15) [left of=19] {$15$};
\node[main, inner sep=2.5pt] (14) [left of=15] {$14$};
\node[main, inner sep=2.5pt] (23) [below of=14] {$23$};
\node[main, inner sep=2.5pt] (22) [right of=23] {$22$};
\node[main, inner sep=2.5pt] (21) [below of=22] {$21$};
\node[main, inner sep=2.5pt] (20) [right of=21] {$20$};
\node[main, inner sep=2.5pt] (24) [left of=21] {$24$};
\node[main, inner sep=2.5pt] (13) [left of=24] {$13$};
\path[->] (1.north east) edge [normalEdge] (2.north west);
\path[->] (2.south west) edge [normalEdge] (1.south east);
\path[->] (3.north east) edge [normalEdge] (4.north west);
\path[->] (4.south west) edge [normalEdge] (3.south east);
\path[->] (4.north east) edge [normalEdge] (5.north west);
\path[->] (5.south west) edge [normalEdge] (4.south east);
\path[->] (5.north east) edge [normalEdge] (6.north west);
\path[->] (6.south west) edge [normalEdge] (5.south east);
\path[->] (9.north east) edge [normalEdge] (8.north west);
\path[->] (8.south west) edge [normalEdge] (9.south east);
\path[->] (8.north east) edge [normalEdge] (7.north west);
\path[->] (7.south west) edge [normalEdge] (8.south east);
\path[->] (12.north east) edge [normalEdge] (11.north west);
\path[->] (11.south west) edge [normalEdge] (12.south east);
\path[->] (11.north east) edge [normalEdge] (10.north west);
\path[->] (10.south west) edge [normalEdge] (11.south east);
\path[->] (10.north east) edge [normalEdge] (16.north west);
\path[->] (16.south west) edge [normalEdge] (10.south east);
\path[->] (16.north east) edge [normalEdge] (18.north west);
\path[->] (18.south west) edge [normalEdge] (16.south east);
\path[->] (14.north east) edge [normalEdge] (15.north west);
\path[->] (15.south west) edge [normalEdge] (14.south east);
\path[->] (15.north east) edge [normalEdge] (19.north west);
\path[->] (19.south west) edge [normalEdge] (15.south east);
\path[->] (23.north east) edge [normalEdge] (22.north west);
\path[->] (22.south west) edge [normalEdge] (23.south east);
\path[->] (13.north east) edge [normalEdge] (24.north west);
\path[->] (24.south west) edge [normalEdge] (13.south east);
\path[->] (24.north east) edge [normalEdge] (21.north west);
\path[->] (21.south west) edge [normalEdge] (24.south east);
\path[->] (21.north east) edge [normalEdge] (20.north west);
\path[->] (20.south west) edge [normalEdge] (21.south east);
\path[->] (1.south west) edge [normalEdge] (3.north west);
\path[->] (3.north east) edge [normalEdge] (1.south east);
\path[->] (3.south west) edge [normalEdge] (12.north west);
\path[->] (12.north east) edge [normalEdge] (3.south east);
\path[->] (12.south west) edge [normalEdge] (13.north west);
\path[->] (13.north east) edge [normalEdge] (12.south east);
\path[->] (4.south west) edge [normalEdge] (11.north west);
\path[->] (11.north east) edge [normalEdge] (4.south east);
\path[->] (11.south west) edge [normalEdge] (14.north west);
\path[->] (14.north east) edge [normalEdge] (11.south east);
\path[->] (14.south west) edge [normalEdge] (23.north west);
\path[->] (23.north east) edge [normalEdge] (14.south east);
\path[->] (23.south west) edge [normalEdge] (24.north west);
\path[->] (24.north east) edge [normalEdge] (23.south east);
\path[->] (5.south west) edge [normalEdge] (9.north west);
\path[->] (9.north east) edge [normalEdge] (5.south east);
\path[->] (9.south west) edge [normalEdge] (10.north west);
\path[->] (10.north east) edge [normalEdge] (9.south east);
\path[->] (10.south west) edge [normalEdge] (15.north west);
\path[->] (15.north east) edge [normalEdge] (10.south east);
\path[->] (15.south west) edge [normalEdge] (22.north west);
\path[->] (22.north east) edge [normalEdge] (15.south east);
\path[->] (22.south west) edge [normalEdge] (21.north west);
\path[->] (21.north east) edge [normalEdge] (22.south east);
\path[->] (2.south west) edge [normalEdge] (6.north west);
\path[->] (6.north east) edge [normalEdge] (2.south east);
\path[->] (6.south west) edge [normalEdge] (8.north west);
\path[->] (8.north east) edge [normalEdge] (6.south east);
\path[->] (8.south west) edge [normalEdge] (16.north west);
\path[->] (16.north east) edge [normalEdge] (8.south east);
\path[->] (16.south west) edge [normalEdge] (17.north west);
\path[->] (17.north east) edge [normalEdge] (16.south east);
\path[->] (17.south west) edge [normalEdge] (19.north west);
\path[->] (19.north east) edge [normalEdge] (17.south east);
\path[->] (19.south west) edge [normalEdge] (20.north west);
\path[->] (20.north east) edge [normalEdge] (19.south east);
\path[->] (7.south west) edge [normalEdge] (18.north west);
\path[->] (18.north east) edge [normalEdge] (7.south east);
\path[->] (18.south west) edge [normalEdge] (20.north east);
\path[->] (20.east)+(.6mm,0) edge [normalEdge] (18.south);
\path[->] (10.south east)+(0,-2mm) edge [normalEdge] (17.west);
\path[->] (17.north west) edge [normalEdge] +(-1.4cm,1.5cm);
\path[->] (22.south east)+(0,-2mm) edge [normalEdge] +(1.2cm,-1.5cm);
\path[->] (20.north west) edge [normalEdge] +(-1.4cm,1.5cm);
\path[->,every loop/.style={looseness=10}] (8) edge [in=15,out=75, loop,normalEdge,darkgreen] (8);
\end{tikzpicture}

%% file: chapters/Conclusions.tex
\section{Conclusions}
In this paper, we introduced and analyzed a DTA model 
for the operation of electrical vehicles.
The model combines the Vickrey deterministic queueing model with graph-based gadgets modeling  recharging operations:  a combined routing and recharging strategy of an EV can be reduced to choosing a  walk 
(possibly containing cycles) within this gadget-extended network. 
As our main theoretical result, we proved the existence of dynamic equilibria in this DTA model.
We further  discretized the model in order to apply a fixed-point algorithm
computing approximate dynamic equilibria. The fixed-point algorithm 
is designed to balance the resulting effective cost among all battery-feasible walks. We applied this algorithm to instances from the literature
demonstrating the effect of battery-constraints and recharging options
to the resulting equilibrium travel times and energy consumption profiles, respectively. As the placement and operation of recharging infrastructure
is a key challenge for the whole development of EV-mobility, we believe
that our model and algorithmic approach can serve as the basis for
predicting the resulting effects of such infrastructure designs.

There are several open problems and challenges that
are still widely open. Our approach is inherently walk-based
and therefore it has limitations when it comes to
computing dynamic equilibria for larger instances.
For the Sioux Falls network, for $b^{max}=10$ and two recharging stations, the number of walks to be
evaluated exceeds $10^8$ which leads to unreasonably high
computational times for the network loading routine. We  observed, however,  that
at termination of the fixed-point algorithm, the number of walks with a positive flow is a small fraction of the
total number of walks.
We believe that further insights are required for designing a better algorithm to
compute approximate capacitated dynamic equilibria for larger instances.

%% file: chapters/appendix.tex
\appendix
\section{Technical Lemmas}

\begin{lemma}\label{lemma:PointwiseConvergenceOfConcatOfUniformConvergence}
	Let $A,B \subseteq \IR$ be two subset of real numbers, $G^k: A \to B$ a sequence of functions converging uniformly to some function $G:A \to B$ and $F^k: B \to \IR$ another sequence of functions converging uniformly to some continuous function $F:B \to \IR$.
	
	Then $F^k\circ G^k: A \to B$ is a sequence of functions converging point-wise to $F\circ G: A \to B$.
\end{lemma}

\begin{proof}
	Take any $\theta \in A$ and $\varepsilon > 0$. Since $F$ is continuous, there exists some $\delta > 0$ such that for every $\theta' \in A$ with $\abs{\theta - \theta'} \leq \delta$ we have $\abs{F(\theta)-F(\theta')} \leq \nicefrac{\varepsilon}{2}$. Furthermore, from the uniform convergence of $G^k$ and $F^k$ we get the existence of some $K \in \IN$ such that for every $k \geq K$ we have $\norm{G-G^k}_\infty \leq \delta$ and $\norm{F-F^k}_\infty \leq \nicefrac{\varepsilon}{2}$. Together this gives us
	\begin{align*}
		&\abs{F\circ G(\theta) - F^k\circ G^k(\theta)} 
			\leq \abs{F\circ G(\theta) - F\circ G^k(\theta)} + \abs{F\circ G^k(\theta) - F^k\circ G^k(\theta)} \\
		&\quad\leq \nicefrac{\varepsilon}{2} + \nicefrac{\varepsilon}{2} = \varepsilon.
	\qedhere
	\end{align*}
\end{proof}

\begin{lemma}\label{lemma:UniformConvergenceOfPointwiseConvergence}
	Let $f^k$ be a sequence of functions from some compact interval $[a,b]$ to $\IR$ that converges point-wise to some function $f$. If all $f^k$ are Lipschitz-continuous with some common Lipschitz constant $L$ and $f$ is continuous then $f^k$ converges uniformly to $f$.
\end{lemma}

\begin{proof}
	Let $\varepsilon > 0$. Since $f$ is continuous on a compact set, it is uniformly continuous. Thus, there exists some $\delta > 0$ such that for any two points $x,y \in [a,b]$ with $\abs{x-y} \leq \delta$ we have $\abs{f(x)-f(y)} \leq \nicefrac{\varepsilon}{3}$. Now, fix some partition of $[a,b]$ into intervals $[x_0,x_1), [x_1,x_2), \dots, [x_{N-1},x_N]$ with length at most $\min\set{\nicefrac{\varepsilon}{3L},\delta}$ each. Then choose $K \in \IN$ such that for all $k \geq K$ and all $i \in \{0,1, \dots, N\}$ we have $\abs{f^k(x_i) - f(x_i)} \leq \nicefrac{\varepsilon}{3}$. Then we have for every $k \geq K$ and any $x \in [a,b]$ that there exists some $i$ with $\abs{x_i-x} \leq \min\set{\nicefrac{\varepsilon}{3L},\delta}$ and, thus, we have
	\begin{align*}
		\abs{f^k(x) - f(x)} 
			&\leq \abs{f^k(x)-f^k(x_i)} + \abs{f^k(x_i)-f(x_i)} + \abs{f(x_i)-f(x)} \\
			&\leq L\cdot\abs{x-x_i} + \frac{\varepsilon}{3} + \frac{\varepsilon}{3} \leq \frac{3\varepsilon}{3} = \varepsilon.
	\end{align*}
	As this holds for all $x \in [a,b]$ (with the same $K$), we have shown that $f^k$ converges uniformly to $f$.
\end{proof}

\begin{lemma}\label{lemma:covering}
	Given a finite interval $[a,b] \subseteq \IR$, a subset $J \subseteq [a,b]$ of positive measure and for every $\theta \in J$ some positive number $\delta_\theta > 0$. Then there exists some $\theta \in J$ such that $J \cap [\theta,\theta+\delta_\theta]$ has positive measure.
\end{lemma}

\begin{proof}
	For any $n \in \IN$ we define the set of all points in $J$ which are the border of some interval of size at least $1/n$ but not in the interior of any interval by
		\[\partial J^n \coloneqq \Set{\bar\theta \in J | \exists \theta \in J: \delta_\theta \geq \frac{1}{n}, \bar\theta \in \set{\theta,\theta+\delta_\theta} \text{ and } \forall \theta \in J: \bar\theta \notin (\theta,\theta+\delta_\theta)}\]
	and claim that there is at most a countable number of such points. Indeed, note that any point in $\partial J^n$ must be the border point of an interval of size at least $1/n$ not containing any other such point. Thus, for any point in $\partial J^n$ there can be at most one other point from $\partial J^n$ within a distance of less than $1/n$. This implies that there are only countably many such points.
	
	But then the set
		\[\partial J \coloneqq \Set{\bar\theta \in J | \forall \theta \in J: \bar\theta \notin (\theta,\theta+\delta_\theta)} = \bigcup_{n \in \IN}\partial J^n\]
	is also countable and, thus, has measure zero. This, in turn, implies that $J \setminus \partial J$ still has positive measure. 
	
	Now, since the Lebesgue measure is inner regular, $J \setminus \partial J$ contains a compact subset $C$ of positive measure. As the family of open intervals $\set{(\theta,\theta+\delta_\theta) | \theta \in J}$ is an open covering of $J \setminus \partial J$, it is also a covering of $C$ and, thus, contains a finite subcover. Consequently, at least one of the open intervals $(\theta,\theta+\delta_\theta)$ of this subcover must have a  intersection of positive measure with $C$. This $\theta$ then also satisfies the desired property of the lemma.
\end{proof}